\documentclass[
 aps,
 prx,
 reprint,
 superscriptaddress,
 longbibliography,
 nofootinbib,
 amsmath,amssymb,
]{revtex4-2}

\usepackage{amsthm,mathtools,bm}
\usepackage{graphicx}
\usepackage{xcolor}
\usepackage{algorithm,algpseudocode}
\AtBeginDocument{\newfloat{algorithm}{tbp}{loa}\floatname{algorithm}{Algorithm~}}
\algrenewcommand{\algorithmicrequire}{\textbf{Input:}}
\algrenewcommand{\algorithmicensure}{\textbf{Output:}}
\usepackage{tikz}
\usetikzlibrary{arrows.meta,positioning,calc,fit,shapes.geometric}
\usepackage{pgfplots}
\pgfplotsset{compat=1.17}
\usepackage[colorlinks=true,linkcolor=blue,citecolor=blue,urlcolor=blue]{hyperref}
\usepackage[nameinlink,capitalise]{cleveref}
\crefname{figure}{Fig.}{Figs.}
\Crefname{figure}{Fig.}{Figs.}
\crefname{equation}{Eq.}{Eqs.}
\Crefname{equation}{Eq.}{Eqs.}
\crefname{section}{Sec.}{Secs.}
\Crefname{section}{Sec.}{Secs.}
\crefname{subsection}{Sec.}{Secs.}
\Crefname{subsection}{Sec.}{Secs.}
\crefname{subsubsection}{Sec.}{Secs.}
\Crefname{subsubsection}{Sec.}{Secs.}

\makeatletter
\def\theHALG@line{\thealgorithm.\arabic{ALG@line}}
\makeatother

\newtheorem{theorem}{Theorem}
\newtheorem{lemma}{Lemma}
\newtheorem{proposition}{Proposition}
\newtheorem{corollary}{Corollary}
\newtheorem{fact}{Fact}
\crefname{fact}{Fact}{Facts}
\Crefname{fact}{Fact}{Facts}
\theoremstyle{definition}
\newtheorem{definition}{Definition}
\theoremstyle{remark}
\newtheorem{remark}{Remark}

\newcommand{\Ot}{\widetilde{O}}
\newcommand{\Omt}{\widetilde{\Omega}}

\newcommand{\C}{\mathbb{C}}
\newcommand{\R}{\mathbb{R}}
\newcommand{\E}{\mathbb{E}}
\newcommand{\norm}[1]{\lVert #1\rVert}
\newcommand{\abs}[1]{\lvert #1\rvert}
\newcommand{\ket}[1]{\lvert #1\rangle}
\newcommand{\bra}[1]{\langle #1\rvert}
\newcommand{\tr}{\operatorname{tr}}
\newcommand{\diag}{\operatorname{diag}}
\newcommand{\polylog}{\operatorname{polylog}}
\newcommand{\poly}{\operatorname{poly}}
\newcommand{\SEL}{\textsc{select}}

\usetikzlibrary{quantikz2}

\begin{document}
\title{$T$-efficient unitary synthesis via block-sparse Hamiltonian simulation}

\author{Diyi Liu}
\affiliation{Lawrence Berkeley National Laboratory,
Berkeley, California 94720, USA}

\author{Shuchen Zhu}
\affiliation{Microsoft Quantum, Redmond, Washington 98052, USA}
\affiliation{Department of Mathematics, Duke University, Durham, North Carolina 27708, USA}
\affiliation{Duke Quantum Center, Duke University, Durham, North Carolina 27701, USA}

\author{Wibe Albert de Jong}
\affiliation{Lawrence Berkeley National Laboratory,
Berkeley, California 94720, USA}

\author{Di Fang}
\affiliation{Department of Mathematics, Duke University, Durham, North Carolina 27708, USA}
\affiliation{Duke Quantum Center, Duke University, Durham, North Carolina 27701, USA}

\author{David Williams-Young}
\affiliation{Microsoft Quantum, Redmond, Washington 98052, USA}

\author{Chao Yang}
\affiliation{Lawrence Berkeley National Laboratory,
Berkeley, California 94720, USA}

\date{\today}

\begin{abstract}
We construct Clifford+$T$ circuits implementing arbitrary $n$-qubit
unitaries with worst-case $T$-count $\Ot(N^{8/7})$ at constant
operator-norm error, where $N=2^n$, improving the previous exponent
$5/4$. The construction uses $\Ot(N^{6/7})$ clean ancillas.
The main idea is to reduce the data addressed by each query
to the block-encoding circuit for a sparse piece
while controlling the additional simulation work needed to combine
the pieces. We simulate a Hermitian logarithm of the target by
flattening its traceless part with a two-layer random similarity transformation,
randomly partitioning the transformed generator into block-sparse
Hermitian pieces, and combining their evolutions using qDRIFT.
For a target operator-norm error $\epsilon\in(0,1/2)$, let
$L=\log_2(N/\epsilon)$, the qubit count plus the requested bits of
accuracy. By successively correcting the approximation error, we
obtain a circuit using
$\Ot(N^{8/7}L^{3/7}+NL)$ $T$ gates and $\Ot(N)$ clean ancillas
whenever $L\le N$.
We also prove that, at constant precision, the exponent $8/7$ is
optimal within a sequential block-SELECT model.
\end{abstract}

\maketitle

\section{Introduction}
\label{sec:intro}

In fault tolerant quantum computing, non-Clifford gates are often substantially
more expensive than Clifford gates. For example, implementations using
surface codes~\cite{fowler2012surface} can supply the gate
$T=\diag(1,e^{i\pi/4})$ through magic state
distillation~\cite{bravyi2005universal}. This motivates using the number of
$T$ gates as a resource measure, distinct from the total gate count, and has led to
extensive optimization of small
circuits~\cite{kliuchnikov2013fast,ross2016optimal,gheorghiu2022tcount}.
Here we ask: \emph{how many $T$ gates
are needed to implement an arbitrary unitary on $n$ qubits to accuracy
$\epsilon$?}

For $N=2^n$, a generic unitary has $N^2$ real parameters.
The challenge is to incorporate these parameters into a Clifford+$T$
circuit using few $T$ gates. At constant precision, a polynomial gap
remains between the known bounds on the number of $T$ gates required
in the worst case. Counting gives
a lower bound of
$\Omega(N)$~\cite{gosset2024quantum,tan2025unitary}.
Low et al.~\cite{Low2024tradingtgatesdirty} give
an upper bound $\Ot(N^{3/2})$. Tan~\cite{tan2025unitary} reduces the
exponent to $4/3$ by combining recursive cosine--sine decomposition
with the synthesis of multiplexed diagonal unitaries. Yuan et al.~\cite{yuan2026general}
obtain exponent $5/4$ through random block
flattening, dilation, and singular value transformation.
Fang et al.~\cite{fang2026nearclifford} use a different approach based
on a logarithm to obtain bounds that depend on the target unitary,
with $\Ot(N^{3/2})$ $T$ gates in the worst case.

Related work studies depth and ancilla tradeoffs using continuously
parameterized gates~\cite{sun2023depth,rosenthal2026query,yuan2023optimal,nehoran2026constant,yuan2026exact,strelchuk2026constant,cao2026measurements}.
Oracle formulations count Boolean queries as primitive
operations~\cite{aaronson2007quantum,rosenthal2026query,lombardi2024onequery,dong2026separations,dong2026preprocessing,brakerski2026scalable,huang2026random}.
A recent construction gives synthesis in polynomial time at constant
error in the diamond norm, using a suitably chosen oracle that depends
on the target~\cite{openai2026unitary}.
Here we count the $T$ gates needed to implement the complete circuit,
including the synthesis of elementary gates and the implementation
of data access.

We obtain the following bound by combining encodings of block-sparse matrices with
Hamiltonian simulation and residual refinement.

\begin{theorem}[Unitary synthesis, informal version of \Cref{thm:constant,thm:refined}]
\label{thm:main-informal}
Let $N=2^n$ and $L=n+\log_2(1/\epsilon)$, with
$0<\epsilon<1/2$ and $L\le N$.
For sufficiently large $n$, every $U\in U(N)$ can be implemented to
error at most $\epsilon$ in operator norm by a Clifford+$T$ circuit using
\begin{equation}
\Ot\!\left(N^{8/7}L^{3/7}+NL\right)
\label{eq:main-bound}
\end{equation}
$T$ gates and $\Ot(N)$ clean ancillas.
At constant precision, a separate choice of parameters gives
$\Ot(N^{8/7})$ $T$ gates with $\Ot(N^{6/7})$ ancillas.
\end{theorem}

\begin{table*}[t]
\caption{\label{tab:comparison}Number of $T$ gates required in the worst case
to implement an arbitrary unitary on $n$ qubits with Clifford+$T$ circuits,
where $N=2^n$ and
$L=n+\log_2(1/\epsilon)$. Here, fixed $\epsilon$ means a
sufficiently small positive error tolerance independent of $N$.
\Cref{thm:constant} gives a smaller ancilla
bound at fixed error, while \Cref{thm:refined} improves the dependence
on accuracy through residual refinement. Ancilla counts follow each
reference's conventions. Our constructions use clean ancillas, while
Ref.~\cite{Low2024tradingtgatesdirty} also allows dirty workspace.
The lower bound in \Cref{thm:optimality} concerns the cost in the sequential
block-SELECT model under the stated limit on information stored in the
unitaries between evolutions.}
\begin{ruledtabular}
\begin{tabular}{llll}
Reference & $T$ gates & Ancillas & Validity\\
\hline
Low, Kliuchnikov, and Schaeffer~\cite{Low2024tradingtgatesdirty} & $\Ot(N^{3/2})$ & $\Ot(N^{1/2})$ & fixed $\epsilon$\\
Tan~\cite{tan2025unitary} & $O(N^{4/3}L^{2/3}+NL)$ & $O(N^{2/3}L^{1/3}+L)$ & all $\epsilon$\\
Yuan, Zhang, and Zi~\cite{yuan2026general} & $O(N^{5/4}L^{5/8}\log N)$ & $O(N\sqrt L)$ & $L\le N$\\
This work, \Cref{thm:constant} & $\Ot(N^{8/7}\epsilon^{-6/7})$ & $\Ot(N^{6/7}\epsilon^{-8/7})$ & $\epsilon\ge N^{-1/8}\polylog N$\\
This work, \Cref{thm:refined} & $\Ot(N^{8/7}L^{3/7}+NL)$ & $\Ot(N)$ & $L\le N$\\
Lower bound~\cite{gosset2024quantum,tan2025unitary} & $\Omega(N)$ & --- & fixed $\epsilon$\\
This work, \Cref{thm:optimality} & $\Omt(N^{8/7})$ & --- & fixed $\epsilon$\\
\end{tabular}
\end{ruledtabular}
\end{table*}

The improvement shown in \Cref{tab:comparison} comes from encoding
block-sparse pieces separately, so each query accesses a smaller lookup
table. %
In the construction of Yuan et
al.~\cite{yuan2026general}, each query addresses a lookup table containing
data for all $N^2$ matrix entries. Singular value amplification uses
this encoding repeatedly, so the lookup cost is paid at every query.
We instead use separate encodings with smaller lookup tables. To do
this, we represent $U$ as evolution under a Hermitian logarithm and
partition its generator into block-sparse Hermitian pieces.
We construct these pieces using a random change of basis followed by
a random partition that preserves Hermiticity.

To approximate evolution under the sum of these pieces, we use
qDRIFT~\cite{campbell2019random,chen2021concentration}, which applies
short evolutions under randomly chosen pieces. More pieces give
smaller lookup tables, but can require more such evolutions to
achieve the same accuracy. We choose the block size and number of
pieces to balance these costs, using
$\Ot(N^{8/7})$ $T$ gates at constant precision.

For higher precision, we successively correct the remaining
approximation error using residual refinement. Each correction involves
a smaller generator, avoiding the costly precision dependence of a
single qDRIFT approximation.
All sampling is performed during classical preprocessing, and the sampled evolutions are compiled into a Clifford+$T$ circuit.

We also prove a matching lower bound in a restricted cost model that
includes our construction. The model represents the target as a
sequence of Hamiltonian evolutions, charging for matrix lookups,
operations within each block, and simulation. It limits how much
information about the target can be stored in the unitaries between
evolutions. Within this model, called the sequential block-SELECT
model, the exponent $8/7$ is optimal in the worst case at constant precision,
up to factors polynomial in $n$.

\begin{theorem}[Optimality within the model, informal version of \Cref{thm:optimality}]
\label{thm:optimality-informal}
At sufficiently small constant error in operator norm, the cost of
unitary synthesis in the sequential block-SELECT model is
$\Omt(N^{8/7})$ in the worst case.
\end{theorem}

The rest of the paper is organized as follows.
\Cref{sec:prelim} introduces the notation and error criterion and
summarizes the tools for data loading and Hamiltonian simulation.
\Cref{sec:overview} illustrates the construction and its cost balance.
\Cref{sec:pieces} develops the
flattening, partitioning, and encoding steps.
\Cref{sec:constant} bounds the resources for unitary synthesis
using qDRIFT and QSVT, and \Cref{sec:refinement} improves the
precision dependence through residual refinement.
\Cref{sec:optimality} states the lower bound in the sequential
block-SELECT model. These last three sections include proof sketches.
\Cref{sec:discussion} concludes the work.
The appendices contain the complete proofs and an extension of the
lower bound to mixed block sizes.

\section{Preliminaries}
\label{sec:prelim}

We first set the notation and error criterion. The remaining
subsections state the data-loading and simulation tools needed.

\subsection{Notation}

Let $n$ be the number of data qubits and $N=2^n$. We order the
computational basis lexicographically by bit strings $x\in\{0,1\}^n$.
For a positive integer $r$, write $[r]=\{1,\dots,r\}$.

The construction groups matrix entries into square blocks. The
block size balances the cost of loading the data with the cost of
implementing the quantum operations.
Let $b$ be the number of rows and columns in each square block,
chosen as a power of two with $b\le N$. Then $D=N/b$ is the number
of block rows and of block columns. For a matrix $A\in\C^{N\times N}$,
group its rows and columns consecutively in sets of $b$. The
\emph{block} $A_{IJ}$ is the $b\times b$ submatrix in block row $I$
and block column $J$, with $I,J\in[D]$.

For a row index $x$, write $x=(x_{\mathrm h},x_{\mathrm l})$. Here
$x_{\mathrm h}$ identifies the block row and $x_{\mathrm l}$ the
row within it. They consist of the first $\log_2D$ and last
$\log_2b$ bits, respectively. The same convention applies to column
indices.
\Cref{fig:block-notation} illustrates the partition and its bit labels.

\begin{figure}[htb]
\centering
\begin{tikzpicture}[x=0.85cm,y=0.85cm,font=\footnotesize]
\fill[black!5] (0,0) rectangle (4,4);
\fill[blue!15] (2,2) rectangle (3,3);
\draw[step=0.5,black!20,very thin] (0,0) grid (4,4);
\draw[step=1,black!70,thick] (0,0) grid (4,4);
\draw[blue!75!black,thick] (2,2) rectangle (3,3);
\node[fill=blue!15,inner sep=1pt] at (2.5,2.5) {$A_{23}$};
\foreach \block in {1,2,3,4}{
  \node at ({\block-0.5},4.25) {$\block$};
  \node at (-0.3,{4.5-\block}) {$\block$};
}
\node at (2,4.7) {block column $J$ ($D=4$)};
\node[rotate=90] at (-0.85,2) {block row $I$};
\draw[{Stealth[length=1.3mm]}-{Stealth[length=1.3mm]}]
  (4.25,2) -- node[right] {$b=2$} (4.25,3);
\node at (2,-0.75)
  {$x=(\underbrace{01}_{x_{\mathrm h}},\underbrace{1}_{x_{\mathrm l}})$};
\end{tikzpicture}
\caption{\label{fig:block-notation}An $N=8$ matrix partitioned into
$D=4$ block rows and columns, with $b=2$ rows and columns per block.
The shaded submatrix $A_{23}$ contains rows $3$--$4$ and columns
$5$--$6$. For the row label $x=011$, the high bits $01$ select
row group $I=2$, and the low bit $1$ selects its second row. %
}
\end{figure}
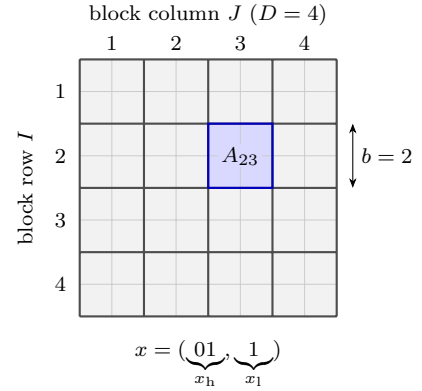

We use $\norm{\cdot}$ for the operator norm and $\norm{\cdot}_F$ for
the unnormalized Frobenius norm. Logarithms are natural unless
written $\log_2$. Both $\Ot(\cdot)$ and $\Omt(\cdot)$ hide factors
polynomial in $n$ only.
Define $F=\mathsf H^{\otimes n}$, where $\mathsf H$ is the Hadamard gate.
$F$ is real and symmetric, with $F^2=I$. For a sign vector
$\sigma\in\{\pm1\}^N$, let $S_\sigma=\diag(\sigma)$ denote the
corresponding diagonal sign matrix.

\subsection{Implementation with clean ancillas}

We measure error on the full output, including the ancillas, using
the standard convention of Ref.~\cite{tan2025unitary}.

\begin{definition}[Approximate implementation with clean ancillas]
\label{def:implementation}
A unitary $\widetilde U$ on $n+m$ qubits \emph{implements} $U\in U(2^n)$ to
error $\epsilon$ if
\[
\norm{\widetilde U(\ket{0^m}\otimes I)-\ket{0^m}\otimes U}\le\epsilon,
\]
where the ancillas are written first. The
$T$-count of $\widetilde U$ is the number of $T$ gates in a Clifford+$T$
circuit for $\widetilde U$, and $m$ is its number of ancillas.
\end{definition}
If $\widetilde U_j$ implements $U_j$ to error
$\epsilon_j$ on a shared ancilla register, then
$\widetilde U_2\widetilde U_1$ implements $U_2U_1$ to error
$\epsilon_1+\epsilon_2$. This follows from the triangle inequality
and unitary invariance of the operator norm.

\subsection{Data loading with Clifford+$T$ circuits}

Lookup tables and classical descriptions of unitaries are inputs to
the compiler. We count the $T$ gates and ancillas required to access
the data and apply the selected unitaries.
Quantum lookup methods trade ancillary workspace for a
lower $T$-count per lookup~\cite{Low2024tradingtgatesdirty,zhu2025unified}.
For example, at constant precision, the optimal worst-case $T$-count
for $n$-qubit diagonal-unitary synthesis is
$\Theta(N^{1/2})$~\cite{gosset2024quantum}.
Related lookup-based sparsity and amplitude oracles reduce
block-encoding gate counts for second-quantized fermionic
Hamiltonians~\cite{liu2025blockencoding}.

The first primitive reads a table entry into an output register by
bitwise XOR.

\begin{lemma}[Boolean and diagonal sign oracles~\cite{Low2024tradingtgatesdirty,gosset2024quantum}]
\label{lem:boolean}
Let $f:\{0,1\}^m\to\{0,1\}^r$ be an arbitrary function. The Boolean oracle
\[
O_f:\quad\ket x\ket y\mapsto\ket x\ket{y\oplus f(x)}
\]
has an exact Clifford+$T$ implementation with $O(\sqrt{r2^m})$
$T$ gates and ancillas.

For any $\sigma\in\{\pm1\}^{2^m}$, the \emph{diagonal sign oracle}
$S_\sigma:=\diag(\sigma)$ acts on computational basis states as
$S_\sigma\ket x=\sigma_x\ket x$.
We implement $S_\sigma$ using the Boolean oracle $O_f$ through
\emph{phase kickback}. Set $r=1$ and $f(x)=(1-\sigma_x)/2$,
and prepare the output qubit in $\ket{-}=(\ket0-\ket1)/\sqrt2$.
A bit flip multiplies $\ket{-}$ by $-1$, so $O_f$ acts as
\[
\ket x\ket{-}\xrightarrow{O_f}(-1)^{f(x)}\ket x\ket{-}
=\sigma_x\ket x\ket{-}.
\]
Thus $O_f$ applies $S_\sigma$ to the input register and preserves
$\ket{-}$ on the output qubit. Preparing $\ket{-}$ and returning it to $\ket0$
use only Clifford gates, so $S_\sigma$ has an exact implementation
with $O(2^{m/2})$ $T$ gates and ancillas.
\end{lemma}

A $\SEL$ oracle uses an address register to select and apply one of
$M$ unitaries acting on the same $d$-dimensional target space.

\begin{lemma}[$\SEL$ oracle~\cite{tan2025unitary}]
\label{lem:tan}
Let $M=2^\mu$, $d=2^\kappa$, and $V_x\in U(d)$ be arbitrary. Define
\begin{equation}
W=\sum_{x\in\{0,1\}^\mu}\ket x\!\bra x\otimes V_x.
\label{eq:multiplexer}
\end{equation}
For every $\delta\in(0,1/2)$ there is a Clifford+$T$ circuit that
implements $W$ to error $\delta$ with $O(d\sqrt{M\Lambda}+d^2\Lambda)$
$T$ gates and ancillas, where $\Lambda=\kappa+\log_2(1/\delta)$.
\end{lemma}

\paragraph*{Single-qubit synthesis.}
\label{tool:rotations}
Single-qubit synthesis~\cite{ross2016optimal} supplies the variable-angle
rotations used in simulation and the global phase of the target.
Every $V\in SU(2)$ has a Clifford+$T$ circuit approximation
$\widetilde V$ satisfying
\begin{equation}
\norm{\widetilde V-V}\le\delta,
\qquad T(\widetilde V)=O(\log(1/\delta)),
\label{eq:single-qubit-synthesis}
\end{equation}
where $T(\widetilde V)$ denotes the circuit's $T$-gate count.

In particular, a
global phase $e^{i\varphi}$ is implemented to error $\delta$ by applying
$\diag(e^{i\varphi},e^{-i\varphi})$ to a clean ancilla in state $\ket0$.

\subsection{Block encodings and Hamiltonian simulation}

A block encoding places a scaled matrix in the zero-ancilla block
of a larger unitary.

\begin{definition}[Block encoding]
A unitary $U_A$ on $a+n$ qubits is an $(\alpha,a,\eta)$-block encoding of
$A\in\C^{N\times N}$ if
\[
\norm{A-\alpha(\bra{0^a}\otimes I)U_A(\ket{0^a}\otimes I)}\le\eta.
\]
\end{definition}

Here $\alpha$ is the subnormalization factor, $a$ is the number of encoding
ancillas, and $\eta$ bounds the matrix error after rescaling by
$\alpha$.
We count one application of $U_A$ or its inverse $U_A^\dagger$,
including their singly controlled versions, as one \emph{query}
to the block encoding. The total number of such queries is the
\emph{query complexity}.

An exact encoding of a Hermitian matrix can be used to
simulate its evolution by quantum singular-value transformation
(QSVT). The following lemma gives the implementation guarantee
and query bound.

\begin{lemma}[Hamiltonian simulation by QSVT~\cite{gilyen2019quantum,low2019hamiltonian}]
\label{lem:hamsim}
Let $U_A$ be an $(\alpha,a,0)$-block encoding of a Hermitian matrix $A$,
$t\in\R$, and $\delta\in(0,1/2)$. There is a circuit $\mathcal S_t(U_A)$
that implements $e^{iAt}$ to error $\delta$ in the sense of
\Cref{def:implementation}. It makes $Q_t=O(\alpha\abs t+\log(1/\delta))$ queries to $U_A$,
$U_A^\dagger$, or their singly controlled versions.

For $\theta=e\alpha\abs t/2\in(0,1)$, the query complexity also satisfies
\begin{equation}
Q_t=O\!\left(1+\frac{\log(1/\delta)}{\log(1/\theta)}\right).
\label{eq:hamsim-small-time}
\end{equation}

If each query, including controlled queries and inverses, has
implementation error at most $\eta$, the overall implementation
error is at most $\delta+Q_t\eta$.
\end{lemma}
The query bounds follow from
\cite[Theorem~58 and Lemma~59]{gilyen2019quantum}. We use
block-encoding error $\delta'=\delta^2/2$: the corresponding
full-output error is at most $\sqrt{2\delta'}=\delta$, as required
by \Cref{def:implementation}. The general and small-time query
bounds follow from Lemma~59 by choosing $q=1$ and
$q=\log(1/\theta)$, respectively. Errors in implementing the
queries add by the composition rule in \Cref{sec:prelim}.

Beyond the queries, the circuit uses $O((a+1)Q_t)$ Clifford and
Toffoli gates, $O(Q_t)$ single-qubit rotations, and $O(a+1)$
reusable clean workspace qubits for the zero-ancilla phase operations.

\subsection{Random product formulas}

To simulate a sum of Hermitian matrices, we combine short evolutions
under the individual pieces. For $A=\sum_{k=1}^sA_k$, the randomized
product formula qDRIFT~\cite{campbell2019random} approximates $e^{iA}$
by choosing these pieces at random (see also
Ref.~\cite{childs2019faster} for randomized orderings).

Assume that each $A_k$ is nonzero and Hermitian, and let
$\lambda=\sum_k\norm{A_k}$. Choose a positive integer $N_q$.
During classical preprocessing, sample independent indices
$K_j\in[s]$ for $j=1,\ldots,N_q$ and set
the evolution times $\tau_j$ by
\begin{equation}
\Pr(K_j=k)=\frac{\norm{A_k}}{\lambda},
\qquad
\tau_j=\frac{\lambda}{N_q\norm{A_{K_j}}}.
\label{eq:qdrift-sampling}
\end{equation}
We call the exact evolution $V_j=e^{i\tau_jA_{K_j}}$ a
\emph{segment}. Applying the segments in sequence gives
$V=V_{N_q}\cdots V_1$, with $V_1$ applied first. Thus $N_q$
counts the sampled evolutions in the product. The same piece
$A_k$ may be selected more than once. We account separately for
errors in implementing these segments.

For synthesis, we choose a fixed sequence of short evolutions whose
product approximates $e^{iA}$ in operator norm.
Chen et al.~\cite{chen2021concentration} show that a randomly sampled
sequence has this property with high probability when it contains
sufficiently many segments. The following lemma gives the required
segment count in terms of the target error and failure probability.

\begin{lemma}[Concentration of qDRIFT~\cite{chen2021concentration}]
\label{lem:qdrift}
Let $A=\sum_{k=1}^sA_k$ with nonzero Hermitian $A_k\in\C^{d\times d}$,
and let $\lambda=\sum_k\norm{A_k}$. For $0<\epsilon\le4\lambda$ and
$0<\delta<1$, qDRIFT uses
\[
N_q=O\!\left(\frac{\lambda^2}{\epsilon^2}\log\frac{2d}{\delta}
+\frac{\lambda^2}{\epsilon}\right)
\]
segments to produce a product $V$ that satisfies
$\norm{V-e^{iA}}\le\epsilon$ with probability at least $1-\delta$.
\end{lemma}

\section{Overview of the construction}
\label{sec:overview}

Our construction reduces the data addressed by each
query to a piece's block encoding.
We use Tan's $\SEL$ oracle synthesis bound~\cite{tan2025unitary}
(\Cref{lem:tan}) for data access. Yuan et al.~\cite{yuan2026general} encode the
full unitary and repeatedly use that encoding for singular-value
amplification. We simulate a Hermitian logarithm using $s$ separately
encoded block-sparse pieces, with few nonzero blocks per block row
and column. Each piece has a smaller lookup table.
qDRIFT~\cite{campbell2019random,chen2021concentration} combines
their evolutions. Balancing the cost per query against the number of
queries yields the $\Ot(N^{8/7})$ $T$-count at constant precision.
\Cref{fig:pipeline} compares the table access, cost per query, and
total query count of the two constructions.

\begin{figure}[!htb]
\centering
\begin{tikzpicture}[
  x=1cm, y=1cm, font=\footnotesize,
  callbox/.style={draw,fill=white,minimum width=0.95cm,
    minimum height=0.55cm,inner sep=2pt},
  access/.style={dashed,thick},
  accessarrow/.style={access,-{Stealth[length=1.3mm]}}]
\node[anchor=west,font=\footnotesize\bfseries] at (0,3.8)
  {(a) Dense encoding: Yuan et al.};
\fill[black!12] (0,1.6) rectangle (1.2,2.8);
\draw[step=0.3,shift={(0,1.6)},black!60] (0,0) grid (1.2,1.2);
\node at (0.6,1.2) {all $D^2$ blocks};
\draw[-{Stealth[length=1.3mm]}] (2,2.1) -- (7.55,2.1);
\node[callbox] (densefirst) at (3,2.1) {query};
\node[callbox] (densesecond) at (4.55,2.1) {query};
\node[fill=white,inner sep=2pt] at (5.65,2.1) {$\cdots$};
\node[callbox] (denselast) at (6.75,2.1) {query};
\draw[access,black!60] (1.2,2.5) -- (1.65,2.5)
  -- (1.65,3) -- (6.75,3);
\draw[accessarrow,black!60] (3,3) -- (densefirst.north);
\draw[accessarrow,black!60] (4.55,3) -- (densesecond.north);
\draw[accessarrow,black!60] (6.75,3) -- (denselast.north);
\node at (4.65,3.35) {same table at every query};
\node at (4.7,1.2) {$\Ot(N+b^2)$ $T$ gates per query};
\node at (4.7,0.7) {$\Ot(\sqrt D)$ queries};
\draw[black!25] (0,0.2) -- (7.55,0.2);

\node[anchor=west,font=\footnotesize\bfseries] at (0,-0.35)
  {(b) Separate piece encodings: this work};
\fill[blue!15] (0,-2.6) rectangle (0.6,-1.4);
\draw[step=0.3,shift={(0,-2.6)},blue!65!black]
  (0,0) grid (0.6,1.2);
\fill[orange!20] (0.8,-2.6) rectangle (1.4,-1.4);
\draw[step=0.3,shift={(0.8,-2.6)},orange!70!black]
  (0,0) grid (0.6,1.2);
\node at (0.3,-2.9) {$H_1$};
\node at (1.1,-2.9) {$H_2$};
\node at (0.7,-3.3) {$O(D^2/s)$ slots each};
\draw[-{Stealth[length=1.3mm]}] (2,-2.05) -- (7.55,-2.05);
\node[callbox,fill=blue!15] (piecefirst) at (3,-2.05) {query};
\node[callbox,fill=orange!20] (piecesecond) at (4.55,-2.05) {query};
\node[fill=white,inner sep=2pt] at (5.65,-2.05) {$\cdots$};
\node[callbox,fill=blue!15] (piecelast) at (6.75,-2.05) {query};
\draw[access,blue!65!black] (0.3,-1.4) -- (0.3,-0.9) -- (6.75,-0.9);
\draw[accessarrow,blue!65!black] (3,-0.9) -- (piecefirst.north);
\draw[accessarrow,blue!65!black] (6.75,-0.9) -- (piecelast.north);
\draw[accessarrow,orange!70!black] (1.1,-1.4) -- (1.1,-0.65)
  -- (4.55,-0.65) -- (piecesecond.north);
\node at (4.7,-2.8) {$\Ot(N/\sqrt s+b^2)$ $T$ gates per query};
\node at (4.7,-3.3) {$\Ot(\sqrt D+s)$ queries};
\end{tikzpicture}
\caption{\label{fig:pipeline}Schematic table access at fixed precision.
Boxes in (a) denote queries to the dense block
encoding of the target matrix in Yuan et al.~\cite{yuan2026general}.
Boxes in (b) denote queries to the block encodings $U_{H_k}$ of
our $s$ pieces,
illustrated with two piece tables. Each table has $O(D^2/s)$ block
slots in the parameter range used in the construction. Dashed
arrows link queries to their tables.
The query counts are totals over the respective circuits.}
\end{figure}
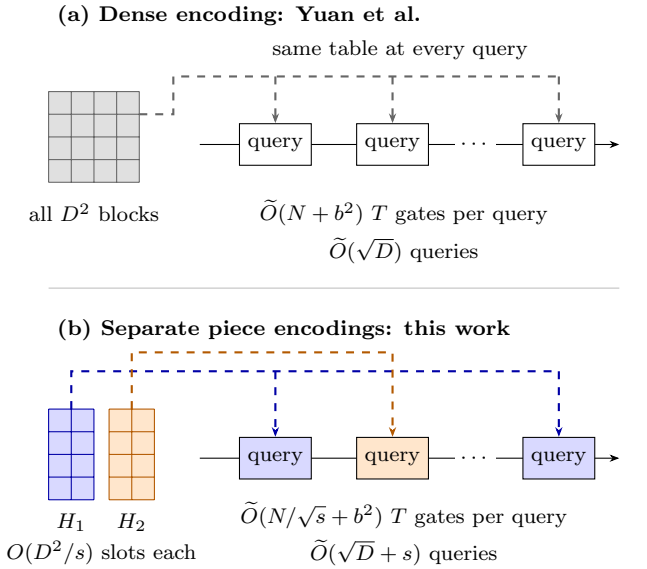

\subsection{Cost of the dense block encoding}

To identify where our construction can save $T$ gates, we first
examine the cost of encoding the target unitary as a whole.
In the dense construction of Yuan et al.~\cite{yuan2026general},
each query addresses a lookup table containing data for all $D^2$
blocks. Its $T$-count is $\Ot(N+b^2)$ by \Cref{lem:tan}.
The term $N$ pays for the lookup, while $b^2$ pays for processing
the selected block. Increasing the block size makes this processing
more expensive.

Larger blocks also reduce the number of queries. At fixed accuracy,
the construction uses $\Ot(\sqrt{N/b})$ queries, so the two costs give
\begin{equation}
T=\Ot\!\bigl(
\underbrace{\sqrt{N/b}}_{\text{queries}}
\underbrace{(N+b^2)}_{\text{$T$ gates per query}}
\bigr).
\label{eq:dense-encoding-cost}
\end{equation}

Choosing $b=\Theta(\sqrt N)$ balances these effects: each query
costs $\Ot(N)$ $T$ gates, and $\Ot(N^{1/4})$ queries suffice. Their
product gives $\Ot(N^{5/4})$. The dense encoding's subnormalization
factor sets the amplification query scale. For typical target unitaries,
this dense encoding requires a subnormalization factor
$\Omega(\sqrt{N/b})$.
\Cref{app:normalization} gives the cost derivation and this lower bound.

The opportunity for improvement is the repeated lookup cost. Every
query pays the term $N$, even after the block size has been optimized.
Our construction uses smaller lookup tables for separate matrix
pieces. We choose the number of pieces to balance the lower cost
per query against the additional queries needed to simulate their sum.

\subsection{Block-sparse Hamiltonian simulation}

We obtain smaller lookup tables by changing what we encode. Write
$U=e^{iH}$, where $H$ is the Hermitian logarithm with eigenvalues in
$(-\pi,\pi]$. Implementing $U$ is then equivalent to simulating $H$
for time one. We will approximate this evolution using short
evolutions under separate matrix pieces.

Before splitting the matrix, we separate its scalar part
$\varphi I$, where $\varphi=\tr H/N$. This part contributes only
the global phase $e^{i\varphi}$, which we apply on a clean ancilla.
We then change basis using the random real orthogonal matrix
$P$ defined in \Cref{lem:flattening}:
\begin{equation}
H_0=H-\varphi I,\qquad H'=PH_0P^T.
\label{eq:generator}
\end{equation}
This choice of $P$ makes all block norms small with high probability.
The scalar part is handled
separately because it is unchanged by a change of basis. Once we
simulate $H'$, the identity $U=e^{i\varphi}P^Te^{iH'}P$ recovers the
target.

We now split $H'$ into $s$ matrices $H_1,\ldots,H_s$, called pieces.
Each piece is still $N\times N$: it keeps its assigned blocks in
their original positions and sets all other blocks to zero. Every
block belongs to exactly one piece, so $H'=\sum_{k=1}^{s}H_k$.

Each piece must be Hermitian so that its short evolutions
$e^{itH_k}$ are unitary for real $t$. We enforce this by assigning
each off-diagonal pair $H'_{IJ}$ and $H'_{JI}=(H'_{IJ})^\dagger$
together to one piece chosen uniformly at random.
Diagonal blocks are already Hermitian and are assigned individually
in the same way to one of the $s$ pieces. Choices for distinct pairs
and diagonal blocks are independent. This guarantees $H_k^\dagger=H_k$ for every piece and
every assignment.
\Cref{fig:hermitian-split} illustrates this assignment with two colors.

\begin{figure}[!htb]
\centering
\begin{tikzpicture}[x=0.47cm,y=0.47cm,font=\footnotesize]
\fill[orange!20] (0,0) rectangle (4,4);
\fill[orange!20] (11.8,0) rectangle (15.8,4);
\foreach \blockrow/\blockcolumn in {1/2,2/1,2/3,3/2,3/4,4/3}{
  \fill[blue!15] ({\blockcolumn-1},{4-\blockrow}) rectangle ++(1,1);
  \fill[blue!15] ({5.9+\blockcolumn-1},{4-\blockrow}) rectangle ++(1,1);
  \fill[white] ({11.8+\blockcolumn-1},{4-\blockrow}) rectangle ++(1,1);
}
\foreach \matrixshift/\matrixname/\blockcount in {0/{H'}/16,5.9/{H_1}/6,11.8/{H_2}/10}{
  \draw[step=1,shift={(\matrixshift,0)},black!55] (0,0) grid (4,4);
  \node at ({\matrixshift+2},5.4) {$\matrixname$};
  \node at ({\matrixshift+2},-0.7) {$\blockcount$ blocks};
  \foreach \blockindex in {1,2,3,4}{
    \node at ({\matrixshift+\blockindex-0.5},4.45) {$\blockindex$};
    \node at ({\matrixshift-0.45},{4.5-\blockindex}) {$\blockindex$};
  }
}
\foreach \matrixshift in {0,5.9}{
  \draw[blue!75!black,very thick] ({\matrixshift+2},2) rectangle ++(1,1);
  \draw[blue!75!black,very thick] ({\matrixshift+1},1) rectangle ++(1,1);
}
\node at (4.75,2) {$=$};
\node at (10.65,2) {$+$};
\end{tikzpicture}
\caption{\label{fig:hermitian-split}One possible assignment with $D=4$
and $s=2$. Colored cells are nonzero $b\times b$ blocks; blank cells
are zero. Each block and its conjugate-transpose partner are assigned
to the same piece, indicated by their shared color. The outlined
blocks show one such pair. Thus $H'=H_1+H_2$ with Hermitian pieces.}
\end{figure}
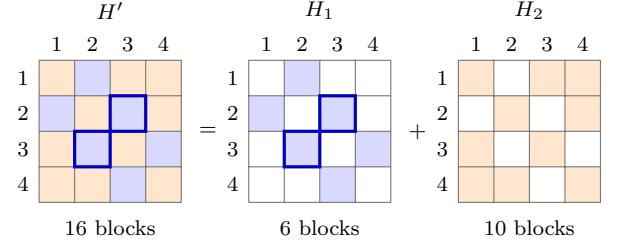

We build a separate lookup table for each piece by listing the row
indices of its nonzero blocks in each block column. Each list
position is called a slot. A table address $(J,u)$ specifies slot
$u$ in column $J$.
For example, column $J=2$ of the blue piece $H_1$ in
\Cref{fig:hermitian-split} has row-index list $(1,3)$, so address
$(J,u)=(2,2)$ selects row $I=3$ and hence block $(H_1)_{3,2}$.
\Cref{fig:slot-table} shows the resulting table.

\begin{figure}[!htb]
\centering
\begin{tikzpicture}[x=1cm,y=1cm,font=\footnotesize]
\node at (0.9,0.75) {$K=H_1$};
\foreach \blockrow/\blockcolumn in {1/2,2/1,2/3,3/2,3/4,4/3}{
  \fill[blue!15] ({0.45*(\blockcolumn-1)},{-0.45*\blockrow})
    rectangle ++(0.45,0.45);
}
\draw[step=0.45,black!55] (0,-1.8) grid (1.8,0);
\foreach \blockindex in {1,2,3,4}{
  \node at ({0.45*(\blockindex-0.5)},0.25) {$\blockindex$};
  \node at (-0.25,{-0.45*(\blockindex-0.5)}) {$\blockindex$};
}
\draw[blue!75!black,very thick] (0.9,-0.9) rectangle ++(0.45,0.45);
\draw[blue!75!black,very thick] (0.45,-1.35) rectangle ++(0.45,0.45);
\draw[-{Stealth[length=1.5mm]}] (2,-0.9) -- (2.45,-0.9);
\node at (0.9,-2.35) {$D=4,\quad m=2$};
\node at (0.9,-2.85) {$Dm=8$ slots};

\node at (4.2,0.75) {lookup table};
\node at (3.6,0.35) {$u=1$};
\node at (4.8,0.35) {$u=2$};
\fill[blue!8] (3,-2.4) rectangle (5.4,0);
\fill[black!10] (4.2,-0.6) rectangle (5.4,0);
\fill[black!10] (4.2,-2.4) rectangle (5.4,-1.8);
\fill[blue!22] (4.2,-1.2) rectangle (5.4,-0.6);
\fill[blue!22] (3,-1.8) rectangle (4.2,-1.2);
\draw[xstep=1.2,ystep=0.6,shift={(3,-2.4)},black!55]
  (0,0) grid (2.4,2.4);
\foreach \columnlabel/\firstneighbor/\secondneighbor in
  {1/2/{pad},2/1/3,3/2/4,4/3/{pad}}{
  \node at (2.7,{-0.6*(\columnlabel-0.5)}) {$\columnlabel$};
  \node at (3.6,{-0.6*(\columnlabel-0.5)}) {\firstneighbor};
  \node at (4.8,{-0.6*(\columnlabel-0.5)}) {\secondneighbor};
}
\draw[blue!75!black,very thick] (4.2,-1.2) rectangle (5.4,-0.6);
\draw[blue!75!black,very thick] (3,-1.8) rectangle (4.2,-1.2);
\node at (4.2,-2.85) {addresses $(J,u)$};
\node[text=blue!75!black] at (4.2,-3.35) {$(2,2)\longleftrightarrow(3,1)$};
\node at (4.2,-3.8) {paired addresses};
\end{tikzpicture}
\caption{\label{fig:slot-table}Lookup table for the blue piece $H_1$ in
\Cref{fig:hermitian-split}. Table rows correspond to block columns
$J$, and table columns to slots $u=1,2$. Numbers give block row
indices $I$, while gray cells represent zero padding. The outlined
blocks $(H_1)_{3,2}$ and $(H_1)_{2,3}$ are adjoints, with paired
addresses $(2,2)$ and $(3,1)$, respectively.}
\end{figure}
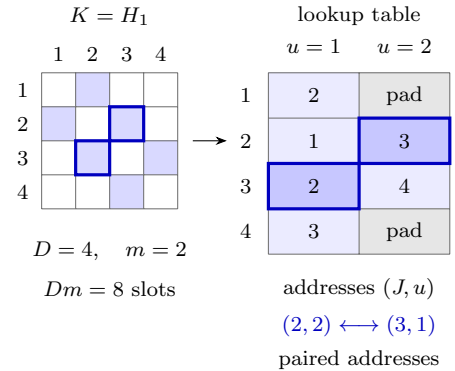

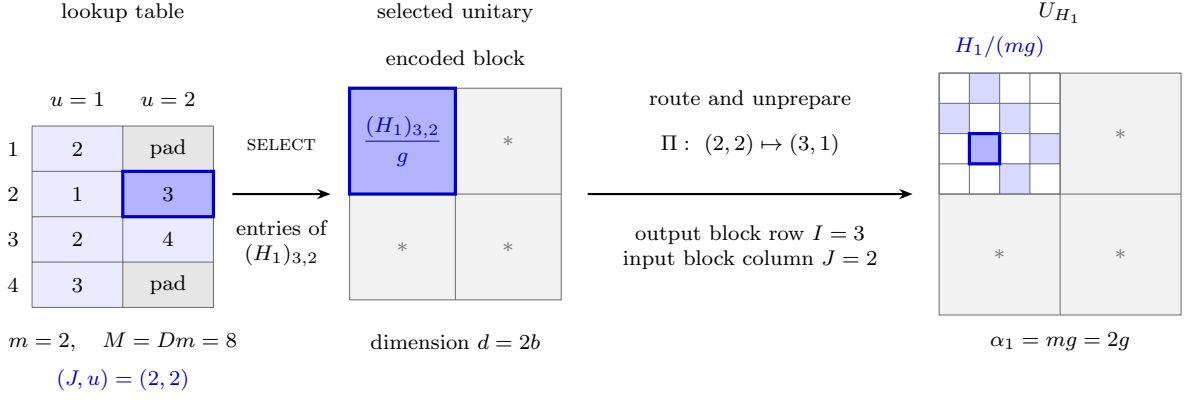
\begin{figure*}[t]
\centering
\begin{tikzpicture}[x=1cm,y=1cm,font=\footnotesize]
\node at (1.2,3.4) {lookup table};
\node at (5.6,3.4) {selected unitary};
\node at (13.6,3.4) {$U_{H_1}$};

\node at (0.6,2.25) {$u=1$};
\node at (1.8,2.25) {$u=2$};
\fill[blue!8] (0,-0.5) rectangle (2.4,1.9);
\fill[black!10] (1.2,1.3) rectangle (2.4,1.9);
\fill[black!10] (1.2,-0.5) rectangle (2.4,0.1);
\fill[blue!30] (1.2,0.7) rectangle (2.4,1.3);
\draw[xstep=1.2,ystep=0.6,shift={(0,-0.5)},black!55]
  (0,0) grid (2.4,2.4);
\foreach \columnlabel/\firstentry/\secondentry in
  {1/2/{pad},2/1/3,3/2/4,4/3/{pad}}{
  \node at (-0.25,{1.9-0.6*(\columnlabel-0.5)}) {$\columnlabel$};
  \node at (0.6,{1.9-0.6*(\columnlabel-0.5)}) {\firstentry};
  \node at (1.8,{1.9-0.6*(\columnlabel-0.5)}) {\secondentry};
}
\draw[blue!75!black,very thick] (1.2,0.7) rectangle (2.4,1.3);
\node at (1.2,-0.95) {$m=2,\quad M=Dm=8$};
\node[text=blue!75!black] at (1.2,-1.45) {$(J,u)=(2,2)$};

\draw[-{Stealth[length=1.8mm]},thick] (2.65,1) -- (3.95,1);
\node at (3.3,1.65) {$\SEL$};
\node[align=center] at (3.3,0.35) {entries of\\$(H_1)_{3,2}$};

\fill[black!5] (4.2,-0.4) rectangle (7,2.4);
\fill[blue!30] (4.2,1) rectangle (5.6,2.4);
\draw[black!55] (4.2,-0.4) rectangle (7,2.4);
\draw[black!55] (5.6,-0.4) -- (5.6,2.4);
\draw[black!55] (4.2,1) -- (7,1);
\draw[blue!75!black,very thick] (4.2,1) rectangle (5.6,2.4);
\node[text=blue!75!black] at (4.9,1.7) {$\dfrac{(H_1)_{3,2}}g$};
\node[text=black!55] at (6.3,1.7) {$\ast$};
\node[text=black!55] at (4.9,0.3) {$\ast$};
\node[text=black!55] at (6.3,0.3) {$\ast$};
\node at (5.6,2.8) {encoded block};
\node at (5.6,-0.95) {dimension $d=2b$};

\draw[-{Stealth[length=1.8mm]},thick] (7.35,1) -- (11.65,1);
\node at (9.5,2.25) {route and unprepare};
\node at (9.5,1.65) {$\Pi:\ (2,2)\mapsto(3,1)$};
\node[align=center] at (9.5,0.3)
  {output block row $I=3$\\input block column $J=2$};

\fill[black!5] (12,-0.6) rectangle (15.2,2.6);
\fill[white] (12,1) rectangle (13.6,2.6);
\foreach \blockrow/\blockcolumn in {1/2,2/1,2/3,3/2,3/4,4/3}{
  \fill[blue!15] ({12+0.4*(\blockcolumn-1)},{1+0.4*(4-\blockrow)})
    rectangle ++(0.4,0.4);
}
\fill[blue!30] (12.4,1.4) rectangle (12.8,1.8);
\draw[step=0.4,shift={(12,1)},black!55] (0,0) grid (1.6,1.6);
\draw[black!55] (12,-0.6) rectangle (15.2,2.6);
\draw[black!55] (13.6,-0.6) -- (13.6,2.6);
\draw[black!55] (12,1) -- (15.2,1);
\draw[blue!75!black,very thick] (12.4,1.4) rectangle (12.8,1.8);
\node[text=blue!75!black] at (12.8,2.95) {$H_1/(mg)$};
\node[text=black!55] at (14.4,1.8) {$\ast$};
\node[text=black!55] at (12.8,0.2) {$\ast$};
\node[text=black!55] at (14.4,0.2) {$\ast$};
\node at (13.6,-0.95) {$\alpha_1=mg=2g$};
\end{tikzpicture}
\caption{\label{fig:encoding-chain}Building $U_{H_1}$ from the lookup
table in \Cref{fig:slot-table} and the matrix entries. The solid outlines
trace address $(J,u)=(2,2)$: its entry $I=3$ selects a unitary with
encoded block $(H_1)_{3,2}/g$. Each selected unitary is constructed
from the corresponding matrix block and $g$ during classical
preprocessing. $\SEL$ applies that unitary for the corresponding
address. The permutation $\Pi$ changes the block label from input
$J=2$ to output $I=3$ and the slot from $u=2$ to $u'=1$.
This branch therefore contributes to block $(3,2)$ of the encoded
matrix.
}
\end{figure*}

The column lists have different lengths, so we add dummy slots
representing zero to give them a common length $m$. Here $m=2$,
so the first column's list $(2)$ becomes $(2,\mathrm{pad})$.
The table then has eight entries across four columns, including
two padding entries, compared with sixteen for the dense table.

In each block column, the $D$ possible blocks are assigned uniformly
among $s$ pieces, giving $D/s$ assigned positions per piece on average.
For our parameters, \Cref{lem:split} bounds every column list by
$O(D/s)$ with high probability. We pad these lists to a common length
$m$, rounding up to a power of two, which increases the required length
by less than a factor of two. Thus $m=O(D/s)$, and each piece's $D$
column lists give $M=Dm=O(D^2/s)$ table entries, compared with $D^2$
for the dense table.

Using the lookup tables and the entries of the matrix blocks,
we construct a block encoding of each piece, as illustrated in
\Cref{fig:encoding-chain}.
Let $g>0$ bound the operator norm of every $b\times b$ block of $H'$.
For each nonzero piece $H_k$, \Cref{lem:encoding} gives an exact
block encoding $U_{H_k}$ with subnormalization factor $\alpha_k=mg$.

For each nonpadding slot address $(J,u)$ with table entry $I$,
we embed the corresponding scaled $b\times b$ block $(H_k)_{IJ}/g$
in a $2b\times2b$ unitary. Applying \Cref{lem:tan} with $d=2b$ and $M=Dm$ gives cost
$\Ot(b\sqrt M+b^2)$. Since $M=O(D^2/s)$, splitting reduces the
lookup term from $N$ to $O(N/\sqrt s)$ while leaving the $b^2$
processing term unchanged. The block-label operations fit this
bound, so the $T$-gate cost of one query to $U_{H_k}$ is
\begin{equation}
c_{\rm query}=\Ot(b\sqrt M+b^2)
=\Ot(N/\sqrt s+b^2).
\label{eq:cuse}
\end{equation}

The total $T$-count also depends on the number of queries needed
to approximate $e^{iH'}$.
qDRIFT approximates this evolution by a product of short evolutions
$e^{i\tau_jH_{K_j}}$, using the sampling rule in
\Cref{eq:qdrift-sampling} with $A_k=H_k$ and omitting zero pieces.
For our parameters, \Cref{lem:split} gives
$\lambda=\sum_k\norm{H_k}=\Ot(\sqrt s)$, so
\Cref{lem:qdrift} supplies $N_q=\Ot(s)$ segments at fixed target
error $\epsilon$.

We implement each segment by QSVT using its corresponding block
encoding $U_{H_{K_j}}$. At error $\epsilon/(4N_q)$ per segment,
\Cref{lem:hamsim} bounds the total query count by
\[
Q=O\!\left(\sum_{j=1}^{N_q}\alpha_{K_j}\tau_j
  +N_q\log(N_q/\epsilon)\right).
\]

To bound the first term, recall that $m=O(D/s)$ and flattening
gives $g=\Ot(1/\sqrt D)$ for $\norm{H'}=O(1)$. Thus
$\alpha_k=mg=\Ot(\sqrt D/s)$ and
$\sum_k\alpha_k=\Ot(\sqrt D)$ (\Cref{cor:normsum}).
\Cref{lem:constant-selection} guarantees a sampled product meeting
the required accuracy and satisfying
\[
\sum_{j=1}^{N_q}\alpha_{K_j}\tau_j
=O\!\left(\sum_k\alpha_k\right)=\Ot(\sqrt D).
\]
Together with $N_q=\Ot(s)$, this gives $Q=\Ot(\sqrt D+s)$
queries in total.

Multiplying this query complexity by the cost per query in
\Cref{eq:cuse} gives
\begin{equation}
T=\Ot\!\bigl(
\underbrace{(\sqrt D+s)}_{\text{\#queries}}
\underbrace{(N/\sqrt s+b^2)}_{\text{$T$ gates per query}}
\bigr).
\label{eq:aggregate-cost}
\end{equation}
Taking $s$ of order $\sqrt D$, and
substituting $D=N/b$ gives
\begin{equation}
T=\Ot\!\left(N^{5/4}b^{-1/4}+N^{1/2}b^{3/2}\right).
\label{eq:ledger}
\end{equation}
The terms balance at $b=\Theta(N^{3/7})$, yielding
$\Ot(N^{8/7})$ $T$ gates.

\section{Constructing block-sparse encodings}
\label{sec:pieces}

\begingroup
To prove the resource bounds used in \Cref{sec:overview}, we construct
block encodings of the Hermitian pieces $H_k$.
Similarity flattening (\Cref{sec:flattening}) gives a common upper
bound $g$ on the operator norms $\norm{H'_{IJ}}$ of the $b\times b$
blocks of $H'$.
Random splitting (\Cref{sec:split}) bounds the number of nonzero
blocks in each block column of each piece, determining the padded
number $m$ of slots per block column and hence the $Dm$ entries
in that piece's lookup table. It also bounds the piece
norms $\norm{H_k}$, whose sum controls the number of qDRIFT
segments needed to approximate $e^{iH'}$.
Finally, \Cref{sec:encoding} constructs block encodings of the pieces
with subnormalization factor $\alpha_k=mg$ and bounds their
$T$-gate cost per query.
\par\endgroup

\subsection{Similarity flattening}
\label{sec:flattening}

To obtain the required subnormalization factor, we need a real orthogonal
matrix $P$ for which every block of $H'=PH_0P^T$ has norm
$\Ot(\norm{H_0}/\sqrt D)$. %

\begingroup
Let $I_d$ denote the $d\times d$ identity matrix. An orthogonal
similarity transformation is a change of basis $A\mapsto PAP^T$,
where $P$ is real orthogonal and hence $P^T=P^{-1}$. It leaves
the scalar part of $H$ unchanged:
\[
P(\varphi I_N)P^T=\varphi I_N.
\]
Each diagonal $b\times b$ block of this scalar matrix is
$\varphi I_b$, with operator norm $\abs\varphi$ in every basis.
We therefore remove the scalar part in \Cref{eq:generator} and
apply flattening to the traceless matrix $H_0$. The quantities
we seek to make small are the operator norms
$\norm{(PH_0P^T)_{IJ}}$ of its individual $b\times b$ blocks.

A random sign layer means multiplication by a diagonal matrix
$S_\sigma=\diag(\sigma)$, whose diagonal entries are independent
and take the values $+1$ and $-1$ with equal probability.
To see why one sign matrix may be insufficient, consider
$P_1=FS_\sigma F$ with $F=\mathsf H^{\otimes n}$.
If $FH_0F$ is diagonal, it commutes
with $S_\sigma$, so the signs cancel:
\[
P_1H_0P_1^T=FS_\sigma(FH_0F)S_\sigma F=H_0.
\]
For such inputs, this random change of basis leaves every block
unchanged. We therefore use two independent sign matrices,
$S_\sigma$ and $S_\tau$, separated by applications of $F$.
The following lemma defines the resulting $P$ and bounds every
block's operator norm with high probability for any fixed
traceless Hermitian input.
\par\endgroup

\begin{lemma}[Similarity flattening]
\label{lem:flattening}
Let $H\in\C^{N\times N}$ be Hermitian with $\tr H=0$.
For independent, uniformly random sign vectors
$\sigma,\tau\in\{\pm1\}^N$, set
\begin{equation}
P=F\,S_\tau\,F\,S_\sigma\,F .
\label{eq:flattening-similarity}
\end{equation}
Then $P$ is a real orthogonal matrix and $e^{iPHP^T}=Pe^{iH}P^T$. For every
$\delta\in(0,1/2)$, with probability at least $1-\delta$ over $\sigma,\tau$,
\begin{equation}
\max_{I,J\in[D]}\big\|(PHP^T)_{IJ}\big\|
\ \le\ \frac{C_A\norm H}{\sqrt D}\,\log\frac N\delta,
\label{eq:flatness}
\end{equation}
where $C_A>0$ is an absolute constant. Moreover, $P$ and $P^T$ each have an
exact Clifford+$T$ implementation with $O(\sqrt N)$ $T$ gates.
\end{lemma}

\begin{proof}
We defer the proof to \Cref{app:flattening}.
\end{proof}

The $1/\sqrt D$ dependence in \Cref{eq:flatness} is optimal in the
worst case, up to the logarithmic factor. To see this, consider the
traceless Hermitian matrix $H=\diag(I_{N/2},-I_{N/2})$, which has
$\norm H=1$ and $\norm H_F^2=N$. A change of basis preserves the
sum of squared magnitudes of the matrix entries, $\norm H_F^2$.
Since $\norm B_F^2\le b\norm B^2$ for every $b\times b$ matrix $B$,
every real orthogonal matrix $P$ satisfies
\[
\begin{aligned}
N=\norm H_F^2
&=\sum_{I,J\in[D]}\norm{(PHP^T)_{IJ}}_F^2\\
&\le bD^2\max_{I,J\in[D]}\norm{(PHP^T)_{IJ}}^2.
\end{aligned}
\]
Since $N=bD$, every such $P$ leaves at least one block with
\[
\max_{I,J\in[D]}\norm{(PHP^T)_{IJ}}
\ge\frac1{\sqrt D}.
\]
The upper bound in \Cref{eq:flatness} is therefore within a factor
$C_A\log(N/\delta)$ of this lower
bound. %

\subsection{Random block-sparse split}
\label{sec:split}

We split $H'$ into $s$ Hermitian pieces to obtain smaller lookup
tables while controlling the piece norms. As in
\Cref{fig:hermitian-split}, each off-diagonal adjoint block pair
and each diagonal block is assigned independently to a uniformly
random piece. Thus $\E(H_k)=H'/s$, with $D/s$ assigned block
positions per piece in each block column on average.

To control the lookup cost per query, we bound the number of
nonzero blocks per block column. This determines the padded
length $m$ and the $Dm$ entries in each piece's lookup table.
We also bound the piece norms, since the qDRIFT segment bound
in \Cref{lem:qdrift} is quadratic in $\lambda=\sum_k\norm{H_k}$.
The following lemma gives both bounds simultaneously for all
pieces with high probability.

\begin{lemma}[Random split]
\label{lem:split}
Let $H'\in\C^{N\times N}$ be Hermitian with
$\norm{H'_{IJ}}\le g$ for all $I,J\in[D]$.
Let $s\ge1$ be an integer and $\delta\in(0,1/2)$, and put
$\ell=\log(8Ns/\delta)$. Form $H_1,\ldots,H_s$ by assigning each
off-diagonal adjoint block pair and each diagonal block independently
to a uniformly random piece, with zero blocks elsewhere.
Then each $H_k$ is Hermitian and $\sum_kH_k=H'$.

With probability at least $1-\delta$, both of the following bounds
hold for every $k\in[s]$. The piece norm satisfies
\begin{equation}
\norm{H_k}\le\norm{H'}/s+\norm{H'}\sqrt{2\ell/s}+\tfrac23g\ell,
\label{eq:split-piece-norm}
\end{equation}
and every block row and block column of $H_k$ contains at most
\[
m_0:=D/s+\sqrt{2D\ell/s}+\ell
\]
nonzero blocks. On the same event, if $g\le\norm{H'}/\sqrt{s\ell}$ then
$\norm{H_k}\le4\norm{H'}\sqrt{\ell/s}$ for all $k$.
\end{lemma}

\begin{proof}
We defer the proof to \Cref{app:split}.
\end{proof}

\subsection{Block encoding of a piece}
\label{sec:encoding}

\begingroup
We construct a block encoding $U_K$ of a nonzero piece
$K=H_k$ using its lookup table (\Cref{fig:slot-table}). On the
event in \Cref{lem:split}, pad each column's list to length $m$,
the smallest power of two at least $m_0$.

Since $\norm{K_{IJ}}\le g$, each block admits an exact
$(g,1,0)$-block encoding.
For each nonpadding address $(J,u)$, the lookup table records
$I$, and $\SEL$ applies the corresponding unitary constructed
from $K_{IJ}$ and $g$ during classical preprocessing.
The circuit in \Cref{fig:sparse-encoding} prepares the slot register
$\mathsf S$ in the uniform state
$\frac1{\sqrt m}\sum_z\ket z_{\mathsf S}$,
applies $\SEL$, routes each block to row $I$ and column $J$ using
$\Pi$, and unprepares the slot register. The resulting encoded block is
\[
\frac1{\sqrt m}\,\frac{K_{IJ}}g\,\frac1{\sqrt m}
=\frac{K_{IJ}}{mg}.
\]
Thus $U_K$ is an exact $(mg,\log_2m+1,0)$-block encoding of $K$.
\par\endgroup

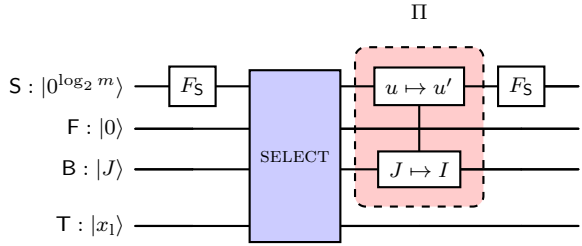
\begin{figure}[!htb]
\centering
\tikzset{
  U/.style={draw,rectangle,minimum width=0.35cm,
    minimum height=0.35cm,fill=white}
}
\begin{quantikz}[row sep=0.3cm,column sep=0.45cm,
  /tikz/every node/.append style={font=\footnotesize}]
\lstick{$\mathsf S:\ket{0^{\log_2m}}$}
  & \gate[style={U}]{F_{\mathsf S}}
  & \gate[4,style={U,fill=blue!20}]{\SEL}
  & \gate[style={U}]{u\mapsto u'}\wire[d][2]{q}
    \gategroup[3,steps=1,style={dashed,rounded corners,fill=red!20},
      background,label style={label position=above,anchor=south,yshift=0.05cm}]{$\Pi$}
  & \gate[style={U}]{F_{\mathsf S}} & \qw \\
\lstick{$\mathsf F:\ket0$}
  & \qw & & \qw & \qw & \qw \\
\lstick{$\mathsf B:\ket J$}
  & \qw & & \gate[style={U}]{J\mapsto I} & \qw & \qw \\
\lstick{$\mathsf T:\ket{x_{\mathrm l}}$}
  & \qw & & \qw & \qw & \qw
\end{quantikz}
\caption{\label{fig:sparse-encoding}Block encoding $U_K$ of a Hermitian
piece $K$ with $b\times b$ blocks. For $K=H_1$, this circuit implements
the construction in \Cref{fig:encoding-chain}.
The data registers are $\mathsf B$,
holding the input block column $J$, and $\mathsf T$, holding the
position $x_{\mathrm l}$ within a block. The encoding ancillas are
$\mathsf S$, holding one of $m$ padded slots per column, and the
single flag qubit $\mathsf F$. The Hadamard layer
$F_{\mathsf S}=\mathsf H^{\otimes\log_2m}$ prepares and unprepares
the uniform slot state. For each nonpadding address $(J,u)$
identifying row $I$, $\SEL$ applies a unitary on
$\mathsf F,\mathsf T$ whose $\mathsf F=0$ block is $K_{IJ}/g$,
where $g$ bounds all block operator norms. Padding encodes zero.
The permutation $\Pi$ exchanges addresses $(J,u)$ and $(I,u')$
of adjoint blocks on $\mathsf S,\mathsf B$. The block with
$\mathsf S=\mathsf F=0$ on input and output is $K/(mg)$.}
\end{figure}

\begingroup
The following lemma bounds the $T$-gate and ancilla costs of a
query to $U_K$.
\par\endgroup

\begin{lemma}[Sparse block encoding via $\SEL$]
\label{lem:encoding}
Let $K\in\C^{N\times N}$ be a nonzero Hermitian matrix with
$\norm{K_{IJ}}\le g$ for all blocks, where $g>0$. Suppose that every block
column of $K$ has at most $m$ nonzero blocks, where $m$ is a power of two.
Then there is an $(mg,\log_2m+1,0)$-block encoding $U_K$ of $K$.
\begingroup
It consists of Hadamard gates, one $\SEL$ oracle selecting among
$M=Dm$ unitaries of dimension $d=2b$, and a basis-state permutation
$\Pi$ defined as follows. List the row indices of each column's
nonzero blocks in $m$ slots, adding dummy slots for zero blocks
where needed. If slot $u$ of column $J$ records row $I$ with
$K_{IJ}\ne0$, let $u'$ be the unique slot in column $I$ recording
row $J$ of the adjoint block $K_{JI}=K_{IJ}^\dagger$.
Define the address map
\[
q(J,u)=(I,u').
\]
For a padding address, set $q(J,u)=(J,u)$; addresses of diagonal
blocks are also fixed. The unitary $\Pi$ permutes the basis
states of the block-label and slot registers according to $q$
and acts as the identity on the flag qubit and the register
for the position within a block. Hermiticity ensures that the
adjoint partner exists and that $q(q(J,u))=(J,u)$, so
$\Pi^{-1}=\Pi$.
\par\endgroup

For every $\eta\in(0,1/2)$, an ordinary or singly controlled query to
$U_K$, or its inverse, can be implemented to error $\eta$ in the
sense of \Cref{def:implementation} with
\begin{equation}
O\!\left(b\sqrt{Dm\Lambda}+b^2\Lambda+\sqrt{Dm\log(2Dm)}\right),
\label{eq:encodingcost}
\end{equation}
$T$ gates and ancillas, where $\Lambda=\log_2(2b)+\log_2(1/\eta)$. The
last term, which comes from the address permutation, is dominated by the
first whenever $b^2\Lambda\ge\log(2Dm)$.
\end{lemma}

\begin{proof}
We defer the proof to \Cref{app:sparse-encoding}.
\end{proof}

Recall from \Cref{lem:split} that $s$ counts all pieces and
$\ell=\log(8Ns/\delta)$, where $\delta$ is the allowed failure
probability. When $s\ell\le D$, the sparsity bound gives
$m_0=O(D/s)$. Rounding $m_0$ up to a power of two then gives
$m=O(D/s)$ and $Dm=O(D^2/s)$ lookup table entries per piece.
Applying \Cref{lem:encoding} gives the following bounds on the
cost per query and the sum of subnormalization factors.

\begin{corollary}[Sum of subnormalization factors and per-query cost]
\label{cor:normsum}
Suppose the conclusions of \Cref{lem:split} hold with $g>0$ and
$s\ell\le D$. Discard zero pieces and encode the remaining $H_k$ by
\Cref{lem:encoding}, with $m$ the smallest power of two that is at
least $m_0$. Then their subnormalization factors $\alpha_k=mg$ satisfy
$\sum_{k:\,H_k\ne0}\alpha_k\le7Dg$.
If $b^2\Lambda\ge\log(2Dm)$, each query to a
block encoding $U_{H_k}$, its inverse, or a singly controlled
version of either costs $O(N\sqrt{\Lambda/s}+b^2\Lambda)$ $T$ gates
and ancillas. Here $s$ counts all pieces, including those discarded as zero.
\end{corollary}

\begin{proof}
We defer the proof to \Cref{app:sparse-encoding}.
\end{proof}

Choose a constant failure probability in \Cref{lem:flattening}.
When its bound holds, we can take $g=O(\norm{H'}\log N/\sqrt D)$,
because $H'$ has the same operator norm as $H_0$.
Under the assumptions of \Cref{cor:normsum}, including
$s\ell\le D$, this gives
\[
\sum_{k:\,H_k\ne0}\alpha_k
\le7Dg=O(\sqrt D\log N\norm{H'}).
\]
In \Cref{sec:constant}, we combine this bound with the piece-norm
estimate to control the total number of block-encoding queries.

\section{Unitary synthesis using qDRIFT and QSVT}
\label{sec:constant}

We now combine qDRIFT and QSVT with the block encodings
$U_{H_k}$ from \Cref{sec:pieces} to implement $U$ to target
error $\epsilon$. These block encodings have subnormalization
factors $\alpha_k=mg$ and are used to approximate the evolution in
\[
U=e^{i\varphi}P^Te^{iH'}P,\qquad H'=\sum_kH_k.
\]

Discard zero pieces and set $\lambda=\sum_k\norm{H_k}$.
Since $\norm{e^{iH'}-I}\le\norm{H'}\le\lambda$, we can omit
the evolution when $\lambda<\epsilon/4$.
Otherwise, apply \Cref{lem:qdrift} at error $\epsilon/4$ to choose
the segment count $N_q$. Select a qDRIFT sample with piece
indices $K_j$ and durations $\tau_j$ whose product satisfies
\[
V=\prod_{j=N_q}^{1}e^{i\tau_jH_{K_j}},\qquad
\norm{V-e^{iH'}}\le\epsilon/4.
\]
Then $e^{i\varphi}P^TVP$ approximates $U$ with the same error
by unitary invariance.

To implement the selected product $V$, we simulate each segment
by QSVT using its block encoding $U_{H_{K_j}}$. Treating these
queries as exact, we allocate error $\epsilon/(4N_q)$ to each
segment, so the total simulation error is at most $\epsilon/4$.
\Cref{lem:hamsim} bounds the total query count by
\[
Q=O\!\left(\sum_{j=1}^{N_q}\alpha_{K_j}\tau_j
       +N_q\log\frac{4N_q}{\epsilon}\right).
\]
Let $c_{\rm query}$ be the bound on the $T$-gate cost per query
from \Cref{cor:normsum}. These queries contribute at most
$Qc_{\rm query}$ $T$ gates. \Cref{app:constant} accounts for
the remaining gate costs and errors from Clifford+$T$ compilation.

To bound the total query count, we control the sampled
product's weighted simulation time $\sum_j\alpha_{K_j}\tau_j$.
The sampling rule in \Cref{eq:qdrift-sampling}, applied to the
nonzero pieces, gives
\[
\E\!\left[\sum_{j=1}^{N_q}\alpha_{K_j}\tau_j\right]
=\sum_{k:\,H_k\ne0}\alpha_k.
\]
The following lemma combines this identity with Markov's inequality
and qDRIFT concentration to bound the weighted simulation time
and approximation error of the same sampled product. We apply it
to the $r$ nonzero pieces, reindexed as $A_1,\ldots,A_r$, with $\alpha_k=mg$.

\begin{lemma}[Sampled-product accuracy and query complexity]
\label{lem:constant-selection}
Let $A=\sum_{k=1}^r A_k$ with nonzero Hermitian $A_k\in\C^{N\times N}$,
and set $\lambda=\sum_k\norm{A_k}$.
For each $k$, suppose an exact block encoding $U_{A_k}$ of $A_k$ with
subnormalization factor $\alpha_k>0$ is given.
For $0<\varepsilon\le4\lambda$ and $p\in(0,1)$, choose $N_q$
using \Cref{lem:qdrift} with error $\varepsilon$ and failure probability
$p/2$. A qDRIFT sample with indices $K_j$, durations $\tau_j$,
and product $V=\prod_{j=N_q}^{1}e^{i\tau_jA_{K_j}}$ satisfies
\[
\norm{V-e^{iA}}\le\varepsilon,
\qquad \mathcal T:=\sum_j\alpha_{K_j}\tau_j\le\frac2p\sum_k\alpha_k
\]
with probability at least $1-p$.

For any sample satisfying both inequalities above
and any $\xi\in(0,1/2)$, QSVT implements $V$ to error
$\xi$ using a total of
\[
Q=O\!\left(\frac1p\sum_k\alpha_k+N_q\log(N_q/\xi)\right)
\]
queries to the block encodings $U_{A_k}$ in total
over all $N_q$ segments, counting their inverses and singly
controlled versions as well.
\end{lemma}

\begin{proof}
We defer the proof to \Cref{app:constant}.
\end{proof}

The following theorem bounds the $T$-count and
ancilla count as functions of $N$ and $\epsilon$.
We summarize the compiler in
\Cref{alg:constant} and the full procedure can be found in \Cref{alg:constant-full}.

\begin{theorem}[Unitary synthesis using qDRIFT and QSVT]
\label{thm:constant}
There is an absolute constant $C_0$ such that the following holds for
every $U\in U(N)$ and every
$\epsilon\in[N^{-1/8}n^{C_0},\,1/2]$. There is a Clifford+$T$ circuit that
implements $U$ to error $\epsilon$ with $\Ot(N^{8/7}\epsilon^{-6/7})$ $T$ gates
and $\Ot(N^{6/7}\epsilon^{-8/7})$ ancillas.
The hidden factors are $O(n^2L)$ and $O(nL)$ respectively, with
$L=n+\log_2(1/\epsilon)\le2n$.
\end{theorem}

\begin{proof}[Proof sketch]
\Cref{lem:constant-selection} combines qDRIFT concentration
with Markov's inequality to select a product that meets both
an accuracy test and a bound on weighted simulation time.

The parameter choice in \Cref{app:constant} has the scalings
\[
b=\Theta(N^{3/7}\epsilon^{-4/7}),\qquad
s=\Theta\!\left(\frac{\sqrt D\,\epsilon^2}{n^2L}\right).
\]
The stated range of $\epsilon$ ensures $s\ge1$.
With these choices, simulating all
$N_q$ sampled evolutions requires $Q=O(\sqrt D\,n)$ queries
in total to the block encodings $U_{H_k}$, including their
inverses and singly controlled versions (\Cref{eq:total-calls}).
Multiplying $Q$ by the $T$-gate cost of one such query from
\Cref{cor:normsum} gives the stated $T$-count. Reusing the
ancillas across these sequential queries gives the ancilla
bound. The full proof in \Cref{app:constant} verifies the
parameter conditions and gives the error analysis for qDRIFT,
QSVT, and the Clifford+$T$ implementation.
\end{proof}

\begin{algorithm}[tb]
\caption[qDRIFT and QSVT unitary synthesis outline]{Unitary synthesis using qDRIFT and QSVT for
\Cref{thm:constant}. Full procedure is in \Cref{alg:constant-full}.}
\label{alg:constant}
\begin{algorithmic}[1]
\Require $U\in U(N)$ and $\epsilon$ in the range of \Cref{thm:constant}.
\Ensure A fixed Clifford+$T$ circuit implementing $U$ to error $\epsilon$.
\State Choose $b,s,\ell$ as in \Cref{app:constant} and separate the logarithm into $\varphi I+H_0$. If $H_0=0$, return the scalar-phase circuit.
\State Resample the similarity and Hermitian split until the bounds of \Cref{lem:flattening,lem:split} hold. Discard zero pieces and set $\lambda=\sum_k\norm{H_k}$. If $\lambda<\epsilon/4$, return the scalar-phase circuit.
\State Encode the pieces and select a qDRIFT product passing the accuracy and weighted-time tests of \Cref{lem:constant-selection}.
\State Compile its evolutions with QSVT, exact basis changes, and the scalar phase, using reusable clean workspace and the implementation tolerances specified in the error analysis of \Cref{app:constant}.
\end{algorithmic}
\end{algorithm}

\section{Unitary synthesis using residual refinement}
\label{sec:refinement}

For a fixed decomposition and failure probability, the qDRIFT
segment bound in \Cref{lem:qdrift} grows as $\epsilon^{-2}$.
To improve this precision dependence, we successively correct
the remaining error, first using block-sparse generators and then
the full residual generator.

Let $W_{l-1}$ be the current unitary approximation to $U$, starting from $W_{-1}=I$. At block-sparse
level $l$, define the residual $R_l=UW_{l-1}^{\dagger}$.
Since $R_lW_{l-1}=U$, applying $R_l$ after $W_{l-1}$
would give the target exactly. We approximate this correction
to obtain $W_l$. As $W_{l-1}$ approaches $U$, $R_l$ approaches
the identity. Its principal Hermitian logarithm
$H^{(l)}=-i\log R_l$, whose eigenvalues lie in $(-\pi,\pi]$,
satisfies
\[
\norm{H^{(l)}}\le\frac\pi2\norm{R_l-I}
=\frac\pi2\norm{U-W_{l-1}}.
\]
Thus a more accurate approximation gives a smaller generator
for the next correction.

For each residual, we repeat the scalar-phase separation,
flattening, and splitting used in \Cref{sec:constant}. Let
$\varphi_l$ be its scalar phase, $P_l$ its change of basis,
and $H'^{(l)}=\sum_k H_k^{(l)}$ its transformed generator, so
$R_l=e^{i\varphi_l}P_l^Te^{iH'^{(l)}}P_l$.
For a residual requiring simulation, qDRIFT selects piece
indices $K_{l,u}$ and durations
$\tau_{l,u}$ for $u=1,\ldots,N_q^{(l)}$, using the unrounded
pieces $H_k^{(l)}$ (\Cref{lem:constant-selection}).

To simplify repeated block-encoding queries, we replace each
selected piece $H_k^{(l)}$ by a nearby Hermitian matrix
$\widetilde K_k^{(l)}$ obtained by rounding its encoding angles.
Using these rounded pieces, define
\[
\begin{aligned}
V_l&=\prod_{u=N_q^{(l)}}^{1}
e^{i\tau_{l,u}\widetilde K_{K_{l,u}}^{(l)}},\\
W_l&=e^{i\varphi_l}P_l^TV_lP_lW_{l-1}.
\end{aligned}
\]
Here $V_l$ approximates $e^{iH'^{(l)}}$, and $W_l$ appends the
resulting correction, including the basis changes and scalar phase,
to $W_{l-1}$.
The rightmost factor acts first. When a residual evolution is
omitted because its generator is zero or sufficiently small,
we set $V_l=I$, as specified in \Cref{alg:refinement-full}.
All choices of matrices, sample indices, and durations are made
during classical preprocessing.

We call $W_l$ the \emph{ideal product} because its exponentials,
basis changes, and scalar phases are evaluated exactly.
Its difference from $U$ includes qDRIFT and rounding errors,
which enter the next residual. Circuit implementation errors
are analyzed separately.

We next explain how rounding simplifies these queries. Round
each encoding's rotation and phase angles to the nearest multiples
of $2\pi/2^B$, where the integer $B$ sets the angle precision.
Given a fixed $(B+1)$-qubit state $\ket\Psi$, controlled modular
addition implements these phases exactly by phase kickback
(\Cref{lem:exactmux}). Each query returns $\ket\Psi$ unchanged
and unentangled with the other registers, so we call it a
\emph{catalyst} and reuse it throughout the circuit.
The following lemma gives the rounded Hermitian encoding and
bounds its matrix error and the $T$-gate and ancilla costs per query.

\begin{lemma}[Block encodings of rounded Hermitian matrices]
\label{lem:roundedencoding}
Let $K$ be Hermitian with blocks of norm at most $g$ and at most $m_0$
nonzero blocks in each block column, and let $m$ be the least power of
two with $m\ge m_0+1$. The extra slot accommodates diagonal-block
symmetrization. For $B\ge\log_2(2b)+2$, there is a Hermitian matrix
$\widetilde K$, determined by $K$, $g$, and $B$, with
\begin{equation}
\norm{\widetilde K-K}\le16\pi\,mg\,b\,2^{-B},
\label{eq:rounded-encoding-error}
\end{equation}
and an $(mg,\log_2m+1,0)$-block encoding $U_{\widetilde K}$
of $\widetilde K$. Queries to $U_{\widetilde K}$,
its inverse, and their singly controlled versions are implemented exactly, given
$\ket\Psi$, with the cost in \Cref{eq:encodingcost} at $\Lambda=B$.
\end{lemma}

\begin{proof}
We defer the proof to \Cref{app:roundedencoding}.
\end{proof}

The rounding-error bound in \Cref{lem:roundedencoding} scales
with the subnormalization factor $mg$. During corrections using
the full generator, $mg$ shrinks with the residual, so a fixed
angle precision gives progressively smaller absolute error bounds
in \Cref{eq:rounded-encoding-error}. \Cref{app:refinement-parameters}
specifies a choice $B=O(n)$ independent of the target precision
$L=n+\log_2(1/\epsilon)$.

With this fixed $B$, all queries can share the same catalyst.
Its product form is
\[
\ket\Psi=\bigotimes_{j=1}^{B+1}
\frac{\ket0+e^{-2\pi i/2^j}\ket1}{\sqrt2},
\]
with qubits ordered from most to least significant. This
factorization follows from the binary expansion of the index in
\Cref{app:exactmux}, and allows each qubit to be prepared separately.

Single-qubit synthesis (\Cref{eq:single-qubit-synthesis}) therefore
prepares the catalyst to error $\delta_c$ using
$O(B\log(B/\delta_c))$ $T$ gates. Each query preserves the exact
catalyst, so we prepare it once before the first query and unprepare
it after the last. Comparing the complete circuit with an exactly
initialized catalyst bounds the combined error by $2\delta_c$,
independently of the number of queries. Choosing $\delta_c=\epsilon/16$
gives error at most $\epsilon/8$ and total preparation and
unpreparation cost $O(nL)$.

These block-encoding queries are then used within QSVT to
approximate the chosen evolutions. QSVT truncation can leave
ancillas nonzero, so we bound its full-output error separately
from the ideal-product error (\Cref{sec:refinement-errors}).
The following theorem accounts for these implementation costs and
errors together with the query costs of both refinement stages.

\begin{theorem}[Unitary synthesis using residual refinement]
\label{thm:refined}
There is an absolute constant $n_0$ such that for every $n\ge n_0$,
every $U\in U(N)$, and every $\epsilon\in(0,1/2)$ with
$L=n+\log_2(1/\epsilon)\le N$, there is a Clifford+$T$ circuit that
implements $U$ to error $\epsilon$ with
\begin{equation}
O\big(n^3(N^{8/7}L^{3/7}+NL)\big)
\label{eq:refinement-total-cost}
\end{equation}
$T$ gates and $O(nN)$ ancillas.
\end{theorem}

\begin{proof}[Proof sketch]
\Cref{lem:roundedencoding} gives Hermitian rounded generators
with a bound on the rounding error and encodings implemented
exactly given a shared catalyst state.
We form each residual from the corresponding ideal product, so
later corrections address both qDRIFT error and data rounding.
The block-sparse levels use geometrically decreasing error
tolerances. As the generator shrinks, we reduce the number of pieces
to control the additive precision cost per segment. The costs of
these levels then sum geometrically.

When the piece count would fall below one,
we correct each remaining residual using a block encoding of
its full transformed and rounded generator $\widetilde G_j$,
where $j$ indexes these corrections. QSVT simulates
$e^{i\widetilde G_j}$ by querying that block encoding.
The subnormalization factors of these block encodings decrease by at
least a factor $N^2$ per correction, and the short-time simulation
bound controls their query complexity. The two stages give the stated
$T$-count with reusable workspace. Simulation truncation, rotation
synthesis, catalyst preparation and unpreparation, and the phase
are treated in a separate error analysis. We give
the full proof in \Cref{app:refinement}.
\end{proof}

For fixed $n$, the $T$-count bound grows at most linearly with the
requested number of accuracy bits. This improves the
$\epsilon^{-6/7}$ dependence in \Cref{thm:constant}.
In particular, for $\epsilon\ge2^{-\poly(n)}$, the number of accuracy
bits is polynomial in $n$, giving a $T$-count of
$2^{8n/7}\poly(n)$ for sufficiently large $n$.
When $N^{1/4}\le L\le N$, the $NL$ term in
\Cref{eq:refinement-total-cost} dominates, so the bound simplifies
to $O(n^3NL)$. The ancilla bound remains $O(nN)$ throughout
$L\le N$. \Cref{sec:refinement-cost} compares this precision
dependence with previous bounds. The compiler is summarized in
\Cref{alg:refinement}, with the full procedure in
\Cref{alg:refinement-full}.

\begin{algorithm}[tb]
\caption[Residual-refinement outline]{Residual refinement for
\Cref{thm:refined}. Full procedure is in \Cref{alg:refinement-full}.}
\label{alg:refinement}
\begin{algorithmic}[1]
\Require $U\in U(N)$ and $\epsilon$ in the range of \Cref{thm:refined}.
\Ensure A fixed Clifford+$T$ circuit implementing $U$ to error $\epsilon$.
\State Choose the parameters and refinement schedule as in \Cref{app:refinement-parameters}. Set the current ideal product $W=I$.
\State For each block-sparse level, select a qDRIFT correction for $UW^\dagger$, round its generators, and left-multiply $W$ by the resulting ideal correction, including basis changes and the scalar phase.
\State Repeatedly correct $UW^\dagger$ using its full rounded generator and update $W$ by the ideal correction until the latest generator-rounding error is at most $\epsilon/8$.
\State Compile the recorded evolutions with QSVT, exact basis changes, and the accumulated scalar phase, preparing and unpreparing one shared catalyst $\ket\Psi$. Use reusable clean workspace and the tolerances in \Cref{sec:refinement-errors}.
\end{algorithmic}
\end{algorithm}

\section{Limits of the sequential block-SELECT model}
\label{sec:optimality}

We now ask whether a different choice of pieces or block size can
improve the $N^{8/7}$ scaling of the $T$-count in
\Cref{thm:constant} when the target error $\epsilon$ is fixed
independently of $N$. The following theorem establishes a matching
lower bound within the sequential block-SELECT model, subject to a
bound on the information carried by the intervening unitaries.

The model represents the target by Hamiltonian evolutions separated
by unitaries $W_j$, called interleavers. Durations are absorbed into
the generators, and $s$ counts exponential factors. Here $s$ and $W_j$
have local meanings, distinct from the split count and residual
approximations used earlier.

\begin{definition}[Sequential block-SELECT design]
\label{def:design}
A design at block size $b$ is a sequence
$\mathcal D=(W_0,H_1,W_1,\dots,H_s,W_s)$ of unitaries $W_j$ and Hermitian
matrices $H_k$. It approximates $U$ to error $\epsilon$ if
\begin{equation}
\norm{e^{i\theta}W_se^{iH_s}\cdots W_1e^{iH_1}W_0-U}\le\epsilon
\label{eq:design-approximation}
\end{equation}
for some $\theta\in\R$.

The Hermitian generators $H_k$ may depend
arbitrarily on the target $U$. Let $n_k$ count the nonzero blocks of
$H_k$, and let $m_k$ be their maximum count in a block row or column
before padding. Set
$c_k=\max_J(\sum_I\norm{(H_k)_{IJ}}^2)^{1/2}$ and $\beta_k=\sqrt{m_k}\,c_k$.
For the sequence $\mathcal D$, we define the
following model of the $T$-gate cost for implementing its
Hamiltonian evolutions $e^{iH_1},\ldots,e^{iH_s}$:
\begin{equation}
\mathrm{cost}(\mathcal D)=\sum_{k=1}^s\max\{1,\beta_k\}\big(b\sqrt{n_k}+b^2\big).
\label{eq:designcost}
\end{equation}
The interleavers have zero model cost. To prevent them from storing
the entire target unitary, we require the entire sequence $(W_0,\dots,W_s)$
to be specified by at most $B_W$ bits. These bits may depend on $U$.
A fixed decoding rule recovers the sequence for every target.
\end{definition}

The model assigns a lookup charge $b\sqrt{n_k}$ and a block-processing
charge $b^2$ to each query to a block encoding of $H_k$. The multiplier
$\max\{1,\beta_k\}$ represents the simulation work for $e^{iH_k}$,
with a minimum charge of one and linear growth for $\beta_k>1$,
motivated by the black-box query bounds of Ref.~\cite{low2019spectral}.
The interleaver budget allows at most $2^{B_W}$ sequences under
the fixed decoding rule. \Cref{thm:optimality} bounds the resulting
model cost in \Cref{eq:designcost}.

\begin{theorem}[Optimality within the sequential block-SELECT model]
\label{thm:optimality}
There are absolute constants $c,C$ with the following property. %
Let $n$ be
large enough and $\epsilon\in[N^{-3},c]$. Then, with probability at least
$1/2$ over a Haar-random $U\in U(N)$, the following holds simultaneously
for every power-of-two block size $b\le N$, with $D=N/b$.
Every design at block size $b$ with
$B_W\le\frac14N^2\log_2(1/\epsilon)$ that approximates $U$ to error $\epsilon$
satisfies
\begin{equation}
\mathrm{cost}(\mathcal D)\ \ge\ \frac{c}{n^{C}}\Big(\sqrt\epsilon\,ND^{1/4}+\epsilon\sqrt D\,b^2\Big).
\label{eq:lowerbound}
\end{equation}
Minimizing over $b$, every such design has cost
$\Omt(\epsilon^{4/7}N^{8/7})$.
\end{theorem}

\begin{proof}[Proof sketch]
Szarek's covering bound~\cite{szarek1982nets} shows that at least half
of Haar-random targets require $\Omega(N^2\log(1/\epsilon))$ bits to be
approximated within error $O(\epsilon)$ under a fixed decoding
rule. The interleaver budget leaves a fixed fraction of this
information to the generators.

We compress a low-cost design by storing each generator's blocks or a
shorter sampled approximation (\Cref{lem:sparsification}). Independent
sampling controls the product error through the sum of squared error
tolerances (\Cref{lem:sparsification-composition}). After rounding,
each stored block uses $O(b^2\log N)$ bits, and the total target error
remains $O(\epsilon)$. The information bound requires
$\Omega(D^2/\log N)$ stored blocks (\Cref{lem:compression-information}).
Balancing direct storage against block sampling, whose required sample
count grows with $\beta_k$, gives the lookup and block-processing
contributions to \Cref{eq:lowerbound}. Recording $b$ in the same code
covers all allowed block sizes. The full proof is in \Cref{app:optimality}.
\end{proof}

At fixed target error $\epsilon$,
\Cref{thm:optimality} matches, up to factors polynomial in $n$,
the model cost of the qDRIFT and QSVT construction in
\Cref{thm:constant}. The $j$th model generator is
$\tau_jH_{K_j}$, with $\beta_j\le\tau_j\alpha_{K_j}$.
The model cost is $\Ot(N^{8/7})$ at
$b=N^{3/7}$, up to power-of-two rounding.
\Cref{app:mixed} extends the exponent bound to mixed block sizes
and larger interleaver budgets.

\section{Discussion}
\label{sec:discussion}

We have constructed Clifford+$T$ circuits for arbitrary $n$-qubit
unitaries using $\Ot(N^{8/7})$ $T$ gates and $\Ot(N^{6/7})$ clean
ancillas at constant operator-norm error, where $N=2^n$.
The construction combines random similarity flattening, a Hermitian
block-sparse split, qDRIFT, and QSVT to reduce the lookup cost per
query while controlling the total simulation work.
Residual refinement gives $\Ot(N^{8/7}L^{3/7}+NL)$ $T$ gates and
$\Ot(N)$ clean ancillas for target error $\epsilon\in(0,1/2)$ with
$L=n+\log_2(1/\epsilon)\le N$.
Within the sequential block-SELECT cost model, the exponent $8/7$
is optimal at constant precision under its interleaver information
restriction.

For synthesis using qDRIFT and QSVT at constant precision, the
$T$-gate and ancilla bounds have product $\Ot(N^2)$,
matching Tan's space--$T$-count tradeoff
$(n+m_{\rm anc})T=\Omega(N^2)$~\cite{tan2025unitary} up to factors
polynomial in $n$. Compared with Tan's construction, we use fewer
$T$ gates but more ancillas. This moves the costs closer to the
balanced scale $T\approx m_{\rm anc}\approx N$, though our upper
bound remains above the general counting lower bound.

Oracle-based synthesis uses a different cost accounting.
OpenAI's recent draft~\cite{openai2026unitary} establishes unitary
synthesis at constant diamond-norm error using polynomially many
elementary gates, qubits, and queries to a suitably chosen Boolean
oracle. The target-dependent Boolean function has unrestricted
truth-table size and circuit complexity, and the polynomial resource
bound excludes the cost of implementing its oracle calls. This
distinction matters when comparing with the Clifford+$T$ constructions
of Tan~\cite{tan2025unitary}, Yuan et al.~\cite{yuan2026general}, and
this work. These constructions include the cost of quantum data
access and achieve $T$-counts $\Ot(N^{4/3})$, $\Ot(N^{5/4})$, and
$\Ot(N^{8/7})$, respectively. In the OpenAI theorem,
$p(n)$ bounds all resources, whereas $m(n)\le p(n)$ denotes the oracle
input length. A generic implementation uses
$O(2^{m(n)/2})$ $T$ gates and ancillas per query, yielding a total
$T$-count bounded by $2^{m(n)/2}\poly(n)$. A polynomial bound on
$m(n)$ alone therefore does not determine an exponential scaling
comparable with these synthesis bounds. For example,
$m(n)\le 2n+O(\log n)$ would suffice to reach $N\poly(n)$ through
generic oracle compilation, but the draft does not establish this
bound or a comparably efficient implementation of its particular
oracle. A quantitative comparison consequently requires accounting
for the oracle's implementation cost and ancillary workspace.

Comparisons at finite size likewise require compiled resource
estimates, because our bounds are asymptotic.
Liu et al.~\cite{liu2026logical} provide such estimates for state
preparation, comparing rotation- and sampling-based methods after
compilation in both $T$-count and total gate count. For our
construction, compiled estimates must account for its large lookup
tables and substantial workspace. When the data have structure,
compressed tables can reduce data-loading costs, as compact alias
tables do for real-amplitude state preparation with polynomial
amplitudes~\cite{liu2026compressed}.

At constant target error, a polynomial gap remains between our
$\Ot(N^{8/7})$ upper bound on $T$-count and the general counting
lower bound $\Omt(N)$. It remains open whether general
Clifford+$T$ circuits can achieve a smaller exponent and ultimately
reach the counting scale.
Our matching $N^{8/7}$ lower bound relies on a counting argument
within the sequential block-SELECT cost model under its interleaver
information restriction. A complementary question is whether a
worst-case lower bound of $\Omega(N^{1+c})$ $T$ gates, for some fixed
$c>0$, can be proved for general circuits with an unrestricted number
of clean ancillas. Workspace raises a further question at higher
precision, where residual refinement uses $O(nN)$ clean ancillas
throughout $L\le N$, independently of $\epsilon$. How much of its
$T$-count improvement can be retained with fewer ancillas?

Similar questions arise beyond qubits. Approximate single-qutrit
Clifford+$R$ synthesis has been studied~\cite{gustafson2025qutrit}.
For arbitrary unitaries on many qudits, the optimal non-Clifford gate
count at constant error remains open.

\begin{acknowledgments}
The authors thank Nathan Baker, Daan Camps, and Yu Tong for valuable
discussions. D.L. and W.A.d.J. were supported by the U.S. Department of
Energy, Office of Science, Office of Advanced Scientific Computing
Research through the Accelerated Research in Quantum Computing
program MACH-Q Project. S.Z., D.F. and C.Y. were supported by the U.S. Department of
Energy, Office of Science, Office of Advanced Scientific Computing
Research through Quantum Utility through Advanced Computational
Quantum Algorithms and the Quantum System Accelerator projects.

We acknowledge the use of GPT-5.5 Pro, GPT 5.6 Sol, GPT-6 Astra to search for algorithm ideas, explore candidate
proof strategies and to assist with initial manuscript preparation.
\end{acknowledgments}

\appendix
\crefalias{section}{appendix}
\crefalias{subsection}{appendix}

\section{Simulation and subnormalization bounds}
\label{app:tools}

\subsection{Single-realization qDRIFT}
\label{app:qdrift}

\begin{proof}[Proof of \Cref{lem:qdrift}]
We bound the bias of the average product and the fluctuations around
that average separately. Choose $N_q$ as the least positive integer satisfying
\begin{equation}
N_q\ge\max\Big\{\frac{44\lambda^2}{\epsilon^2}\log\frac{2d}\delta,\ \frac{2\lambda^2}{\epsilon}\Big\}.
\label{eq:qdrift-concentration-count}
\end{equation}
The two terms control fluctuations and bias, respectively.
Draw independent indices $K_1,\dots,K_{N_q}$ with the same
distribution and define the segments by
\begin{align*}
\Pr(K_j=k)&=\norm{A_k}/\lambda,\\
V_j&=\exp\big(i\tfrac{\lambda}{N_q}A_{K_j}/\norm{A_{K_j}}\big).
\end{align*}
Set $V=V_{N_q}\cdots V_1$. Since $\epsilon\le4\lambda$,
the first term in the threshold is bounded below by a positive
constant. Rounding up therefore preserves the stated asymptotic bound.

We apply the estimates of Ref.~\cite{chen2021concentration} with
$H=-A$ and $t=1$, so their evolution $e^{-itH}$ is $e^{iA}$.
Proposition~3.2 bounds the bias:
\[
\norm{e^{iA}-\E V}\le\lambda^2/N_q\le\epsilon/2.
\]
The last inequality is the second condition on $N_q$.

Proposition~3.3 bounds fluctuations about the average:
\[
\Pr\big(\norm{V-\E V}\ge\epsilon/2\big)
\le2d\,e^{-N_q\epsilon^2/(44\lambda^2)}\le\delta.
\]
Here the last inequality is the first condition on $N_q$. Outside
this exceptional event, the triangle inequality gives
$\norm{V-e^{iA}}\le\epsilon$.
\end{proof}

\subsection{Subnormalization factors for block encodings}
\label{app:normalization}

We consider an exact block encoding of
$A\in\C^{N\times N}$ constructed in three steps.
The data registers $\mathsf B$ and $\mathsf T$ hold the block label
and the position within a block. The ancillary registers
$\mathsf R$ and $\mathsf F$ hold an additional block label and a
flag qubit. We write the registers in the order
$\mathsf R,\mathsf B,\mathsf F,\mathsf T$ and start with
\[
\ket0_{\mathsf R}\ket J_{\mathsf B}
\ket0_{\mathsf F}\ket\psi_{\mathsf T},
\qquad \ket\psi\in\C^b.
\]
We first describe preparations without padding.
\begin{enumerate}
\item A preparation $P_w$, controlled on the input block label $J$,
  prepares
  \[
  P_w^{(J)}\ket0_{\mathsf R}
  =\sum_I w_{IJ}\ket I_{\mathsf R}.
  \]
  The resulting state is
  \[
  \sum_I w_{IJ}\ket I_{\mathsf R}\ket J_{\mathsf B}
  \ket0_{\mathsf F}\ket\psi_{\mathsf T}.
  \]
\item A $\SEL$ oracle applies $C_{IJ}$ to $\mathsf F,\mathsf T$
  on address $(I,J)$, giving
  \[
  \sum_I w_{IJ}\ket I_{\mathsf R}\ket J_{\mathsf B}
  C_{IJ}\bigl(\ket0_{\mathsf F}\ket\psi_{\mathsf T}\bigr).
  \]
  Swapping the two label registers then gives
  \[
  \sum_I w_{IJ}\ket J_{\mathsf R}\ket I_{\mathsf B}
  C_{IJ}\bigl(\ket0_{\mathsf F}\ket\psi_{\mathsf T}\bigr).
  \]
\item A second preparation $P_v$, controlled on the output block
  label $I$, is defined by
  \[
  P_v^{(I)}\ket0_{\mathsf R}
  =\sum_{J'}v_{IJ'}\ket{J'}_{\mathsf R}.
  \]
  Applying its inverse gives the full output state
  \begin{equation*}
  \begin{split}
  \ket{\mathrm{out}_J}
  =\sum_I w_{IJ}&\bigl(P_v^{(I)\dagger}\ket J_{\mathsf R}\bigr)
  \ket I_{\mathsf B}\\
  &{}\otimes C_{IJ}\bigl(\ket0_{\mathsf F}\ket\psi_{\mathsf T}\bigr).
  \end{split}
  \end{equation*}
  Its overlap with the zero state of $\mathsf R$ is determined by
  \[
  \bra0_{\mathsf R}P_v^{(I)\dagger}\ket J_{\mathsf R}
  =\overline{v_{IJ}}.
  \]
\end{enumerate}
Define the encoded block of $C_{IJ}$ by
\[
(C_{IJ})_{00}
=(\bra0_{\mathsf F}\otimes I_{\mathsf T})C_{IJ}
(\ket0_{\mathsf F}\otimes I_{\mathsf T}).
\]
The component of the output with both ancillary registers zero is
\begin{equation}
\begin{split}
&(\bra0_{\mathsf R}\otimes I_{\mathsf B}
  \otimes\bra0_{\mathsf F}\otimes I_{\mathsf T})
  \ket{\mathrm{out}_J}\\
&\qquad=\sum_I\overline{v_{IJ}}w_{IJ}\ket I_{\mathsf B}
  (C_{IJ})_{00}\ket\psi_{\mathsf T}.
\end{split}
\label{eq:prepared-block-output}
\end{equation}
For a block encoding of $A$ with subnormalization factor $\alpha$,
this component must equal
$\sum_I\ket I_{\mathsf B}(A_{IJ}/\alpha)\ket\psi_{\mathsf T}$.
Comparing the coefficients for each $I$ and every $\ket\psi$ gives
the block identity below.
\begingroup
The full preparation states are normalized. Padding may add
orthogonal branches whose selected unitaries have zero encoded blocks.
These branches contribute nothing to
\Cref{eq:prepared-block-output}, so the block identity remains valid,
while the squared amplitudes on physical addresses may sum to less
than one. Thus
\begin{equation}
\begin{gathered}
\frac{A_{IJ}}{\alpha}
=\overline{v_{IJ}}\,w_{IJ}\,(C_{IJ})_{00},\\
\sum_I\abs{w_{IJ}}^2\le1,
\qquad\sum_J\abs{v_{IJ}}^2\le1.
\end{gathered}
\label{eq:prepared-block-encoding}
\end{equation}
The first amplitude sum is for each fixed $J$ and the second for
each fixed $I$. Since $C_{IJ}$ is unitary,
$\norm{(C_{IJ})_{00}}\le1$.
\par\endgroup
The amplitudes may depend on $A$ in an arbitrary way. The construction
of Ref.~\cite{yuan2026general} and \Cref{lem:encoding} are of this form.
In the latter, $A=K$. A nonpadding slot address $(J,u)$ identifies
the block pair $(I,J)$: slot $u$ in column $J$ contains $K_{IJ}$.
Let $u'$ be the slot in column $I$ containing its adjoint
$K_{JI}=K_{IJ}^\dagger$. The address map $q(J,u)=(I,u')$ then gives
\[
\Pi\ket u_{\mathsf S}\ket J_{\mathsf B}
=\ket{u'}_{\mathsf S}\ket I_{\mathsf B},
\]
where $\mathsf S$ is the slot register. Thus $J\mapsto I$ routes
$K_{IJ}$ to output block $I$, while $u\mapsto u'$ supplies the
slot address used by the output unpreparation. This is the role
of the explicit label swap above. Uniform slot preparations give
$w_{IJ}=v_{IJ}=1/\sqrt m$ on nonzero blocks. Together with
$(C_{IJ})_{00}=K_{IJ}/g$, \Cref{eq:prepared-block-encoding} therefore
gives the encoded block $K_{IJ}/(mg)$.

\begin{lemma}[Lower bound on the subnormalization factor]
\label{lem:normalization}
An exact block encoding of $A$ satisfying
\Cref{eq:prepared-block-encoding}, with subnormalization factor $\alpha$, satisfies
$\alpha\ge\frac1D\sum_{I,J}\norm{A_{IJ}}$.
\end{lemma}

\begin{proof}
\begingroup
The block identity in \Cref{eq:prepared-block-encoding} gives
$\norm{A_{IJ}}\le\alpha\abs{v_{IJ}}\abs{w_{IJ}}$.
Summing the preparation-amplitude bounds over the controlled
labels bounds each total squared-amplitude sum by $D$.
Cauchy--Schwarz therefore gives
\begin{equation}
\begin{split}
\sum_{IJ}\norm{A_{IJ}}&\le\alpha\sum_{IJ}\abs{v_{IJ}}\abs{w_{IJ}}\\
&\le\alpha\Big(\sum_{IJ}\abs{v_{IJ}}^2\Big)^{1/2}\Big(\sum_{IJ}\abs{w_{IJ}}^2\Big)^{1/2}\\
&\le\alpha D .
\end{split}
\label{eq:normalization-cauchy-schwarz}
\end{equation}
\endgroup
\end{proof}

\paragraph*{Cost of the uniform dense block encoding.}
To justify the cost in \Cref{sec:overview}, we derive the
subnormalization factor of the block encoding of the transformed
target unitary $U'$ and the resulting number of
queries needed for oblivious amplitude amplification. We then apply
\Cref{lem:normalization} to Haar-random targets to show that the
$\sqrt D$ scale is necessary for block encodings of the form
\Cref{eq:prepared-block-encoding}.

Yuan et al.~\cite{yuan2026general} use random input and output bases
to obtain $\norm{U'_{IJ}}\le g=\Ot(1/\sqrt D)$ with high probability.
Embed each contraction $U'_{IJ}/g$ as the encoded block of a
$2b$-dimensional unitary $C_{IJ}$. Uniform preparation and
unpreparation, $w_{IJ}=v_{IJ}=1/\sqrt D$, each supply one
factor $1/\sqrt D$, giving
\begin{equation}
\frac{1}{\sqrt D}\,\frac{U'_{IJ}}{g}\,\frac{1}{\sqrt D}
=\frac{U'_{IJ}}{Dg}.
\label{eq:dense-block-scaling}
\end{equation}
Thus the full encoded matrix is $U'/\alpha$, where
$\alpha=Dg=\Ot(\sqrt D)$. The subnormalization factor
determines how many queries to this full block encoding of $U'$
are needed for oblivious amplitude amplification. For any normalized input, the
zero-ancilla component has squared norm $p=1/\alpha^2$ because $U'$
is unitary. Oblivious amplitude amplification, implemented through
QSVT~\cite{gilyen2019quantum}, uses the block encoding of $U'$, its inverse, and zero-ancilla phase
shifts; $O(1/\sqrt p)=O(\alpha)=\Ot(\sqrt{N/b})$ queries suffice at
fixed accuracy. Selecting among $D^2$ unitaries of dimension $2b$
costs $\Ot(N+b^2)$ $T$ gates per query by \Cref{lem:tan}, giving
\Cref{eq:dense-encoding-cost}.

For a Haar-random $U$ and $b\ge C\sqrt{\log N}$, concentration of measure on
$U(N)$~\cite{meckes2019random} shows that, with high probability, every
block satisfies $\norm{U_{IJ}}_F\ge b/(2\sqrt N)$. Comparing the
Frobenius and operator norms gives
\[
\norm{U_{IJ}}\ge\norm{U_{IJ}}_F/\sqrt b\ge\frac12\sqrt{b/N}.
\]
There are $D^2$ blocks, so $\sum_{IJ}\norm{U_{IJ}}\ge D^{3/2}/2$.
\Cref{lem:normalization} then gives $\alpha\ge\sqrt D/2$.
This gives the scale of the subnormalization factor used in \Cref{sec:overview}.

For any additive split $H'=\sum_kH_k$ whose pieces $H_k$ are
block encoded in this way,
apply \Cref{lem:normalization} to each piece and then use the triangle
inequality: $\sum_k\alpha_k\ge\sum_{IJ}\norm{H'_{IJ}}/D$. Splitting
reduces the data addressed per query, and the sum of subnormalization factors obeys
the same lower bound.

\section{Matrix preprocessing}
\label{app:preprocessing}

\subsection{Similarity flattening}
\label{app:flattening}

We bound the diagonal contribution by Hanson--Wright and the
off-diagonal contribution by decoupling and matrix Rademacher bounds.
The required concentration inequalities are as follows.

\begin{fact}[Hanson--Wright~\cite{rudelson2013hanson}]
\label{fact:hw}
Let $\sigma$ be uniform on $\{\pm1\}^N$ and $A\in\R^{N\times N}$. For all $t\ge0$,
$\Pr(\abs{\sigma^TA\sigma-\tr A}>t)\le2\exp\big(-c\min\{t^2/\norm A_F^2,\,t/\norm A\}\big)$.
\end{fact}

\begin{fact}[Matrix Rademacher series~\cite{tropp2015introduction}]
\label{fact:rademacher}
Let $B_y\in\C^{d_1\times d_2}$ be fixed and $\tau_y$ independent Rademacher
signs. Put
$v=\max\{\norm{\sum_yB_yB_y^\dagger},\norm{\sum_yB_y^\dagger B_y}\}$. Then
$\Pr(\norm{\sum_y\tau_yB_y}\ge t)\le(d_1+d_2)e^{-t^2/(2v)}$.
\end{fact}

\begin{fact}[Decoupling~\cite{delapena1995decoupling,delapena1999decoupling}]
\label{fact:decoupling}
Let $\tau$ be a vector of independent random variables and $\tau'$ an
independent copy of it. Let $f_{yz}$ be measurable maps into a separable
Banach space. There is an absolute constant $C_{\rm dec}$ such that for all $t>0$,
$\Pr\big(\norm{\sum_{y\ne z}f_{yz}(\tau_y,\tau_z)}>t\big)\le C_{\rm dec}\Pr\big(C_{\rm dec}\norm{\sum_{y\ne z}f_{yz}(\tau_y,\tau'_z)}>t\big)$.
\end{fact}

\Cref{fact:rademacher} is Theorem~4.1.1 of
Ref.~\cite{tropp2015introduction}; for \Cref{fact:decoupling}, see
also Theorem~3.4.1 of Ref.~\cite{delapena1999decoupling}.

\begin{proof}[Proof of \Cref{lem:flattening}]
$P$ is real orthogonal, so $e^{iPHP^T}=Pe^{iH}P^T$.
The applications of $F=\mathsf H^{\otimes n}$
use only Clifford gates, and the two sign layers cost $O(\sqrt N)$
$T$ gates by \Cref{lem:boolean}, for either $P$ or $P^T$.
For the norm bound, the case $H=0$ is straightforward. Otherwise, assume $\norm H=1$ without loss of generality, as we can rescale the matrix $H$. Put $G=FHF$ and
$K=FS_\sigma GS_\sigma F$, so $PHP^T=FS_\tau KS_\tau F$.
We separate its diagonal part $\diag(K)$ from its remaining
entries; the second sign layer leaves $\diag(K)$ unchanged.

\emph{Diagonal part.} Each $K_{yy}$ has the distribution of
$N^{-1}\sigma^T\mathrm{Re}(G)\sigma$,
since every entry of $F$ is $\pm N^{-1/2}$.
The real symmetric matrix $\mathrm{Re}(G)$ has
trace zero, operator norm at most $1$, and Frobenius norm at most
$\sqrt N$. Set $\ell=\log(N/\delta)$. By \Cref{fact:hw} and
$\norm{\mathrm{Re}(G)}_F^2\le N$, with failure probability at most $\delta/2$,
\begin{equation}
\norm{\diag(K)}
\le C\left(\sqrt{\frac{\ell}{N}}+\frac{\ell}{N}\right)
\le \frac{2C\ell}{\sqrt N}.
\label{eq:diagonal-tail}
\end{equation}
Every block of $F\diag(K)F$ has norm at most this bound.

\emph{Off-diagonal part.} Condition on $\sigma$ and set
$A=K-\diag(K)$. Then $\norm A\le2$ and
$\sum_z\abs{A_{yz}}^2\le1$ for every $y$.
Let $R_I$ contain the $b$ rows of $F$ in block $I$, with columns
$r_y^I$ for $y\in[N]$. Thus $R_IR_I^T=I_b$ and $\norm{r_y^I}^2=b/N$.
Fix $I,J$ and write $X=R_IS_\tau AS_\tau R_J^T$.
To use \Cref{fact:decoupling}, introduce an independent copy
$\tau'$ and write the decoupled block as
\[
V=AS_{\tau'}R_J^T,\qquad
Y=R_IS_\tau V=\sum_y\tau_y B_y,
\]
where $v_y^\dagger$ is the $y$-th row of $V$ and
$B_y=r_y^I v_y^\dagger$.

Put $\eta=\delta/(2C_{\rm dec}D^2)$. Each $v_y$ is a vector
Rademacher series with variance parameter at most
$(b/N)\sum_z\abs{A_{yz}}^2\le b/N$.
By \Cref{fact:rademacher} and a union bound over $y$,
$\max_y\norm{v_y}^2\le C(b/N)\ell$ except with probability
$\eta/2$; here $\log(N(b+1)/\eta)=O(\ell)$.
For any such $\tau'$, the conditional matrix variance is controlled by
\begin{equation}
\begin{aligned}
\sum_y B_yB_y^\dagger
&\preceq \max_y\norm{v_y}^2\,I_b,\\
\sum_y B_y^\dagger B_y
&=\frac bN V^\dagger V\preceq\frac{4b}{N}I_b,
\end{aligned}
\label{eq:decoupled-variance}
\end{equation}
where the last inequality uses $\norm V\le\norm A\le2$.
A second application of \Cref{fact:rademacher}, now to $\tau$,
gives $\norm Y\le C\sqrt{b/N}\,\ell$ except with an additional
probability $\eta/2$. Including the bad choices of $\tau'$, the
failure probability is at most $\eta$.

Decoupling therefore gives
$\Pr(\norm X>C\sqrt{b/N}\,\ell)\le C_{\rm dec}\eta$,
absorbing its constant factor into $C$. A union bound over
the $D^2$ block pairs yields, with failure probability at most $\delta/2$,
\begin{equation}
\max_{I,J}\norm{(FS_\tau AS_\tau F)_{IJ}}
\le C\sqrt{\frac bN}\,\ell .
\label{eq:offdiagonal-bound}
\end{equation}
This estimate holds for every $\sigma$. Combining the two parts,
using $b\ge1$, and restoring $\norm H$ proves \Cref{eq:flatness}
with total failure probability at most $\delta$.
\end{proof}

\subsection{Random block-sparse split}
\label{app:split}

\begin{proof}[Proof of \Cref{lem:split}]
Let $B_p$ retain one adjoint block pair of $H'$ (one block for a
diagonal pair), with zero elsewhere. Then $H'=\sum_pB_p$,
each $B_p$ is Hermitian, and $\norm{B_p}\le g$.
For a fixed class $k$,
$H_k-H'/s=\sum_p(\xi_p-1/s)B_p$, where the assignment indicators
$\xi_p$ are independent Bernoulli$(1/s)$ variables.
Each summand has norm at most $g$, and
\[
\left \| \sum_p\E(\xi_p-1/s)^2B_p^2 \right \|
\le\frac{1}{s}\left \|\bigoplus_I(H'^2)_{II}\right \|
\le\frac{\norm{H'}^2}{s}.
\]
Here $\sum_pB_p^2=\bigoplus_I(H'^2)_{II}$, whose blocks are
compressions of $H'^2$.

Recall that the matrix Bernstein inequality~\cite{tropp2012user}
applies to independent, mean-zero Hermitian matrices
$Z_p\in\C^{N\times N}$ satisfying $\norm{Z_p}\le R$ almost surely.
With variance parameter $v=\norm{\sum_p\E Z_p^2}$, it gives, for
every $t>0$,
\[
\Pr\!\left(\left\|\sum_p Z_p\right\|>t\right)
\le 2N\exp\!\left(-\frac{t^2}{2(v+Rt/3)}\right).
\]
For $u>0$, the choice $t=\sqrt{2vu}+\tfrac23Ru$ satisfies
$t^2\ge2u(v+Rt/3)$, and hence
\[
\Pr\!\left(\left\|\sum_p Z_p\right\|>
\sqrt{2vu}+\tfrac23Ru\right)\le2Ne^{-u}.
\]
We apply this bound with $Z_p=(\xi_p-1/s)B_p$.
For the fixed matrix $H'$ and class $k$, these matrices are
independent and Hermitian, with $\E Z_p=0$ and
$\norm{Z_p}\le\norm{B_p}\le g$.
The variance estimate above allows us to use $R=g$ and
$v\le\norm{H'}^2/s$. Taking $u=\ell$ and using
$\sum_p Z_p=H_k-H'/s$ therefore yields
\[
\Pr\!\left(\norm{H_k-H'/s}>
\norm{H'}\sqrt{2\ell/s}+\tfrac23g\ell\right)\le2Ne^{-\ell}.
\]
The triangle inequality yields \Cref{eq:split-piece-norm}.

The $D$ potential blocks in a fixed row belong to distinct pairs,
so its occupancy is at most a Binomial$(D,1/s)$ variable.
The scalar Chernoff bound gives occupancy at most
$m_0=D/s+\sqrt{2D\ell/s}+\ell$ except with probability $e^{-\ell}$.
Hermiticity gives the same bound for columns. A union bound over
$s$ norm events and $sD$ row events costs at most
$(2Ns+sD)e^{-\ell}\le\delta$.

Finally, $\ell\ge1$ and the assumed bound on $g$ imply
$\norm{H'}/s,g\ell\le\norm{H'}\sqrt{\ell/s}$.
Thus \Cref{eq:split-piece-norm} is at most
$(1+\sqrt2+2/3)\norm{H'}\sqrt{\ell/s}<4\norm{H'}\sqrt{\ell/s}$.
\end{proof}

\section{Sparse and rounded block encodings}
\label{app:encoding}

We first construct a block encoding of a block-sparse Hermitian
matrix $K$, then an exact block encoding of its rounded approximation
$\widetilde K$, including inverse and controlled queries.

\subsection{Block encoding of $K$ and controlled queries}
\label{app:sparse-encoding}

\begin{proof}[Proof of \Cref{lem:encoding}]
\emph{Construction.} \begingroup
We trace the circuit in
\Cref{fig:sparse-encoding} to verify that it is a block encoding of
$K$ with subnormalization factor $mg$. The data register $\mathsf B$ holds one of the $D=N/b$
block labels and contains $\log_2D=\log_2(N/b)$ qubits. The data
register $\mathsf T$ holds the position within a $b\times b$ block
and contains $\log_2b$ qubits. Together, $\mathsf B,\mathsf T$
contain the $n=\log_2N$ data qubits. The ancillas requires post-selection are the $\log_2m$-qubit slot register $\mathsf S$
and the flag qubit $\mathsf F$. List each column's nonzero blocks
in $m$ slots, adding padding where needed, and let $\mathcal N(J)$
be the set of nonpadding slots in column $J$. For
$u\in\mathcal N(J)$, let $I=I(J,u)$ be the recorded row and set
$A_{J,u}=K_{IJ}/g$; for padding, set $A_{J,u}=0$.

Since $\norm{A_{J,u}}\le1$, the Halmos
dilation~\cite{halmos1950normal}
\[
V_{J,u}=
\begin{pmatrix}
A_{J,u} & \sqrt{I-A_{J,u}A_{J,u}^\dagger}\\[2pt]
\sqrt{I-A_{J,u}^\dagger A_{J,u}} & -A_{J,u}^\dagger
\end{pmatrix}
\]
is a unitary on $\mathsf F,\mathsf T$, with the block rows and
columns indexed by the flag. In particular,
\[
\begin{aligned}
V_{J,u}\ket0_{\mathsf F}\ket\psi_{\mathsf T}
={}&\ket0_{\mathsf F}A_{J,u}\ket\psi_{\mathsf T}\\
&+\ket1_{\mathsf F}
\sqrt{I-A_{J,u}^\dagger A_{J,u}}\ket\psi_{\mathsf T}.
\end{aligned}
\]
Thus the encoded block applies $K_{IJ}/g$ for a nonpadding
slot. For padding, $V_{J,u}=X_{\mathsf F}\otimes I_{\mathsf T}$,
so this component vanishes. These unitaries are constructed during
classical preprocessing, and $\SEL$ applies $V_{J,u}$ controlled
by $(J,u)$, leaving the address registers unchanged.

Let $\Pi$ be the permutation defined in
\Cref{lem:encoding}, which acts on the slot and block-label
registers $\mathsf S,\mathsf B$ according to the address map $q$.
For a nonpadding address $(J,u)$ recording row $I$, let $u'$
be the slot of column $I$ recording row $J$. Hermiticity gives
$K_{JI}=K_{IJ}^\dagger$, so this partner exists, and
\[
\Pi\ket u_{\mathsf S}\ket J_{\mathsf B}
=\ket{u'}_{\mathsf S}\ket I_{\mathsf B}.
\]
The partner maps back to $(J,u)$; padding and diagonal-block
addresses are fixed. Hence $q^2=\mathrm{id}$ and $\Pi$ is unitary.
Updating the slot together with the block label makes this routing
reversible. The registers $\mathsf F,\mathsf T$ are unchanged.
\endgroup
With $F_{\mathsf S}=\mathsf H^{\otimes\log_2m}$
acting on the slot register, set
\begin{equation}
U_K=F_{\mathsf S}\,\Pi\,\SEL\,F_{\mathsf S}.
\label{eq:sparse-encoding-circuit}
\end{equation}

\begingroup
To compute its encoded block, take an arbitrary input
$\ket J_{\mathsf B}\ket\psi_{\mathsf T}$ and initialize both
ancillary registers to zero. The first Hadamard layer prepares
\[
\frac1{\sqrt m}\sum_u
\ket u_{\mathsf S}\ket0_{\mathsf F}
\ket J_{\mathsf B}\ket\psi_{\mathsf T}.
\]
After $\SEL$, the component with flag $\mathsf F=0$ is
\[
\frac1{\sqrt m}\sum_{u\in\mathcal N(J)}
\ket u_{\mathsf S}\ket0_{\mathsf F}
\ket J_{\mathsf B}\,\frac{K_{IJ}}g\ket\psi_{\mathsf T},
\]
where $I=I(J,u)$ in each summand. The remaining component has
flag $\mathsf F=1$ and cannot contribute to the zero-ancilla
corner, since neither $\Pi$ nor the final Hadamard layer changes
the flag. Applying $\Pi$ routes the displayed component to
\[
\frac1{\sqrt m}\sum_{u\in\mathcal N(J)}
\ket{u'}_{\mathsf S}\ket0_{\mathsf F}
\ket I_{\mathsf B}\,\frac{K_{IJ}}g\ket\psi_{\mathsf T},
\]
where $u'$ is the partner slot associated with $(J,u)$.
For every output slot,
\[
\bra{0^{\log_2m}}_{\mathsf S}F_{\mathsf S}
\ket{u'}_{\mathsf S}=\frac1{\sqrt m}.
\]
Consequently, after the final Hadamard layer, the component with
both ancillary registers zero is
\[
\begin{aligned}
&\ket{0^{\log_2m}}_{\mathsf S}\ket0_{\mathsf F}\otimes
\frac1{mg}\sum_{u\in\mathcal N(J)}
\ket I_{\mathsf B} K_{IJ}\ket\psi_{\mathsf T}\\
&\qquad=\ket{0^{\log_2m}}_{\mathsf S}\ket0_{\mathsf F}\otimes
\frac K{mg}\bigl(\ket J_{\mathsf B}\ket\psi_{\mathsf T}\bigr).
\end{aligned}
\]
Here we used that each nonzero block in column $J$ occurs in
exactly one slot and that
$K=\sum_{I,J}\ket I_{\mathsf B}\!\bra J_{\mathsf B}\otimes K_{IJ}$.
The preparation and unpreparation each contribute
a factor $1/\sqrt m$. The final layer need not return the slot
register to zero with certainty; the calculation extracts its
zero-output component. By linearity,
\begin{equation}
\begin{aligned}
&\bigl(\bra{0^{\log_2m}}_{\mathsf S}\bra0_{\mathsf F}
\otimes I_{\mathsf B\mathsf T}\bigr)\\
&\quad\times U_K\bigl(\ket{0^{\log_2m}}_{\mathsf S}\ket0_{\mathsf F}
\otimes I_{\mathsf B\mathsf T}\bigr)=\frac K{mg}.
\end{aligned}
\label{eq:sparse-encoding-corner}
\end{equation}
Here $I_{\mathsf B\mathsf T}$ is the identity on the data registers.
This proves that $U_K$ is an exact $(mg,\log_2m+1,0)$-block
encoding of $K$.
\par\endgroup

\begingroup
\emph{Compilation.} We first bound the cost of $\SEL$ and the
address permutation $\Pi$, then explain how to implement the
controlled and inverse queries required by QSVT.
There are $Dm$ addresses $(J,u)$, one for each block column and
slot. For each address, the selected unitary $V_{J,u}$ acts on
the flag qubit $\mathsf F$ and the $b$-dimensional register
$\mathsf T$, so its dimension is $2b$. Substituting $M=Dm$ and
$d=2b$ into \Cref{lem:tan} gives
$O(b\sqrt{Dm\Lambda}+b^2\Lambda)$ $T$ gates and ancilla qubits,
where $\Lambda=\log_2(2b)+\log_2(1/\eta)$. The error $\eta$
bounds the implementation of the entire $\SEL$ oracle.

To implement $\Pi$, write $x=(J,u)$ and let
$r=\log_2(Dm)$ be the number of address qubits. Introduce a
second $r$-qubit register initialized to $\ket{0^r}$. Compute
$q(x)$ into this register, swap the two registers, and use a
second Boolean lookup to erase the old address:
\[
\begin{aligned}
\ket x\ket{0^r}
&\longmapsto\ket x\ket{q(x)}\\
&\longmapsto\ket{q(x)}\ket x\\
&\longmapsto\ket{q(x)}\ket{x\oplus q(q(x))}\\
&=\ket{q(x)}\ket{0^r}.
\end{aligned}
\]
Here $\oplus$ denotes bitwise XOR, and the last equality uses
$q^2=\mathrm{id}$. This computation also applies to superpositions
of addresses by linearity. Each lookup has $Dm$ entries with
$r$ output bits, so the two lookups together cost
$O(\sqrt{Dm\log(2Dm)})$ $T$ gates and ancilla qubits by
\Cref{lem:boolean}. They are exact and introduce no additional
error. The Hadamard layers and the register swap use only
Clifford gates. Adding the costs proves \Cref{eq:encodingcost}.

For a controlled query, introduce a control qubit. We describe
the circuit for each computational-basis state $\ket c$, where
$c\in\{0,1\}$; its action on superpositions follows by linearity.
For every address $(J,u)$, select the
identity $I_{\mathsf F\mathsf T}$ when $c=0$ and $V_{J,u}$ when
$c=1$. Denote this operation by $\SEL_c$. Likewise, let $\Pi_c$
leave $(J,u)$ unchanged when $c=0$ and apply $q$ when $c=1$.
Both operations leave the control state $\ket c$ unchanged.
Applying both Hadamard layers gives, when the control is in $\ket c$,
\[
F_{\mathsf S}\Pi_c\SEL_c F_{\mathsf S}
=\begin{cases}
F_{\mathsf S}^2=I,&c=0,\\
U_K,&c=1.
\end{cases}
\]
Thus, when $c=0$, the middle operations leave all registers
unchanged, and the two unconditional
applications of $F_{\mathsf S}$ then cancel. This cancellation
holds for every state of the slot register. Including $c$ in
the lookup address doubles the table sizes from $Dm$ to $2Dm$,
which changes the cost by only a constant factor. The resulting
compiled circuit approximates this controlled unitary to error
$\eta$. Reversing the gate order and replacing each gate by its
adjoint implements the inverse of either query with the same
$T$-gate and ancilla costs. Unitary invariance of the norm also
preserves the implementation error in \Cref{def:implementation}.

To use the block encoding $U_K$ of $K$ in QSVT, we distinguish
its block-encoding ancillas from the additional ancilla qubits used
to compile the circuit. The block-encoding ancillas are the slot register $\mathsf S$
and flag qubit $\mathsf F$, totaling $\log_2m+1$ qubits. They
are part of the input to $U_K$ and may be nonzero, in
superposition, or entangled with the data registers during QSVT.
The implementation guarantee therefore applies to every input
state on $\mathsf S,\mathsf F,\mathsf B,\mathsf T$.
Let $\mathsf W$ denote the compiler workspace, a register of $w$
additional ancilla qubits initialized to $\ket{0^w}_{\mathsf W}$.
For every normalized joint state
$\ket\Phi_{\mathsf S\mathsf F\mathsf B\mathsf T}$, the compiled
circuit $\widetilde U_K$ satisfies
\[
\norm{\widetilde U_K\bigl(\ket{0^w}_{\mathsf W}
\ket\Phi_{\mathsf S\mathsf F\mathsf B\mathsf T}\bigr)
-\ket{0^w}_{\mathsf W}U_K
\ket\Phi_{\mathsf S\mathsf F\mathsf B\mathsf T}}\le\eta.
\]
The controlled and inverse queries satisfy the corresponding
guarantees.

QSVT reflections apply a phase flip conditioned on the
$\log_2m+1$ qubits in $\mathsf S,\mathsf F$ being zero, with $O(1)$ additional
control qubits. Computing this condition with Toffoli gates,
applying the phase flip, and uncomputing the condition uses
$O(1+\log m)$ $T$ gates and temporary ancilla qubits.
Single-qubit rotations are counted separately, as in
\Cref{lem:hamsim}.

The ideal queries and reflections return their temporary
ancillas to zero. Approximate queries may leave residual errors
in these ancillas, which are included in the full error analysis above. The composition rule in \Cref{sec:prelim} therefore
permits these ancillas to be reused without resetting them,
with the $Q_t$ queries contributing at most $Q_t\eta$ to the
total error.
\par\endgroup
\end{proof}

\begin{proof}[Proof of \Cref{cor:normsum}]
Since $s\ell\le D$ and $m\le2m_0$,
\[
sm\le2(D+\sqrt{2Ds\ell}+s\ell)\le(4+2\sqrt2)D<7D.
\]
There are at most $s$ nonzero pieces, so $\sum_k\alpha_k\le smg\le7Dg$.
Substitute $Dm\le7D^2/s$ and $bD=N$ into \Cref{eq:encodingcost};
the assumed $b^2\Lambda\ge\log(2Dm)$ absorbs the permutation term.
\end{proof}

\subsection{Exact rounded $\SEL$ oracles}
\label{app:exactmux}

\begin{lemma}[Exact rounded $\SEL$ oracles]
\label{lem:exactmux}
Let $M=2^\mu$, $d=2^\kappa$, and $B\ge\kappa+2$, and let
$W=\sum_x\ket x\!\bra x\otimes V_x$ with $V_x\in U(d)$. There are unitaries
$\widetilde V_x$ with $\norm{\widetilde V_x-V_x}\le2^{2-B}\pi d$, a fixed
$(B+1)$-qubit state $\ket\Psi$, and a Clifford+$T$ circuit $C$ with
$O(d\sqrt{MB}+d^2B)$ $T$ gates and ancillas such that
\begin{equation}
C\big(\ket x\ket\psi\ket{0^w}\ket\Psi\big)=\ket x\,\widetilde V_x\ket\psi\,\ket{0^w}\ket\Psi
\label{eq:exact-rounded-multiplexer}
\end{equation}
for all $x$ and $\psi$. The same holds, at the same cost, for the inverse,
for the singly controlled version, and when some branches carry
$\widetilde V_y^\dagger$ in place of $\widetilde V_x$.
\end{lemma}

\begin{proof}[Proof of \Cref{lem:exactmux}]
\begingroup
For $d\ge2$, decompose each $V_x$ using the same sequence of
$d-1$ uniformly controlled single-qubit gates~\cite{tan2025unitary}.
For each setting of the control qubits, the selected single-qubit
unitary has the form
$e^{i\omega_0}R_z(\omega_1)R_y(\omega_2)R_z(\omega_3)$,
with angles depending on $x$ and the control values.
Round each angle to the nearest multiple of $2\pi2^{-B}$.
This changes each uniformly controlled gate by at most
$4\pi2^{-B}$ in operator norm. Summing the errors over the
$d-1$ gates gives
$\norm{\widetilde V_x-V_x}\le4\pi(d-1)2^{-B}$.

The fixed catalyst state
\[
\ket\Psi=2^{-(B+1)/2}\sum_{k=0}^{2^{B+1}-1}e^{-2\pi ik/2^{B+1}}\ket k
\]
implements each rounded angle exactly by phase
kickback~\cite{gidney2018halving}. Adding an integer $a$ modulo
$2^{B+1}$ to the catalyst register multiplies the state by
$e^{2\pi ia/2^{B+1}}$ and leaves the catalyst unchanged.
For a rounded angle $\theta=2\pi q/2^B$, with $q\in\mathbb Z$,
first add $-q$, then add $2q$ conditioned on the target qubit
being $\ket1$. The target states $\ket0$ and $\ket1$ acquire
phases $e^{-i\theta/2}$ and $e^{i\theta/2}$, respectively,
which implements $R_z(\theta)$.
Clifford conjugation gives
$R_y(\theta)=S\mathsf H R_z(\theta)\mathsf H S^\dagger$,
where $S=\diag(1,i)$. Adding $2q$ without the target control implements
$e^{i\theta}$. For $d=1$, only this scalar phase is needed;
rounding it gives error at most $\pi2^{-B}$ and cost
$O(\sqrt{MB}+B)$.

For each address $x$, the angle table contains $O(d)$ angles
for each of $O(d)$ uniformly controlled phase and rotation
gates, hence $O(d^2B)$ bits. \Cref{lem:boolean} loads and
erases this table at cost $O(d\sqrt{MB})$, following the lookup
scheme in~\cite[Lemma~4.1]{tan2025unitary}.
For each uniformly controlled gate, reversibly selecting the
appropriate angle from its table costs $O(dB)$ Toffoli gates,
and the modular additions cost $O(B)$~\cite{cuccaro2004new}.
Uncomputing the selected angle has the same selection cost.
Summing over the $O(d)$ gates gives total
cost $O(d\sqrt{MB}+d^2B)$.

Reversing the circuit and adjointing its gates implements the
inverse. For a controlled query, include the control bit in
the lookup address and select either the original angle table
or an all-zero table, which implements the identity. This
doubles the number of addresses and preserves the cost bound.
For adjoint pairs, let $\boldsymbol{\theta}_x$ denote the ordered
list of rounded angles in the elementary phase and rotation
gates implementing $\widetilde V_x$. Append the same elementary
gate sequence in reverse order. Using the concatenated angle lists
$(\boldsymbol{\theta}_x,\mathbf0)$ and
$(\mathbf0,-\boldsymbol{\theta}_x^{\mathrm{rev}})$ for the two
branches implements $\widetilde V_x$ and $\widetilde V_x^\dagger$,
respectively. Here $\mathrm{rev}$ reverses the angle list, and
$\mathbf0$ is a zero list of the same length. The zero angles
make one half the identity; reversing the gate order and
negating the angles in the other half gives the adjoint.
This doubles the sequence length and preserves the asymptotic cost.
\par\endgroup
\end{proof}

\subsection{Block encodings of rounded Hermitian matrices}
\label{app:roundedencoding}

\begin{proof}[Proof of \Cref{lem:roundedencoding}]
If $K=0$, set $\widetilde K=0$. Applying a Pauli $X$ gate to the
flag qubit $\mathsf F$ gives an exact block encoding of the zero matrix,
since $\bra0 X\ket0=0$. For $K\ne0$, modify the lookup tables in
\Cref{lem:encoding} as follows. For each off-diagonal
pair, round the dilation of $K_{IJ}/g$ using \Cref{lem:exactmux}
and place its exact adjoint at the adjoint block's address. Hence
$\widetilde K_{JI}=\widetilde K_{IJ}^\dagger$.

For a nonzero diagonal block, use two fixed slots carrying a
rounded dilation $\widetilde V$ of $K_{JJ}/(2g)$ and its exact adjoint.
Let $\widetilde A$ denote the encoded $b\times b$ block of $\widetilde V$. The two slots then contribute
\[
\widetilde K_{JJ}=g(\widetilde A+\widetilde A^\dagger),
\]
which is Hermitian. Each block column contains at most one diagonal
block, so replacing its single slot by two slots uses at most
$m_0+1\le m$ slots in that column. At every unused slot, select
$X_{\mathsf F}\otimes I_{\mathsf T}$ and let $\Pi$ fix its address.
This operation maps the flag from $\ket0$ to $\ket1$, so the slot
contributes zero to the encoded matrix. To implement these branches
exactly, use zero angle words at unused addresses in the rounded
multiplexer, making it act as the identity there. Compute a bit
indicating whether the address is unused, apply $X_{\mathsf F}$
conditioned on this bit, and uncompute the bit. By \Cref{lem:boolean},
this additional lookup costs $O(\sqrt{Dm})$ $T$ gates and ancilla qubits.

By \Cref{lem:exactmux} with $d=2b$, rounding changes each dilation
by at most $8\pi b2^{-B}$ in operator norm. Taking its block with
flag $\mathsf F=0$ on input and output cannot increase this error.
After multiplying by $g$, each off-diagonal block therefore has error
at most $8\pi gb2^{-B}$. For a diagonal block, write
$A=K_{JJ}/(2g)=A^\dagger$. Its error satisfies
\[
\begin{aligned}
\norm{\widetilde K_{JJ}-K_{JJ}}
&=g\norm{(\widetilde A-A)+(\widetilde A-A)^\dagger}\\
&\le 2g\norm{\widetilde A-A}
\le16\pi gb2^{-B}.
\end{aligned}
\]

Let $E=\widetilde K-K$. The unused slots contribute exactly zero,
so $E_{IJ}=0$ whenever $K_{IJ}=0$. Since $E$ is Hermitian, it has
at most $m_0\le m$ nonzero blocks in each block row as well as in
each block column. For a Hermitian block matrix, the operator norm
is bounded by the largest sum of block norms in a row. Hence
\[
\norm E\le\max_I\sum_J\norm{E_{IJ}}
\le16\pi mg b2^{-B},
\]
which proves \Cref{eq:rounded-encoding-error}. To see the first
inequality, write a vector $x$ in blocks $x_I$ and bound
$|\langle x,Ex\rangle|$ by
$\sum_{I,J}\norm{E_{IJ}}\norm{x_I}\norm{x_J}$.
Using $2\norm{x_I}\norm{x_J}\le\norm{x_I}^2+\norm{x_J}^2$
and $\norm{E_{IJ}}=\norm{E_{JI}}$ gives the stated bound.

Use these slot unitaries in the circuit
$U_{\widetilde K}=F_{\mathsf S}\Pi\widetilde{\SEL} F_{\mathsf S}$,
where $\widetilde{\SEL}$ denotes the resulting selection oracle.
Given the catalyst $\ket\Psi$, \Cref{lem:exactmux} implements the
rounded multiplexer exactly at cost $O(b\sqrt{DmB}+b^2B)$.
The additional lookup for unused slots is absorbed in this cost.
The Hadamard layers and the address permutation are exact; the
latter costs $O(\sqrt{Dm\log(2Dm)})$ as shown in the proof of
\Cref{lem:encoding}. Together these circuits implement the block
encoding $U_{\widetilde K}$ of $\widetilde K$ exactly with the $T$-gate and ancilla costs in
\Cref{eq:encodingcost} at $\Lambda=B$. They return the catalyst
unchanged and all additional ancilla qubits to $\ket0$.

Reversing the circuit and replacing each gate by its adjoint
implements $U_{\widetilde K}^\dagger$. For a singly controlled
query, condition both $\widetilde{\SEL}$ and $\Pi$ on the control
qubit, leaving the two Hadamard layers unconditional. When the
control is $\ket0$, the middle operations are the identity and
$F_{\mathsf S}^2=I$; when it is $\ket1$, the circuit applies
$U_{\widetilde K}$. Including the control bit in the lookup addresses
at most doubles the table sizes, preserving the asymptotic cost.
Reversing this circuit gives the singly controlled inverse.
\end{proof}

\section{Unitary synthesis using qDRIFT and QSVT}
\label{app:constant}

We prove the shared selection lemma, then establish the parameter
and resource bounds for the compiler in \Cref{alg:constant-full},
summarized in \Cref{alg:constant}.

\begin{algorithm*}[t]
\caption[Detailed qDRIFT and QSVT compiler]{\raggedright Detailed classical
preprocessing and circuit assembly for \Cref{thm:constant}.}
\label{alg:constant-full}
\begin{algorithmic}[1]
\Require $U\in U(N)$ and $\epsilon$ in the range of \Cref{thm:constant}.
\Ensure A fixed Clifford+$T$ circuit implementing $U$ to error $\epsilon$.
\item[\textbf{Setup:}] From $N$ and $\epsilon$, choose the block size $b$,
set $D=N/b$, and choose the number $s$ of partition classes and the
split parameter $\ell$ as in \Cref{app:constant}.
\Statex \textbf{Stage I: Preprocessing.}
\State Compute the principal Hermitian logarithm of $U$ and separate its scalar part $\varphi I$ and traceless part $H_0$.
\If{$H_0=0$}
  \State \Return a phase circuit implementing $e^{i\varphi}I$ with error $\epsilon/8$.
\EndIf
\Repeat \Comment{Success per trial $\ge3/4$.}
  \State Draw a fresh similarity $P$ and form $H'=PH_0P^T$.
  \State Sample a fresh partition of $H'$ into $s$ Hermitian pieces $H_k$. \Comment{Preserve Hermitian block pairs.}
\Until{all block-norm, piece-norm, and sparsity bounds hold at failure parameter $1/8$ (\Cref{lem:flattening,lem:split}).}
\State Discard zero pieces and set $\lambda=\sum_k\norm{H_k}$.
\If{$\lambda<\epsilon/4$} \Comment{$\norm{e^{iH'}-I}\le\lambda$.}
  \State \Return a phase circuit implementing $e^{i\varphi}I$ with error $\epsilon/8$.
\EndIf
\State Build a block encoding $U_{H_k}$ of each remaining piece $H_k$, with subnormalization factor $\alpha_k$ (\Cref{cor:normsum}).
\Statex \textbf{Stage II: qDRIFT certification.}
\State Use \Cref{lem:constant-selection} with $\varepsilon=\epsilon/4$, $p=1/2$, and $N_q$ from \Cref{eq:qdrift-segments}.
\Repeat \Comment{Success per trial $\ge1/2$.}
  \State Sample fresh piece indices $K_j$ and durations $\tau_j$; form their ideal product $V$.
\Until{the accuracy and weighted-time tests pass (\Cref{lem:constant-selection}).}
\Statex \textbf{Stage III: Circuit assembly.}
\State Implement the sampled evolutions with QSVT.
\State Compile the queries to $U_{H_k}$, their inverses
and singly controlled versions, and the single-qubit rotations
to the tolerances specified in the error analysis of \Cref{app:constant}.
\State Assemble the following circuit, using exact $P,P^T$ and reusable clean workspace:
\begin{equation}
e^{i\varphi}P^T
\Big(\prod_{j=N_q}^{1}\mathcal S_{\tau_j}(U_{H_{K_j}})\Big)P .
\label{eq:constant-compiled-product}
\end{equation}
\State \Return the compiled circuit.
\end{algorithmic}
\end{algorithm*}

\begin{proof}[Proof of \Cref{lem:constant-selection}]
We first bound the weighted simulation time of the sampled product.
For each segment $j$, \Cref{eq:qdrift-sampling} gives
\[
\begin{aligned}
\E[\alpha_{K_j}\tau_j]
&=\sum_{k=1}^r\frac{\norm{A_k}}{\lambda}
  \frac{\alpha_k\lambda}{N_q\norm{A_k}}\\
&=\frac1{N_q}\sum_{k=1}^r\alpha_k.
\end{aligned}
\]
Summing over the $N_q$ segments yields
\begin{equation}
\E\mathcal T
=\sum_{j=1}^{N_q}\E[\alpha_{K_j}\tau_j]
=\sum_{k=1}^r\alpha_k.
\label{eq:weightedtime}
\end{equation}
Since $\mathcal T\ge0$, Markov's inequality gives
\[
\Pr\!\left(\mathcal T>\frac2p\sum_k\alpha_k\right)
\le\frac{\E\mathcal T}{(2/p)\sum_k\alpha_k}
=\frac p2.
\]
The choice of $N_q$ in \Cref{lem:qdrift} also ensures
\[
\Pr\!\left(\norm{V-e^{iA}}>\varepsilon\right)\le\frac p2.
\]
By the union bound, the probability that either requirement fails
is at most $p$. Thus the same sampled product satisfies both the
accuracy bound and the bound on $\mathcal T$ with probability at
least $1-p$.

Now fix any sample satisfying both bounds.
Apply \Cref{lem:hamsim}
to segment $e^{i\tau_jA_{K_j}}$ using queries to $U_{A_{K_j}}$
at simulation error $\xi/N_q$. This error lies in $(0,1/2)$,
as required by the lemma. Use a common clean ancilla register for
all segments, adding unused ancillas where necessary.
The composition rule in \Cref{sec:prelim} bounds the implementation
error of the product by $\sum_{j=1}^{N_q}\xi/N_q=\xi$, including
the error on the ancillas.

For segment $j$, the query count is
\[
Q_j=O\!\left(\alpha_{K_j}\tau_j+\log\frac{N_q}{\xi}\right).
\]
Summing over all segments and using the accepted bound on
$\mathcal T$ gives
\[
\begin{aligned}
Q=\sum_{j=1}^{N_q}Q_j
&=O\!\left(\mathcal T+N_q\log\frac{N_q}{\xi}\right)\\
&=O\!\left(\frac1p\sum_k\alpha_k
           +N_q\log\frac{N_q}{\xi}\right).
\end{aligned}
\]
These counts include queries to the block encodings $U_{A_k}$ of $A_k$, their
inverses, and their singly controlled versions, as in \Cref{lem:hamsim}.
\end{proof}

\begin{proof}[Proof of \Cref{thm:constant}]
\emph{Parameters.} Set $L=n+\log_2(1/\epsilon)$.
Choose $b$, a power of 2, satisfying
\[
\frac12 N^{3/7}\epsilon^{-4/7}\le b\le 2N^{3/7}\epsilon^{-4/7}.
\]
Put $D=N/b$ and set
\begin{equation}
s=\Big\lfloor\frac{c_s\sqrt D\,\epsilon^2}{n^2L}\Big\rfloor,
\qquad \ell=\log(64Ns),
\label{eq:sconst}
\end{equation}
with a sufficiently small absolute constant $c_s>0$.
For sufficiently large $C_0$ and $N$, the stated range of $\epsilon$
ensures $s\ge1$ and $s\ell\log^2(8N)\le D/C$.

Write $U=e^{iH}$, choosing the Hermitian logarithm $H$ with eigenvalues
in $(-\pi,\pi]$; this is our convention for the principal Hermitian
logarithm. Set $\varphi=\tr H/N$ and $H_0=H-\varphi I$.
Since $\norm{H}\le\pi$ and $\abs{\varphi}\le\pi$,
\[
\norm{H_0}\le\norm{H}+\abs{\varphi}\le2\pi.
\]
If $H_0=0$, then $U=e^{i\varphi}I$, so only the scalar-phase circuit
is needed. Otherwise, for an accepted choice of $P$,
\Cref{lem:flattening} with $\delta=1/8$ bounds the operator norm of
every $b\times b$ block of $H'=PH_0P^T$ by
\[
\norm{H'_{IJ}}\le g:=\frac{C_A\norm{H_0}\log(8N)}{\sqrt D},
\qquad I,J\in[D].
\]
The parameter condition $s\ell\log^2(8N)\le D/C$, with $C\ge C_A^2$,
and the identity $\norm{H'}=\norm{H_0}$ imply
$g\le\norm{H'}/\sqrt{s\ell}$, as required by \Cref{lem:split}.
Applying \Cref{lem:split} with $\delta=1/8$ and then
\Cref{cor:normsum} gives
\begin{equation}
\lambda\le8\pi\sqrt{s\ell},\qquad
\sum_k\alpha_k\le7Dg=O(\sqrt D\,n).
\label{eq:constant-piece-norm-sum}
\end{equation}
Here $s$ includes classes whose pieces were discarded as zero.
The similarity and split pass jointly with probability at least
$3/4$, giving a constant expected trial count.
If $\lambda<\epsilon/4$, the phase alone
suffices since $\norm{e^{iH'}-I}\le\lambda$.

\emph{Selection and queries.} For $\lambda\ge\epsilon/4$, choose
\begin{equation}
N_q=\Big\lceil\frac{704\lambda^2}{\epsilon^2}\log(8N)+\frac{8\lambda^2}{\epsilon}\Big\rceil.
\label{eq:qdrift-segments}
\end{equation}
Apply \Cref{lem:constant-selection} with $A_k=H_k$,
$\varepsilon=\xi=\epsilon/4$, and $p=1/2$.
Each trial passes the accuracy and weighted-time tests with probability at
least $1/2$, and the accepted product has QSVT implementation error
$\epsilon/4$.
Now $N_q=O(sn^2/\epsilon^2)$ and
$\log(N_q/\epsilon)=O(L)$. Substituting \Cref{eq:sconst} gives
\begin{equation}
Q=O(\sqrt D\,n+sn^2L/\epsilon^2)=O(\sqrt D\,n).
\label{eq:total-calls}
\end{equation}
Implement each query to $U_{H_k}$, its inverse,
or a singly controlled version with error at most
$\eta=\epsilon/(8Q)$, so $\Lambda=O(L)$.
Since $b=\Omega(N^{3/7})$ and $m\le7D$, the domination condition
in \Cref{cor:normsum} holds, giving per-query cost
$c_{\rm query}=O(N\sqrt{L/s}+b^2L)$ $T$ gates and reusable clean ancillas.

\emph{Error and resources.} The basis changes and reflections are
exact. Allocate $\epsilon/8$ each to the combined
errors from implementing the queries to $U_{H_k}$, all
$O(Q)$ rotations, and the scalar phase, synthesizing the rotations and
the phase by the single-qubit synthesis bound in
\Cref{eq:single-qubit-synthesis}.
Together with the qDRIFT
and QSVT errors, the composition rule gives
\[
\epsilon/4+\epsilon/4+\epsilon/8+\epsilon/8+\epsilon/8
=7\epsilon/8.
\]
The $O(\sqrt N)$ basis-change cost, $O(nQ)$ reflection cost,
$O(QL)$ rotation cost, and $O(L)$ phase cost are all lower order.
Thus
\begin{equation}
T=O\Big(n^2L\Big(\frac{ND^{1/4}}{\epsilon}+\sqrt D\,b^2\Big)\Big).
\label{eq:constant-total-cost}
\end{equation}
Substituting $b=\Theta(N^{3/7}\epsilon^{-4/7})$ and $D=N/b$
makes both terms $\Theta(N^{8/7}\epsilon^{-6/7})$.
Sequential queries reuse workspace, so the ancilla count is
$O(c_{\rm query})=O(nL\,N^{6/7}\epsilon^{-8/7})$, including the
smaller logical registers and sign-oracle workspace.
\end{proof}

\begin{remark}[Deterministic product formulas]
\label{rem:strang}
Replacing qDRIFT by second-order Strang splitting uses a deterministic
product formula, but gives the weaker exponent $7/6$. The similarity and
partition still involve classical randomness; both compilers output
fixed circuits. The commutator
bounds of Ref.~\cite{childs2021theory} bound the error after $r$ steps by
$\Ot(\sqrt s/r^2)$. This uses $\norm{H_j}=\Ot(s^{-1/2})$ from
\Cref{lem:split}, and $\norm{\sum_{k>j}H_k}=\Ot(1)$ from the same
Bernstein argument applied to unions of classes. The resulting product
formula is fixed once the classical similarity and partition are chosen.
At constant accuracy, the error bound requires $r=\Ot(s^{1/4})$,
so its $2rs$ segments number $\Ot(s^{5/4})$.

Keeping this number of order $\sqrt D$, as for qDRIFT, gives
$s\approx D^{2/5}$. Balancing the resulting lookup and processing
costs gives $b\approx N^{4/9}$ and $T=\Ot(N^{7/6})$. The
smaller qDRIFT exponent follows because its segment count depends on the
pieces only through the bound $\lambda^2=\Ot(s)$.
\end{remark}

\section{Proof of unitary synthesis using residual refinement}
\label{app:refinement}

We prove \Cref{thm:refined} for \Cref{alg:refinement-full},
summarized in \Cref{alg:refinement}.
\Cref{app:refinement-parameters} specifies the parameters, residual
data, and acceptance tests.
\Cref{sec:refinement-errors} proves residual contraction, termination,
and the final error bound. These estimates give the $T$-count and
ancilla bounds in \Cref{sec:refinement-cost}, followed by a comparison
of precision dependence.
Finally, \Cref{sec:effective-compilation} gives a terminating classical
compiler for computably specified targets with the quantum resource
bounds of either \Cref{thm:constant} or \Cref{thm:refined}.

\subsection{Parameters and residual data}
\label{app:refinement-parameters}

Constants $C,c$ are absolute, and $n\ge n_0$ is sufficiently large.

For the pieces of \Cref{lem:split}, $sm\le s(2m_0+2)\le9D$ when
$s\ell\le D$. So the subnormalization factors satisfy $\sum_k\alpha_k\le9Dg$, and
each query to the block encoding of a rounded piece
costs $O(N\sqrt{B/s}+b^2B)$, as in \Cref{cor:normsum}, whenever
$b^2B\ge\log(2Dm)$. The word length below satisfies this condition.

\paragraph*{Parameters.} Use the word length
\begin{equation}
B=3n+\big\lceil\log_2\big(9216\pi^2C_A\log(8N)\big)\big\rceil .
\label{eq:wordlength}
\end{equation}
The block encodings of the rounded pieces and full rounded residual
generators use \Cref{lem:roundedencoding}, with words relative to
the current block scale, and share one state $\ket\Psi$. Since
$b\sqrt D=\sqrt{Nb}\le N$, this choice gives
$9216\pi^2C_A\,b\sqrt D\log(8N)\,2^{-B}\le N^{-2}$ with $B=O(n)$.
We use this inequality in both stages.

Let $L_0=c_0N^{1/4}n^{-7/8}$, with $c_0$ small enough for the level
stage to be nonempty when $L\le L_0$. Choose the block size once:
if $L\le L_0$, take $b$ to be a power of two in the range
$[\frac12,2]\cdot N^{3/7}L^{2/7}$ and put $D=N/b$.
Run the qDRIFT levels, then correct the
remaining residual by block-encoding and simulating its full
rounded generator. If $L>L_0$, take $D$ to be
the least power of two at least $\min\{N,n^2L^2\}$ and put $b=N/D$;
in this regime, run only the corrections.

For $L\le L_0$, \Cref{sec:refinement-cost} bounds the total lookup
cost of the qDRIFT levels by $O(n^2ND^{1/4}\sqrt L)$ and their
total block-processing cost by $O(n^2\sqrt D\,b^2)$.
Substituting $D=N/b$ and balancing these contributions gives
$b^{7/4}=\Theta(N^{3/4}\sqrt L)$, hence the choice
$b=\Theta(N^{3/7}L^{2/7})$. Both contributions are then
$O(n^2N^{8/7}L^{3/7})$.

For $L>L_0$, the first correction starts from a traceless
generator of norm at most $2\pi$. By
\Cref{lem:flattening,lem:roundedencoding}, the block encoding of
its full rounded generator has subnormalization factor
$O(n\sqrt D)$. The chosen $D$ satisfies $D\le2n^2L^2$, so this
factor is $O(n^2L)$. Moreover, if $D<N$, then $D\ge n^2L^2$ gives
\[
b^2B\le\frac{N^2B}{n^4L^4}
=O\!\left(\frac{NB}{\sqrt n}\right)
=O(N\sqrt B),
\]
where we used $L>L_0$ and $B=\Theta(n)$. If $D=N$, then $b=1$
and the same bound holds. Thus the block-processing cost
$O(b^2B)$ is at most the order of the lookup cost $O(N\sqrt B)$,
and each query to the block encoding of a full rounded correction
generator costs $O(N\sqrt B)$ in this regime.

For the level stage, put $\epsilon_l=2^{-l}/4$, $a_0=2\pi$, and
$a_l=4\epsilon_{l-1}$ for $l\ge1$, and define
\begin{equation}
s_l=\Big\lfloor\frac{c_s\sqrt D\,a_l}{nL}\Big\rfloor .
\label{eq:sl}
\end{equation}
$a_l$ bounds the level's generator norm, with $a_l\le8\pi\epsilon_l$.
The tolerances halve for $l\ge1$, so qDRIFT needs only constant
relative precision.
When $L\le L_0$, let $l_*$ be the first index with $s_{l_*}=0$; then
$1\le l_*=O(n)$. When $L>L_0$, set $l_*=0$, so the level loop is empty.
Indeed, at $l=0$ the argument of the floor in \Cref{eq:sl} is
$\Theta(c_sN^{2/7}L^{-8/7}/n)$, at least one for $L\le L_0$
when $c_0$ is sufficiently small.

\paragraph*{Residual data.}
For a residual unitary $R$, the routine
$\operatorname{Residual}(R)=(H,\varphi,H_0)$ computes
\[
H=-i\log R,\qquad \varphi=\frac{\tr H}{N},\qquad
H_0=H-\varphi I.
\]
We choose the eigenvalues of $H$ in $(-\pi,\pi]$, including
$\pi$ at eigenvalue $-1$ of $R$, as in \Cref{eq:generator}.
Thus $R=e^{i\varphi}e^{iH_0}$ and $H_0$ is traceless.

At level $l$, the current ideal approximation is $W_{l-1}$, with
$W_{-1}=I$, and the residual is $R_l=UW_{l-1}^{-1}$.
Apply the routine above to obtain
$(H^{(l)},\varphi_l,H_0^{(l)})$.
If $H_0^{(l)}=0$, set $P_l=V_l=I$, $H'^{(l)}=0$, and
$\lambda_l=0$; the phase update handles this level exactly.
For $H_0^{(l)}\ne0$, choose a real orthogonal $P_l$ and a split
satisfying the bounds of \Cref{lem:flattening,lem:split}, each
with failure parameter $1/8$, as in \Cref{alg:refinement-full}.
Let $H'^{(l)}=P_lH_0^{(l)}P_l^T$ and write its split as
$H'^{(l)}=\sum_{k\in\mathcal I_l}H_k^{(l)}$, where $\mathcal I_l$
indexes only nonzero pieces. The parameter $s_l$ counts all pieces
in the original random split, including those that are zero, so
$\lvert\mathcal I_l\rvert\le s_l$. The largest block norm is
$g_l=\max_{I,J}\norm{(H'^{(l)})_{IJ}}$, and the
total piece norm is $\lambda_l=\sum_{k\in\mathcal I_l}\norm{H_k^{(l)}}$.
Each subnormalization factor is $\alpha_k^{(l)}=m_lg_l$, where $m_l$ is the
least power of two at least $m_{0,l}+1$ and $m_{0,l}$ is the sparsity
bound of \Cref{lem:split}. Let $\widetilde K_k^{(l)}$ be the
rounded pieces supplied by \Cref{lem:roundedencoding}.
Since $m_l\ge m_{0,l}$, \Cref{lem:encoding}
also supplies exact block encodings of the unrounded pieces
$H_k^{(l)}$ with the same subnormalization factors $\alpha_k^{(l)}$.
Thus the sampling and weighted-time bounds in
\Cref{lem:constant-selection} apply to these pieces. The
simulation circuits use the rounded pieces
$\widetilde K_k^{(l)}$ instead.

When $\lambda_l\ge\epsilon_l/4$, the qDRIFT sampling law and
durations are
\begin{equation}
\begin{aligned}
\Pr(K_{l,u}=k)&=\norm{H_k^{(l)}}/\lambda_l,\\
\tau_{l,u}&=\frac{\lambda_l}{N_q^{(l)}\norm{H_{K_{l,u}}^{(l)}}}.
\end{aligned}
\label{eq:level-sampling}
\end{equation}
Keep the unrounded and rounded products distinct:
\begin{equation}
\begin{aligned}
V_l^{\rm true}&=\prod_{u=N_q^{(l)}}^1 e^{i\tau_{l,u}H_{K_{l,u}}^{(l)}},\\
V_l&=\prod_{u=N_q^{(l)}}^1 e^{i\tau_{l,u}\widetilde K_{K_{l,u}}^{(l)}}.
\end{aligned}
\label{eq:level-products}
\end{equation}
Products act from right to left. The classical qDRIFT acceptance test
uses $V_l^{\rm true}$; the next residual and compiled segments use
$V_l$. When $\lambda_l\ge\epsilon_l/4$, we sample until both
of the following tests pass, using \Cref{lem:constant-selection} with
$\varepsilon=\epsilon_l/2$ and $p=1/2$:
\begin{equation}
\begin{aligned}
\norm{V_l^{\rm true}-e^{iH'^{(l)}}}&\le\epsilon_l/2,\\
\sum_u\alpha_{K_{l,u}}^{(l)}\tau_{l,u}&\le4\sum_k\alpha_k^{(l)}.
\end{aligned}
\label{eq:level-acceptance}
\end{equation}
The first test bounds the qDRIFT approximation error. The second
bounds the sum of simulation times weighted by the block-encoding
subnormalization factors. This sum controls both the query count
and the error introduced when the sampled generators are rounded
(\Cref{eq:level-rounded-error}). If $\lambda_l<\epsilon_l/4$,
the algorithm sets $V_l=I$ and skips sampling at that level.

The correction stage starts from the rounded ideal product
$W^{(0)}=W_{l_*-1}$ produced by the levels; when no levels are
used, $W^{(0)}=I$. At correction $j\ge1$, the transformation
needed to turn the current approximation $W^{(j-1)}$ into $U$ is
\[
R^{(j)}=U(W^{(j-1)})^{-1},
\]
so $R^{(j)}W^{(j-1)}=U$.
Apply the same residual decomposition to define
\[
\begin{aligned}
G_j&=-i\log R^{(j)},\qquad \chi_j=\frac{\tr G_j}{N},\\
G_{0,j}&=G_j-\chi_j I.
\end{aligned}
\]
Thus $R^{(j)}=e^{i\chi_j}e^{iG_{0,j}}$: $\chi_j$ is the scalar
phase, and $G_{0,j}$ is the traceless generator to be simulated.
The notation
$\operatorname{Residual}(R^{(j)})=(G_j,\chi_j,G_{0,j})$
in \Cref{alg:refinement-full} refers to these definitions.

For $G_{0,j}\ne0$, choose a real orthogonal matrix $P_j^{\rm c}$
satisfying the flattening bound of \Cref{lem:flattening}, by
resampling with failure parameter $1/4$ until the bound holds.
The superscript $\rm c$ identifies a change of basis used in
the correction stage. Set
\[
\begin{aligned}
G'_j&=P_j^{\rm c}G_{0,j}(P_j^{\rm c})^T,\\
g_j&=\max_{I,J}\norm{(G'_j)_{IJ}}.
\end{aligned}
\]
This choice bounds the norm of every block of $G'_j$ by
$C_A\norm{G_{0,j}}\log(4N)/\sqrt D$.
Apply \Cref{lem:roundedencoding} to the full matrix $G'_j$,
using word length $B$ and $m=2D$, to obtain a Hermitian matrix
$\widetilde G_j$ and an exact block encoding of $\widetilde G_j$.
Here $m=2D$ accommodates the at most $D$ nonzero blocks in
each block column and the additional slot for diagonal-block
symmetrization. Define the subnormalization factor of this block
encoding and the error from rounding the generator by
\begin{equation}
\alpha_j=2Dg_j,\qquad \rho_j=\norm{\widetilde G_j-G'_j}.
\label{eq:correction-data}
\end{equation}
The correction first applies $P_j^{\rm c}$, simulates
$\widetilde G_j$ for time one, and then applies
$(P_j^{\rm c})^T$ and the scalar phase. The resulting ideal update is
\begin{equation}
W^{(j)}=e^{i\chi_j}(P_j^{\rm c})^Te^{i\widetilde G_j}
P_j^{\rm c}W^{(j-1)}.
\label{eq:correction-ideal-update}
\end{equation}
We stop when $\rho_j\le\epsilon/8$; the next subsection proves
that $\norm{U-W^{(j)}}\le\rho_j$.
If $G_{0,j}=0$, set $P_j^{\rm c}=I$, $\widetilde G_j=0$,
and $\alpha_j=\rho_j=0$. Then the scalar phase alone gives
$W^{(j)}=U$, and no Hamiltonian simulation is needed.

\begin{algorithm*}[t]
\caption[Detailed residual-refinement compiler]{\raggedright Detailed
classical preprocessing and circuit assembly for \Cref{thm:refined}.}
\label{alg:refinement-full}
\begin{algorithmic}[1]
\Require $U\in U(N)$ and $\epsilon$ in the range of \Cref{thm:refined}.
\Ensure A fixed Clifford+$T$ circuit implementing $U$ to error $\epsilon$.
\item[\textbf{Setup:}] From $N$ and $\epsilon$, choose the block size $b$,
word length $B$, number $l_*$ of block-sparse levels, level tolerances
$\epsilon_l$, generator-norm bounds $a_l$, and piece counts $s_l$;
set $D=N/b$ (\Cref{app:refinement-parameters}).
\Statex \textbf{Stage I: Refinement levels.}
\State Set $W_{-1}=I$.
\For{$l=0,\ldots,l_*-1$}
  \State Form $R_l=UW_{l-1}^{-1}$ and $(H^{(l)},\varphi_l,H_0^{(l)})=\Call{Residual}{R_l}$; set $P_l=V_l=I$.
  \If{$H_0^{(l)}\ne0$}
    \Repeat \Comment{Success per trial $\ge3/4$.}
      \State Draw a fresh similarity $P_l$ and split $H'^{(l)}=P_lH_0^{(l)}P_l^T$ into $s_l$ Hermitian pieces.
    \Until{the block-norm, piece-norm, and sparsity bounds hold at failure parameter $1/8$ (\Cref{lem:flattening,lem:split}).}
    \State Discard zero pieces; compute their norm sum $\lambda_l$ and construct block encodings of the rounded pieces $\widetilde K_k^{(l)}$ with subnormalization factors $\alpha_k^{(l)}$ (\Cref{lem:roundedencoding}).
    \If{$\lambda_l\ge\epsilon_l/4$}
      \State Use \Cref{lem:constant-selection} with $\varepsilon=\epsilon_l/2$ and $p=1/2$, choosing the least sufficient $N_q^{(l)}$ from \Cref{eq:qdrift-concentration-count}.
      \Repeat \Comment{Success per trial $\ge1/2$.}
        \State Sample a fresh unrounded product $V_l^{\rm true}$ (\Cref{eq:level-sampling,eq:level-products}).
      \Until{the accuracy and weighted-time tests pass (\Cref{eq:level-acceptance}).}
      \State Replace its generators by the rounded ones to form $V_l$ (\Cref{eq:level-products}). \Comment{Include rounding in the residual.}
    \EndIf
  \EndIf
  \State Set $W_l=e^{i\varphi_l}P_l^TV_lP_lW_{l-1}$.
\EndFor
\Statex \textbf{Stage II: Corrections using the full residual generator.}
\State Set $W^{(0)}=W_{l_*-1}$ and $j=0$.
\Repeat \Comment{$O(1+L/n)$ corrections suffice.}
  \State Increase $j$ by one; form $R^{(j)}=U(W^{(j-1)})^{-1}$ and $(G_j,\chi_j,G_{0,j})=\Call{Residual}{R^{(j)}}$.
  \If{$G_{0,j}=0$}
    \State Set $P_j^{\rm c}=I$, $\widetilde G_j=0$, and $\alpha_j=\rho_j=0$. \Comment{Only the scalar phase remains.}
  \Else
    \State Resample $P_j^{\rm c}$ until the flattening bound holds at failure parameter $1/4$ (\Cref{lem:flattening}). \Comment{Success per trial $\ge3/4$.}
    \State Form the full rounded generator $\widetilde G_j$ and construct its block encoding with subnormalization factor $\alpha_j$; compute its rounding error $\rho_j$ (\Cref{lem:roundedencoding,eq:correction-data}).
  \EndIf
  \State Set $W^{(j)}=e^{i\chi_j}(P_j^{\rm c})^Te^{i\widetilde G_j}P_j^{\rm c}W^{(j-1)}$.
\Until{$\rho_j\le\epsilon/8$.}
\State Set $J=j$.
\Statex \textbf{Stage III: Circuit assembly.}
\State Set the tolerances for simulation, rotations, catalyst preparation and unpreparation, and the phase as specified in the error analysis of \Cref{sec:refinement-errors}.
\State Prepare $\ket\Psi$; compile $P_l^TV_lP_l$, then $(P_j^{\rm c})^Te^{i\widetilde G_j}P_j^{\rm c}$, in increasing index order (\Cref{lem:exactmux,lem:hamsim}).
\State Apply the phase $\Phi=\sum_{l=0}^{l_*-1}\varphi_l+\sum_{j=1}^{J}\chi_j$, unprepare $\ket\Psi$, and \Return the circuit.
\end{algorithmic}
\end{algorithm*}

\subsection{Accuracy and termination}
\label{sec:refinement-errors}

\begin{proof}[Proof of \Cref{thm:refined}: accuracy]
The products $W_l$ and $W^{(j)}$ use exact exponentials of rounded
generators, so their error includes qDRIFT and generator rounding.
After bounding this error and proving termination, we add the
circuit-implementation errors.

\emph{Level contraction.} We prove by induction that level $l$
produces $\norm{U-W_l}=\norm{R_{l+1}-I}\le\epsilon_l$.
Initially $\norm{H_0^{(0)}}\le2\pi=a_0$; for $l\ge1$, the
induction hypothesis $\norm{R_l-I}\le\epsilon_{l-1}\le1/4$
bounds every principal eigenphase in absolute value by
$\pi\epsilon_{l-1}/2$, giving
$\norm{H_0^{(l)}}\le2\norm{H^{(l)}}\le\pi\epsilon_{l-1}<a_l$.
If $H_0^{(l)}=0$, the phase update gives $W_l=U$.
If $H_0^{(l)}\ne0$ but $\lambda_l<\epsilon_l/4$, setting $V_l=I$
gives
\[
\norm{V_l-e^{iH'^{(l)}}}
\le\norm{H'^{(l)}}\le\lambda_l<\epsilon_l/4.
\]

For sampled levels, $\lambda_l\ge\epsilon_l/4$ and $s_l\le N$,
so $\zeta=\log(64N^2)=O(n)$ bounds the logarithmic factor
in \Cref{lem:split} at failure parameter $1/8$.
The bounds $s_l=O(\sqrt D/(nL))$ and $\sqrt D=\Omega(N^{1/4})$
give $s_l\zeta\log^2(8N)\le D/C_A^2$ and $s_l\zeta\le D$,
so the accepted flattening and splitting bounds yield
$\lambda_l\le4a_l\sqrt{s_l\zeta}$ and
\begin{equation}
\sum_k\alpha_k^{(l)}\le9C_Aa_l\sqrt D\log(8N).
\label{eq:level-normalization-bound}
\end{equation}
Since $\epsilon_l/2\le4\lambda_l$ and $a_l/\epsilon_l\le8\pi$,
\Cref{lem:constant-selection} with $A_k=H_k^{(l)}$,
subnormalization factors $\alpha_k^{(l)}$, $\varepsilon=\epsilon_l/2$,
and $p=1/2$ uses $N_q^{(l)}=O(s_ln^2)$ segments, with the
accepted product satisfying
$\norm{V_l^{\rm true}-e^{iH'^{(l)}}}\le\epsilon_l/2$.
The generator-rounding bound from \Cref{lem:roundedencoding} is
\[
\norm{\widetilde K_k^{(l)}-H_k^{(l)}}
\le16\pi\alpha_k^{(l)}b2^{-B}.
\]
Duhamel's formula and product telescoping therefore give
\begin{equation}
\begin{split}
\norm{V_l-e^{iH'^{(l)}}}
&\le\frac{\epsilon_l}{2}
 +16\pi b2^{-B}\sum_u\alpha_{K_{l,u}}^{(l)}\tau_{l,u}\\
&\le\epsilon_l.
\end{split}
\label{eq:level-rounded-error}
\end{equation}
The last inequality follows by bounding the weighted-time sum
by $4\sum_k\alpha_k^{(l)}$ and applying
\Cref{eq:level-normalization-bound,eq:wordlength} with
$a_l\le8\pi\epsilon_l$.

In every branch, substituting
$W_l=e^{i\varphi_l}P_l^TV_lP_lW_{l-1}$ and
$R_l=e^{i\varphi_l}P_l^Te^{iH'^{(l)}}P_l$ into
$R_{l+1}=UW_l^{-1}$ gives
\begin{equation}
R_{l+1}=P_l^Te^{iH'^{(l)}}V_l^{-1}P_l.
\label{eq:level-residual-update}
\end{equation}
Hence
$\norm{U-W_l}=\norm{R_{l+1}-I}
=\norm{e^{iH'^{(l)}}-V_l}\le\epsilon_l$, closing the induction.
This bounds the entire residual, including the errors of
preceding levels.

\emph{Corrections and stopping.}
If $G_{0,j}=0$, the phase update gives $W^{(j)}=U$ and the
algorithm stops. Otherwise \Cref{lem:flattening,lem:roundedencoding}
give
\[
\begin{aligned}
\alpha_j&\le2C_A\sqrt D\norm{G_{0,j}}\log(8N),\\
\rho_j&\le16\pi\alpha_jb2^{-B}.
\end{aligned}
\]
Using \Cref{eq:correction-ideal-update} and
$R^{(j)}=e^{i\chi_j}(P_j^{\rm c})^Te^{iG'_j}P_j^{\rm c}$,
the next residual is
\begin{equation}
\begin{aligned}
R^{(j+1)}&=U(W^{(j)})^{-1}\\
&=(P_j^{\rm c})^Te^{iG'_j}e^{-i\widetilde G_j}P_j^{\rm c}.
\end{aligned}
\label{eq:correction-residual-update}
\end{equation}
Unitary invariance and Duhamel's formula then imply
\[
\begin{aligned}
\norm{U-W^{(j)}}&=\norm{R^{(j+1)}-I}\\
&=\norm{e^{iG'_j}-e^{i\widetilde G_j}}\le\rho_j.
\end{aligned}
\]
Thus $\rho_j\le\epsilon/8$ guarantees
$\norm{U-W^{(j)}}\le\epsilon/8$.

If levels were used, $s_{l_*}=0$ gives
$\norm{G_{0,1}}\le a_{l_*}<nL/(c_s\sqrt D)$; otherwise
$\norm{G_{0,1}}\le2\pi$ and $\sqrt D\le\sqrt2\,nL$.
In both cases, the flattening bound gives $\alpha_1=O(n^2L)$.
By \Cref{eq:wordlength}, $\rho_j\le N^{-2}\alpha_j$, so
$\rho_1<1$ for sufficiently large $n$ because $L\le N$.

Whenever $\rho_j<1$, the eigenphase bound for
$\norm{R^{(j+1)}-I}\le\rho_j$ gives
$\norm{G_{0,j+1}}\le\pi\rho_j$; combining this with the
flattening and rounding bounds yields
\begin{equation}
\begin{split}
\alpha_{j+1}&\le32\pi^2C_A\,b\sqrt D\log(8N)\,2^{-B}\alpha_j\\
&\le N^{-2}\alpha_j.
\end{split}
\label{eq:correction-contraction}
\end{equation}
The last inequality uses \Cref{eq:wordlength}, and induction gives
$\alpha_j\le N^{-2(j-1)}\alpha_1$ and
$\rho_j\le N^{-2j}\alpha_1<1$ until the algorithm stops.
For $\alpha_1>0$, the stopping test is satisfied whenever
$j\ge1$ and $2j\log N\ge\log(8\alpha_1/\epsilon)$.
Together with $\alpha_1\le Cn^2N$ and $L=n+\log_2(1/\epsilon)$,
this gives $J=O(1+L/n)$; a zero traceless generator stops
the algorithm immediately.

\emph{Error analysis.} The ideal-product error
is at most $\epsilon/8$; allocate $\epsilon/8$ each to simulation
truncation, rotation synthesis, catalyst preparation and
unpreparation, and phase synthesis.

Let $N_{\rm sim}$ count the sampled level evolutions and correction
evolutions that are not skipped.
For $N_{\rm sim}>0$, applying \Cref{lem:hamsim} with
$\delta=\epsilon/(8N_{\rm sim})$ per call gives total full-output
error at most $\epsilon/8$, including ancilla leakage.
The bounds $l_*=O(n)$, $N_q^{(l)}=O(s_ln^2)$, and
$J=O(1+L/n)$ make $N_{\rm sim}$ polynomial in $N$, so
$\log(1/\delta)=O(L)$ whenever simulations occur.

If the QSVT circuits use $N_{\rm rot}>0$ variable-angle rotations,
synthesize each to operator-norm error $\epsilon/(8N_{\rm rot})$, giving total
error at most $\epsilon/8$ by gate telescoping.
The query bounds in \Cref{lem:hamsim} make $N_{\rm rot}$
polynomial in $N$, so each rotation costs $O(L)$ by
\Cref{eq:single-qubit-synthesis}.
Synthesize the accumulated phase
$\Phi=\sum_{l=0}^{l_*-1}\varphi_l+\sum_{j=1}^J\chi_j$
to error $\epsilon/8$ using \Cref{eq:single-qubit-synthesis}.
The changes of basis, address permutations, and zero-ancilla
reflections have exact Clifford+$T$ implementations.

With the exact catalyst $\ket\Psi$, queries to the block encodings
of $\widetilde K_k^{(l)}$ and $\widetilde G_j$ are exact and
preserve $\ket\Psi$.
Preparing $\ket\Psi$ to error $\delta_c=\epsilon/16$ and reversing
the preparation after all queries contributes at most
$2\delta_c=\epsilon/8$, by comparison with the exact-catalyst
circuit, independently of the number of queries.
Together with the ideal-product error, these contributions give
total full-output error at most $5\epsilon/8<\epsilon$ by the
composition rule following \Cref{def:implementation}.

\emph{Classical trials.} At a level requiring random choices,
the joint flattening and splitting test succeeds with probability
at least $3/4$, by a union bound for two failure parameters $1/8$.
Conditional on acceptance, the qDRIFT accuracy and weighted-time
tests succeed together with probability at least $1/2$
(\Cref{lem:constant-selection}); each correction's flattening
test succeeds with probability at least $3/4$.
Thus each test takes a constant expected number of fresh trials,
giving $O(n+L/n)$ trials over $O(n)$ levels and $O(1+L/n)$
corrections.
\end{proof}

\subsection{Resource bounds}
\label{sec:refinement-cost}

\begin{proof}[Proof of \Cref{thm:refined}: resources]
\emph{Levels.} At each sampled level $0\le l<l_*$,
\Cref{lem:hamsim} implements the $N_q^{(l)}$ evolutions
$e^{i\tau_{l,u}\widetilde K_{K_{l,u}}^{(l)}}$, each to error
$\delta=\epsilon/(8N_{\rm sim})$, using a total of
\[
Q_l=O\!\left(
\sum_{u=1}^{N_q^{(l)}}\alpha_{K_{l,u}}^{(l)}\tau_{l,u}
+N_q^{(l)}\log(1/\delta)\right)
\]
queries to the block encodings of $\widetilde K_k^{(l)}$.
The weighted-time test in \Cref{eq:level-acceptance},
$\log(1/\delta)=O(L)$, $N_q^{(l)}=O(s_ln^2)$, and
\Cref{eq:level-normalization-bound,eq:sl} give
\[
Q_l=O\Big(\sum_{k\in\mathcal I_l}\alpha_k^{(l)}+N_q^{(l)}L\Big)
=O(a_l\sqrt D\,n),
\]
while identity levels use no queries.
Since $s_l\ge1$ implies $s_l\ge c_s\sqrt Da_l/(2nL)$,
the per-query cost $O(N\sqrt{B/s_l}+b^2B)$ with $B=O(n)$
gives the level cost
\begin{equation}
O\big(n^2ND^{1/4}\sqrt{a_lL}+n^2a_l\sqrt D\,b^2\big).
\label{eq:refinement-level-cost}
\end{equation}
Summing over $0\le l<l_*$, with
$\sum_l\sqrt{a_l},\sum_la_l=O(1)$, gives
\[
O\big(n^2(ND^{1/4}\sqrt L+\sqrt D\,b^2)\big).
\]
Balancing these terms with $b=\Theta(N^{3/7}L^{2/7})$
gives $O(n^2N^{8/7}L^{3/7})$.

\emph{Corrections.} Simulating $e^{i\widetilde G_1}$ requires
$O(\alpha_1+L)=O(n^2L)$ queries to the block encoding of
$\widetilde G_1$. For $2\le j\le J$ with $G_{0,j}\ne0$,
\Cref{eq:correction-contraction} gives
$e\alpha_j/2\le Cn^2N^{3-2j}<1$ and
$\log(2/(e\alpha_j))=\Omega((j-1)n)$.
The short-time simulation bound in \Cref{lem:hamsim} therefore
gives $O(1+L/((j-1)n))$ queries to the block encoding of
$\widetilde G_j$, totaling $O(1+(L/n)\log L)$ over
$2\le j\le J=O(1+L/n)$; zero generators require no queries.
At $O(N\sqrt B+b^2B)$ per query, all corrections cost
$O(n^{5/2}NL+n^3b^2L)$.

For $L\le L_0$, the chosen block size gives
\[
b^2L\le4N^{6/7}L^{11/7}\le4N^{8/7}L^{3/7},
\]
using $L\le N^{1/4}$. For $L>L_0$, either $D=N$ and $b=1$, or
$b^2B\le N^2B/(n^4L^4)=O(N\sqrt n)$ by the definition of $L_0$.
Thus in both regimes the correction cost is
$O(n^{5/2}NL+n^3N^{8/7}L^{3/7})$.

\emph{Remaining resources.} QSVT uses $O(1)$ variable-angle
rotations per query to the block encoding of a rounded piece
$\widetilde K_k^{(l)}$ or correction generator $\widetilde G_j$,
each costing $O(L)$ by
\Cref{eq:single-qubit-synthesis}; their total cost is
\[
O(n\sqrt D\,L+n^2L^2)=O(n^3NL),
\]
using $L\le N$. The changes of basis cost
$O((n+L/n)\sqrt N)$, reflections $O(n)$ per query,
catalyst preparation and unpreparation $O(nL)$, and the accumulated
phase $O(L)$, all within the $T$-count in \Cref{eq:refinement-total-cost}.
Sequential queries reuse $O(N\sqrt B+b^2B)=O(nN)$ ancilla qubits,
which also suffice for the catalyst, logical registers, and sign oracle.
\end{proof}

\paragraph*{Comparison of precision dependence.}
For $L=N^{\gamma+o(1)}\le N$ with fixed $\gamma\ge0$,
\Cref{thm:refined} has a smaller leading power of $N$ than Tan's
bound~\cite{tan2025unitary} when $\gamma<1$ and that of
Ref.~\cite{yuan2026general} when $\gamma<2/3$, ignoring polynomial
factors in $n$.
Including these factors, $L=o(N^{2/3}/n^{16/3})$ suffices for
improvement over the latter bound, which is smaller among those
compared here for $N^{2/3}<L\le N$.
\Cref{thm:refined} covers $L\le N$, and Tan's bound also applies
for $L>N$.

For $N^{1/4}\le L\le N$, our bound is $O(n^3NL)$, with the
$\Ot(NL)$ term arising from the first correction's $O(n^2L)$ queries
to the block encoding of $\widetilde G_1$.
The gap to the counting scale $\Omega(N\sqrt L)$ is a factor
$\sqrt L$, up to polynomial factors in $n$; here
$L=\Theta(\log_2(1/\epsilon))$.

\subsection{Effective classical compilation}
\label{sec:effective-compilation}

The constructions underlying \Cref{thm:constant,thm:refined} treat
matrix logarithms, exponentials, norms, comparisons, rounding, and
sampling from exact finite distributions as oracle operations.
Their probability estimates bound the expected number of resampling
trials by a constant for the constant-precision construction and by
$O(n+L/n)$ in total for the refined construction.
These trial counts do not bound classical running time: the oracles
need not have unit cost on a rational-arithmetic RAM, and certified
entry approximations alone may not decide zero tests or rounding ties.

Finite circuit search avoids these decisions by using the strict
error margins in the existence proofs. It yields the following
computability guarantee, with no polynomial bound on classical
compilation time.

\begin{proposition}[Effective compilation]
\label{prop:effective-compiler}
Suppose that $n\ge1$, the entries of $U\in U(N)$ are supplied by
algorithms returning certified rational approximations to any requested
accuracy, and $\epsilon>0$ is rational or supplied by such an algorithm.
Given a choice of \Cref{thm:constant,thm:refined} and a promise that
$(n,\epsilon)$ lies in its parameter range, a deterministic algorithm
terminates and outputs a finite Clifford+$T$ circuit with full-output
implementation error strictly less than $\epsilon$ and the selected
theorem's asymptotic $T$- and ancilla bounds.
The compiler relies on the parameter-range promise and requires only
computability, with no bound on the time to approximate the entries.
\end{proposition}

\begin{proof}
We use $\mathsf H,S,\mathrm{CNOT},T$ on specified wires, retaining
their matrix phases and replacing $T^\dagger$ by $S^\dagger T$
without changing the $T$-count.
All auxiliary registers start in $\ket0$, as in
\Cref{def:implementation}.

\emph{Strict witnesses at the original precision.}
The proof in \Cref{app:constant} allocates errors
$\epsilon/4,\epsilon/4,\epsilon/8,\epsilon/8,\epsilon/8$,
whose sum is $7\epsilon/8$; its early exits have no larger error.
The proof in \Cref{sec:refinement-errors} allocates five contributions
of at most $\epsilon/8$, whose sum is $5\epsilon/8$.
Thus either proof establishes the existence of a circuit with full-output
error at most $7\epsilon/8$ and the stated quantum resources, even if
its exact-data choices are not computable.
These witnesses use the original $\epsilon$, preserving both the lower
endpoint in \Cref{thm:constant} and the endpoint $L=N$ in
\Cref{thm:refined}.

\emph{Computable resource scales.}
With $L=n+\log_2(1/\epsilon)$, define the comparison scales
$(F_T,F_A)$ for \Cref{thm:constant,thm:refined} by
\[
\begin{array}{ll}
F_T=n^2L N^{8/7}\epsilon^{-6/7},
&F_A=nL N^{6/7}\epsilon^{-8/7},\\[2pt]
F_T=n^3\bigl(N^{8/7}L^{3/7}+NL\bigr),
&F_A=nN,
\end{array}
\]
respectively. These scales are at least one on the stated parameter
ranges, and the witnesses satisfy $T\le c_TF_T$ and $m\le c_AF_A$
for fixed absolute constants $c_T,c_A$.
The compiler uses rationally computed integer budgets and increases
them to accommodate these constants.

First choose a positive rational $r$ such that
\[
\frac{31}{32}\epsilon<r<\epsilon.
\]
For computable real $\epsilon$, refine certified rational bounds
$0<l\le\epsilon\le u$ until
$(31/32)u<l$, and choose $r$ strictly between $(31/32)u$ and $l$.
Let $p$ be the least nonnegative integer such that $r2^p\ge1$,
found by rational comparisons, and put $\overline L=n+p$.
Since $\epsilon/2<r<\epsilon\le1$, minimality gives
\[
\log_2(1/\epsilon)\le p<\log_2(1/\epsilon)+2,
\qquad L\le\overline L\le3L.
\]
The search uses $\overline L$ only to set resource budgets; it may
exceed $N$ at $L=N$, since the witness still uses the original
$\epsilon$ and $L$.

For nonnegative rational $x$, let $R_7(x)$ be the least nonnegative
integer $j$ with $x\le j^7$, found by rational comparisons; then
\[
x^{1/7}\le R_7(x)<x^{1/7}+1.
\]
For \Cref{thm:constant}, define the integer budgets by
\[
\begin{aligned}
 t&=n^2\overline L\,R_7(N^8/r^6),\\
 a&=n\overline L\,R_7(N^6/r^8).
\end{aligned}
\]
For \Cref{thm:refined}, instead use
\[
\begin{aligned}
 t&=n^3\bigl(R_7(N^8\overline L^3)+N\overline L\bigr),\\
 a&=nN.
\end{aligned}
\]
All arguments of $R_7$ are rational. Using
$r>\epsilon/2$, $L\le\overline L\le3L$, and the fact that the relevant
root factors are at least one gives
\[
F_T\le t\le9F_T,\qquad F_A\le a\le15F_A
\quad\hbox{for \Cref{thm:constant}},
\]
and
\[
F_T\le t\le3F_T,\qquad a=F_A
\quad\hbox{for \Cref{thm:refined}}.
\]
Thus the rational-arithmetic budgets satisfy
$F_T\le t\le15F_T$ and $F_A\le a\le15F_A$ in either regime.

\emph{A finite family at each resource budget.}
For each fixed number $q$ of wires, the group $\mathcal C_q$
generated by the elementary Clifford gates $\mathsf H,S,\mathrm{CNOT}$ is
finite as a group of matrices with their scalar phases retained.
Conjugation permutes the finite set of signed Hermitian Pauli
matrices. A matrix in the kernel commutes with all single-wire
$X$ and $Z$ matrices and is therefore $zI$.
Each elementary Clifford generator has determinant in
$\{1,-1,i,-i\}$, so this scalar kernel satisfies
$z^{4\cdot2^q}=1$; the finite image and kernel imply that
$\mathcal C_q$ is finite.
Its entries lie in $\mathbb Q(e^{i\pi/4})$, where exact arithmetic and
equality are decidable, so breadth-first closure under the generators
terminates with a finite gate word for each element of $\mathcal C_q$.

For integer budgets $A,B$, set $q=n+A$ and enumerate the finite set
of matrix products
\[
C_h T_{i_h} C_{h-1}\cdots T_{i_1}C_0,
\]
where $0\le h\le B$, $C_j\in\mathcal C_q$, and $1\le i_j\le q$.
Grouping consecutive Clifford gates puts every circuit with at most
$B$ occurrences of $T$ in this finite family, regardless of its
Clifford-word length. Circuits with fewer than $A$ ancillas are included
by adding unused clean ancillas, preserving their $T$-count and
full-output error.

\emph{A terminating test for every candidate.}
Write $d=2^{n+A}$, let
$J_A=\ket{0^A}\otimes I_N:\mathbb C^N\longrightarrow\mathbb C^d$,
and for a candidate matrix $C$ put
\[
E_C=CJ_A-J_AU,\qquad e_C=\norm{E_C}.
\]
This full-output error map includes leakage out of the clean-ancilla
subspace. For any positive rational $\eta$, compute a rational upper
bound $u_C$ satisfying
\[
e_C\le u_C\le e_C+\eta
\]
in finite time as follows. Compute a rational complex matrix $Q$
approximating every entry of $E_C$ to within $\eta/(4dN)$. Then
$\norm{E_C-Q}\le\norm{E_C-Q}_{\mathrm F}\le\eta/4$.
Starting with $j=0$, test the rational numbers $s_j=j\eta/2$
until $s_j^2I_N-Q^\dagger Q$ is positive semidefinite.
Each test checks that all principal minors of this rational Hermitian
matrix are nonnegative. The loop terminates at a number $s$ satisfying
$\norm Q\le s<\norm Q+\eta/2$, except that $s=\norm Q=0$
is also possible. Hence $u_C=s+\eta/4$ satisfies
the stated error certificate using only rational comparisons.

Use this test with $\eta=r/64$, and accept a candidate precisely when
the resulting rational upper bound satisfies $u_C<r$.
Every candidate receives a finite accept-or-reject test; acceptance
certifies $e_C<r<\epsilon$, and every witness with
$e_C\le7\epsilon/8$ is accepted because
\[
u_C\le\frac78\epsilon+\frac r{64}
 <\frac{57}{64}\epsilon
 <\frac{31}{32}\epsilon<r.
\]

\emph{Search, termination, and resources.}
For $k=1,2,\ldots$, set $A=ka$ and $B=kt$, enumerate
the finite normal-form family above, and apply the finite test to every
candidate, returning the first one accepted and otherwise exhausting
the stage before proceeding.
Fix any integer $K\ge\max\{1,c_T,c_A\}$.
Each stage terminates, and stage $K$ contains a witness with unused
clean ancillas added if needed, so the search returns a circuit at
some $k\le K$ with
\[
T\le Kt=O(F_T),\qquad m\le Ka=O(F_A).
\]
The constants are uniform in $U,n,\epsilon$, although the compiler
need not know $K$; its output uses finitely many Clifford gates, whose
number is unrestricted by the theorems.
The full matrix norm retains scalar phase, and the ancilla bound
includes the phase-synthesis ancilla from the existence proofs.
The compiler uses certified input approximations, finite circuit
enumeration, and rational error certificates.
\end{proof}

\section{Lower bounds in the sequential block-SELECT model}
\label{app:optimality}

The proof of \Cref{thm:optimality} uses three ingredients.
\Cref{lem:sparsification} compresses each generator by sampling blocks.
\Cref{lem:sparsification-composition} controls the resulting product error.
Together they supply the compression used in
\Cref{lem:compression-information}, which bounds the description length
required by a typical target. In \Cref{app:allocation}, we optimize
the allocation to turn this information bound into a cost lower bound.

\subsection{Unbiased block sparsification}
\label{app:sparsification}

\begin{lemma}[Unbiased block sparsification]
\label{lem:sparsification}
Let $K=H_k\ne0$ be a Hermitian generator with the block parameters
of \Cref{def:design}, and put $w_k=\sum_{I,J}\norm{K_{IJ}}$.
For $\epsilon_k>0$, $\delta\in(0,1)$, and any integer
\begin{equation}
t\ \ge\ t_k(\epsilon_k,\delta):=3D\Big(1+\frac{\beta_k^2}{\epsilon_k^2}\Big)\log\frac{2N}\delta ,
\label{eq:sparsify}
\end{equation}
there is a Hermitian estimator $\widehat K$ formed from $t$
independent sampled adjoint block pairs, each of norm at most $w_k/t$.
It satisfies $\E\widehat K=K$, $\norm{\widehat K}\le w_k\le D\beta_k$,
and $\Pr(\norm{\widehat K-K}\le\epsilon_k)\ge1-\delta$.
\end{lemma}

\begin{proof}
Rowwise Cauchy--Schwarz and the block Schur test give
\begin{equation}
\max_I\sum_J\norm{K_{IJ}}\le\beta_k,\qquad \norm K\le\beta_k,\qquad w_k\le D\beta_k .
\label{eq:rowsum}
\end{equation}
Sample $t$ unordered block pairs independently. Pair $\{I,J\}$ with $I<J$
has probability $2\norm{K_{IJ}}/w_k$, and the diagonal pair $\{I,I\}$ has
probability $\norm{K_{II}}/w_k$. Here $E_{IJ}$ is the $D\times D$
matrix with a $1$ in position $(I,J)$ and zeros elsewhere.
A sampled pair contributes
\begin{equation}
\frac {w_k}{2t}\big(E_{IJ}\otimes V_{IJ}+E_{JI}\otimes V_{IJ}^\dagger\big)
\label{eq:sampled-block-pair}
\end{equation}
for $I<J$, or $\frac{w_k}tE_{II}\otimes V_{II}$ for a diagonal pair,
with $V_{IJ}=K_{IJ}/\norm{K_{IJ}}$ for a sampled nonzero block.
Their sum $\widehat K$ is Hermitian, unbiased, and has
$\norm{\widehat K}\le w_k$. Centered sample norms are at most $2w_k/t$.
The $I$-th block of a sample's second moment is at most
$(w_k/t^2)\sum_J\norm{K_{IJ}}$, giving total variance at most
$w_k\beta_k/t$. Matrix Bernstein~\cite{tropp2012user}, with
$w_k\le D\beta_k$ and $2\beta_k/\epsilon_k\le1+\beta_k^2/\epsilon_k^2$,
gives the claimed tail bound for \Cref{eq:sparsify}.
\end{proof}

\subsection{Composition of sparsification errors}
\label{app:composition}

For a design, an \emph{allocation} consists of a subset
$\mathcal O\subseteq\{k:H_k\ne0\}$ and a parameter $a>0$.
Each selected generator receives the positive error tolerance
\begin{equation}
\epsilon_k^2=\nu\beta_k/a,\qquad \nu:=\frac{\epsilon^2}{72\log(4N)} .
\label{eq:sparsification-error-allocation}
\end{equation}
The allocation is \emph{valid} if $\sum_{k\in\mathcal O}\beta_k\le a$, so that
$\sum_{k\in\mathcal O}\epsilon_k^2\le\nu$.

\begin{lemma}[Composition of sparsification errors]
\label{lem:sparsification-composition}
Fix a design with product $X=W_se^{iH_s}\cdots W_1e^{iH_1}W_0$,
an error $0<\epsilon\le1/2$, and a valid allocation $(\mathcal O,a)$.
For $k\in\mathcal O$, draw independent estimators from
\Cref{lem:sparsification} with $\delta_k=\min\{1/2,\nu/(4Da)\}$
and $\lceil t_k(\epsilon_k,\delta_k)\rceil$ samples; denote this
integer count by $t_k$. Condition each estimator separately on
$G_k=\{\norm{\widehat H_k-H_k}\le\epsilon_k\}$.
Keeping $H_k$ for $k\notin\mathcal O$, let $\widehat X$ be the resulting
product. Under these conditioned laws,
\[
\Pr(\norm{\widehat X-X}\le2\epsilon)\ge1/2.
\]
\end{lemma}

\begin{proof}
Write $\Delta_k=\widehat H_k-H_k$ for $k\in\mathcal O$; expectations use the
independent conditioned laws. Before conditioning, each error has
mean zero and norm at most $2D\beta_k$ by \Cref{lem:sparsification}.
The acceptance probability is at least $1-\delta_k\ge1/2$, so
\[
\norm{\E\Delta_k}
\le\frac{2\delta_kD\beta_k}{1-\delta_k}
\le4\delta_kD\beta_k\le\epsilon_k^2.
\]

The second-order Duhamel expansion has a linear term bounded by
$\norm{\E\Delta_k}$ and a remainder bounded by
$\frac12\E\norm{\Delta_k}^2$, so
\begin{equation}
\begin{split}
\norm{\E e^{i(H_k+\Delta_k)}-e^{iH_k}}
&\le\norm{\E\Delta_k}+\tfrac12\E\norm{\Delta_k}^2\\
&\le2\epsilon_k^2 .
\end{split}
\label{eq:sparsification-exponential-bias}
\end{equation}
The unitary factors make this estimate uniform in $H_k$.

List the sampled generator indices as $k_1,\ldots,k_{\abs{\mathcal O}}$
in product order, and let $X_j$ replace the first $j$ of them.
Write $X_j-X_{j-1}=Z_j+d_j$, where $d_j$ is its conditional mean
given previous samples. Unitary prefixes and suffixes preserve
the bounds $\norm{d_j}\le2\epsilon_{k_j}^2$ and
$\norm{X_j-X_{j-1}}\le\epsilon_{k_j}$, hence
$\norm{Z_j}\le3\epsilon_{k_j}$ since $\epsilon_{k_j}<1$.
Matrix Azuma~\cite{tropp2012user}, applied to the Hermitian dilations
of the martingale differences $Z_j$, gives with probability at least $1/2$
\[
\norm{\sum_jZ_j}\le\sqrt{72\sum_j\epsilon_{k_j}^2\log(4N)}\le\epsilon.
\]

The drift is at most $\sum_j\norm{d_j}\le2\nu$, so on the same event
$\norm{\widehat X-X}\le\epsilon+2\nu\le2\epsilon$.
\end{proof}

\subsection{A counting argument}
\label{app:counting}

\begin{lemma}[Information bound for compressed designs]
\label{lem:compression-information}
Fix the interleaver decoding rule of \Cref{def:design}.
There are absolute constants $c,C>0$ such that, for sufficiently
large $n$ and $\epsilon\in[N^{-3},c]$, with probability at least
$1/2$ over a Haar-random $U\in U(N)$, the following holds
simultaneously for every power-of-two block size $b\le N$.
Every design at block size $b$ approximating $U$ to error $\epsilon$ with
$\mathrm{cost}(\mathcal D)\le N^3$, $s\log N\le N^2/C$, and
$B_W\le\frac14N^2\log_2(1/\epsilon)$ satisfies
\begin{equation}
\sum_{k\notin\mathcal O}n_k+\sum_{k\in\mathcal O}\min\{n_k,t_k\}\ \ge\ I_0:=\frac{cD^2}{\log N}
\label{eq:information}
\end{equation}
for every valid allocation $(\mathcal O,a)$ with $a\le N^3$.
Here $t_k$ is the integer sample count in
\Cref{lem:sparsification-composition}.
\end{lemma}

\begin{proof}
Fix a target, a design, and an allocation satisfying the hypotheses.
Among the selected generators, store those with $n_k\le t_k$ directly.
The remaining sampled set is a valid allocation with the same $a$,
so \Cref{lem:sparsification-composition} gives a product within
$2\epsilon$ of the design.
The hypotheses imply $s,\sum_k\beta_k\le N^3$,
$\log(1/\delta_k)=O(\log N)$, and $t_k\le N^{15}$ for large $n$.

All stored blocks have norm at most $N^4$: sampled blocks are bounded
by $w_k\le D\beta_k$ from \Cref{lem:sparsification}, and directly
stored blocks by $\beta_k\le N^3$.
Use an absolute operator-norm net of mesh $\epsilon''=\epsilon/N^{20}$
in this ball, with Hermitian diagonal blocks and exact adjoint partners.
Rounding preserves Hermiticity, so the resulting product remains unitary.
The net has at most $(3N^4/\epsilon'')^{2b^2}$ points, costing
$O(b^2\log N)$ bits per block since $\epsilon\ge N^{-3}$.

Fix one target-independent prefix-free decoder. Its header records
$\log_2b$, $s$, $\mathcal O$, block counts, and $\theta$ to precision
$\epsilon$, using $O(\log n+s\log N+\log(1/\epsilon))$ bits.
Append the $B_W$ interleaver bits and each block's position
($2\log_2D$ bits) and rounded value. Directly stored generators need
only their upper triangle; sampled blocks include their scale.
Block counts and values capture the dependence on the allocation parameter $a$.

There are at most $N^{15}$ stored blocks per generator and $N^3$ generators,
so rounding changes the product by at most
$2N^{18}\epsilon''\le\epsilon$. The code length satisfies
\begin{equation}
\begin{split}
\mathcal L&\le C b^2\log N\Big(\sum_{k\notin\mathcal O}n_k+\sum_{k\in\mathcal O}\min\{n_k,t_k\}\Big)\\
&\quad+B_W+O\big(\log n+s\log N+\log(1/\epsilon)\big) .
\end{split}
\label{eq:description-length}
\end{equation}
Design error, sampling, rounding, and phase error total
$\epsilon+2\epsilon+\epsilon+\epsilon=5\epsilon$.

\emph{Szarek's bound.} We next bound the description length required
to approximate a typical target. For an absolute $c_1>0$, the covering numbers of
$U(N)$ in operator norm are at least $(c_1/r)^{N^2}$~\cite{szarek1982nets}.
A maximal family of disjoint
radius-$r$ balls has centers forming a $2r$-net, so it contains at
least $(c_1/(2r))^{N^2}$ balls. By unitary invariance, these balls have
equal Haar measure; disjointness therefore bounds the measure of
each by $(C_1r)^{N^2}$ for an absolute $C_1\ge1$.

Fewer than $2^{\mathcal L+1}$ words of this fixed code have length
at most $\mathcal L$, and each decodes uniquely~\cite{cover2006elements}.
Their $5\epsilon$ neighborhoods have total measure at most
$2^{\mathcal L+1}(5C_1\epsilon)^{N^2}$. Thus at least half of Haar
targets require $N^2\log_2(1/(5C_1\epsilon))-2$ bits.
Because the header records the block size, subset, and counts,
this single Haar event covers all target-dependent choices of
designs, block sizes, and valid allocations.

Choose $c\le1/2$ so that
$\log_2(1/(5C_1\epsilon))\ge\frac23\log_2(1/\epsilon)$, and $C$
so that the header term in \Cref{eq:description-length} is at most
$\frac1{12}N^2\log_2(1/\epsilon)$ for large $n$.
Subtracting this term and the interleaver budget from the required
description length gives \Cref{eq:information}, after reducing $c$.
\end{proof}

\subsection{Optimizing the error allocation}
\label{app:allocation}

\begin{proof}[Proof of \Cref{thm:optimality}]
Use the constant $C$ and the Haar event of probability at least $1/2$
from \Cref{lem:compression-information}.
Designs with $\mathrm{cost}(\mathcal D)>N^3$ or $s\log N>N^2/C$
have cost $\Omega(N^2/\log N)$, which implies \Cref{eq:lowerbound}
since its target scales sum to at most $2N^2$.
For the remaining designs, the lemma gives \Cref{eq:information}
for every valid allocation. We choose an allocation for which the
selected generators have sampled descriptions with fewer blocks than
direct storage. This connects the information bound to the model cost.

Let $C'=4\log(8DN^4/\nu)=O(\log N)$.
By \Cref{lem:sparsification,lem:sparsification-composition}, the
rounded sample counts satisfy $t_k\le C'D(1+a\beta_k/\nu)$.
Call a generator \emph{small} if $n_k\le2C'D$ and
\emph{big} otherwise. For $a>0$ let
\begin{equation}
\mathcal O(a)=\{k\ \text{big}:\ n_k>2C'D\beta_ka/\nu\},
\label{eq:optimized-sampling-set}
\end{equation}
and put $A(a)=\sum_{k\in\mathcal O(a)}\beta_k$.
Selected generators satisfy $t_k<n_k/2+n_k/2=n_k$.
If all generators are small, use $\mathcal O=\emptyset$ and any $0<a\le N^3$;
\Cref{eq:information} gives case~(a) below.
Otherwise the non-increasing, right-continuous function $A$ crosses
the line $a$ at $a_*=\inf\{a>0:A(a)\le a\}\in(0,N^3]$.
Thus $A(a_*)\le a_*$ gives a valid allocation, while
$A(a_*^-)\ge a_*$ says that the big generators with
$n_k\ge2C'D\beta_ka_*/\nu$ have total weight at least $a_*$.
We call these the closed-threshold group.

Inserting the allocation $(\mathcal O(a_*),a_*)$ into~\Cref{eq:information},
\begin{multline}
\sum_{\rm small}n_k+\#\{\text{big}\}\cdot C'D\\
+\sum_{\text{big}\setminus\mathcal O(a_*)}n_k+\frac{C'Da_*A(a_*)}\nu\ \ge\ I_0 .
\label{eq:information-four-terms}
\end{multline}
The second term bounds the additive sampling allowance using the
number of big generators. One of the four terms must be at least $I_0/4$.
Split the cost into
\begin{equation}
\begin{aligned}
\mathrm{cost}_{\rm look}&=\sum_k\max\{1,\beta_k\}\,b\sqrt{n_k},\\
\mathrm{cost}_{\rm proc}&=\sum_k\max\{1,\beta_k\}\,b^2,
\end{aligned}
\label{eq:cost-parts}
\end{equation}
so $\mathrm{cost}=\mathrm{cost}_{\rm look}+\mathrm{cost}_{\rm proc}$.
The target scales are
\begin{equation}
\Gamma_{\rm look}:=\sqrt\epsilon\,ND^{1/4},\qquad
\Gamma_{\rm proc}:=\epsilon\sqrt D\,b^2 .
\label{eq:target-scales}
\end{equation}
We prove both bounds for whichever term of
\Cref{eq:information-four-terms} is large.

\emph{Case (a): small generators.} If $\sum_{\rm small}n_k\ge I_0/4$,
there are at least $I_0/(8C'D)=\Omt(D)$ such generators.
Their block-processing cost is $\Omt(Db^2)$, and
\[
\mathrm{cost}_{\rm look}\ge b\sum_{\rm small}n_k/\sqrt{2C'D}
=\Omt(bD^{3/2})\ge\Omt(\Gamma_{\rm look}).
\]

\emph{Case (b): many big generators.} At least $I_0/(4C'D)$ big generators,
each with $n_k>2C'D$, give the same two bounds.
Since $N=bD$, $D\ge1$, and $\epsilon\le1$, these dominate
$\Gamma_{\rm look}$ and $\Gamma_{\rm proc}$.

\emph{Case (c): directly stored big generators.} Suppose
$\sum_{\text{big}\setminus\mathcal O(a_*)}n_k\ge I_0/4$. These generators have $\beta_k\ge n_k\nu/(2C'Da_*)$.
Their lookup cost is at least $bI_0/4$ times
\begin{equation}
\min_{x>0}\frac{\max\{1,x\nu/(2C'Da_*)\}}{\sqrt x}
=\sqrt{\frac \nu{2C'Da_*}},
\label{eq:lookup-cost-minimum}
\end{equation}
and their block-processing cost is at least $b^2\nu I_0/(8C'Da_*)$.

The closed-threshold group has weight at least $a_*$ and
$\sqrt{n_k}\ge\sqrt{2C'Da_*/\nu}\sqrt{\beta_k}$.
Since $\max\{1,\beta_k\}\sqrt{\beta_k}\ge\beta_k$, it gives lookup
cost at least $b\sqrt{2C'D/\nu}\,a_*^{3/2}$ and block-processing cost
at least $b^2a_*$.
The groups may overlap at the threshold, so combine their bounds
by taking the maximum. Minimizing each maximum over $a_*>0$ gives
\begin{align*}
\mathrm{cost}_{\rm look}
&\ge c\,b I_0^{3/4}\left(\frac{\nu}{C'D}\right)^{1/4},\\
\mathrm{cost}_{\rm proc}
&\ge c\,b^2\sqrt{\frac{\nu I_0}{C'D}}.
\end{align*}
Substituting $I_0$ and $\nu$ gives
$\Omt(bD^{5/4}\nu^{1/4})=\Omt(\Gamma_{\rm look})$ and
$\Omt(b^2\sqrt{\nu D})=\Omt(\Gamma_{\rm proc})$, respectively.

\emph{Case (d): sampled big generators.} Suppose $C'Da_*A(a_*)/\nu\ge I_0/4$.
Since $A(a_*)\le a_*$, we get
$a_*\ge\sqrt{\nu I_0/(4C'D)}=\Omt(\epsilon\sqrt D)$.
Use the same closed-threshold group as in case (c). Its block-processing
and lookup bounds now give directly
$\mathrm{cost}_{\rm proc}\ge b^2a_*=\Omt(\Gamma_{\rm proc})$ and
$\mathrm{cost}_{\rm look}\ge b\sqrt{2C'D/\nu}\,a_*^{3/2}=\Omt(bD^{5/4}\nu^{1/4})=\Omt(\Gamma_{\rm look})$.

In every case
$\mathrm{cost}=\Omt(\Gamma_{\rm look}+\Gamma_{\rm proc})$, where all
hidden factors are powers of $\log N$. This proves~\Cref{eq:lowerbound}.

\emph{Minimizing over $b$.} Substituting $D=N/b$, the sum
$\sqrt\epsilon N^{5/4}b^{-1/4}+\epsilon N^{1/2}b^{3/2}$
has unrestricted minimum $\Theta(\epsilon^{4/7}N^{8/7})$ at
$b=\Theta(\epsilon^{-2/7}N^{3/7})$. The same lower bound holds
when $b$ is restricted to the allowed powers of two.
\end{proof}

\subsection{Mixed block sizes and larger interleaver budgets}
\label{app:mixed}

We extend \Cref{thm:optimality} to mixed block sizes and larger
interleaver budgets. A \emph{mixed-block design} is a sequence
$\mathcal D=(W_0,H_1,W_1,\dots,H_s,W_s)$ as in
\Cref{def:design}, except that each generator $H_k$ carries its
own power-of-two block size $b_k\le N$, which may be chosen after
seeing $U$. The parameters $n_k$, $m_k$, $c_k$, and $\beta_k$ of generator
$k$ are computed in the partition of \Cref{sec:prelim} into blocks
of size $b_k$, and
\begin{equation}
\mathrm{cost}(\mathcal D)=\sum_{k=1}^s\max\{1,\beta_k\}\big(b_k\sqrt{n_k}+b_k^2\big).
\label{eq:mixedcost}
\end{equation}
A design at block size $b$ is the special case in which $b_k=b$ for all $k$.

\begin{theorem}[Mixed block sizes and larger interleaver budgets]
\label{thm:mixed}
There are absolute constants $c,c_0,C>0$ with the following property.
Let $n$ be large enough, $\gamma\in(0,3/4]$, and
$\epsilon\in[N^{-3},c_0^{1/\gamma}]$. Then, with probability at least
$1/2$ over a Haar-random $U\in U(N)$, every mixed-block design with
$B_W\le(1-\gamma)N^2\log_2(1/\epsilon)$ that approximates $U$ to error
$\epsilon$ satisfies
\begin{equation}
\mathrm{cost}(\mathcal D)\ \ge\ \frac{c\,\gamma}{n^{C}}\,\epsilon^{4/7}N^{8/7}.
\label{eq:mixedbound}
\end{equation}
\end{theorem}

For $\gamma=3/4$ the interleaver budget is that of
\Cref{thm:optimality}. The condition $\epsilon\le c_0^{1/\gamma}$ reads
$\gamma\ge\log(1/c_0)/\log(1/\epsilon)$. Within the model, a
worst-case cost guarantee smaller than $\epsilon^{4/7}N^{8/7}$ by
a fixed power of $N$ therefore requires an interleaver budget
allowing all but a fraction $O(1/\log(1/\epsilon))$ of the
$N^2\log_2(1/\epsilon)$ description bits. At $\epsilon=N^{-3}$,
this fraction is $O(1/n)$.

\begin{proof}
Group generators by $G_r=\{k:b_k=2^r\}$, with $b^{(r)}=2^r$ and $D_r=N/b^{(r)}$,
keeping their original product order. As before assume
$\mathrm{cost}\le N^3$, so $s,\sum_k\beta_k\le N^3$.
Here $\gamma$ is the interleaver-budget fraction.

\emph{Joint compression.} Give each class an allocation
$(\mathcal O_r,a_r)$ satisfying $a_r>0$, $\mathcal O_r\subseteq\{k\in G_r:H_k\ne0\}$, and
$\sum_{k\in\mathcal O_r}\beta_k\le a_r\le N^3$. Assign it the share
$\nu_r=\nu/(n+1)$ of the total squared-error tolerance $\nu$.
Zero generators contribute identity factors and are kept exact.
Every sampled generator has $\beta_k>0$ and $\epsilon_k>0$. Apply
\Cref{lem:sparsification} using
$\epsilon_k^2=\nu_r\beta_k/a_r$ and
$\delta_k=\min\{1/2,\nu_r/(4D_ra_r)\}$.
The conditioning estimate in the proof of
\Cref{lem:sparsification-composition} gives
$\norm{\E\Delta_k}\le\epsilon_k^2$. Its martingale argument applies
across all classes because $\sum_k\epsilon_k^2\le\sum_r\nu_r=\nu$,
giving a jointly sparsified product within $2\epsilon$ of the original.

The fixed code in the proof of \Cref{lem:compression-information}
adds $O(s\log n)$ header
bits for the $b_k$. Each class-$r$ block costs $O((b^{(r)})^2\log N)$
bits including its position. The counts and values capture the
class's allocation choices.
Thus the same counting event covers every mixed design and allocation.

Put $c_0=(5C_1)^{-2}$. Then
$\epsilon\le c_0^{1/\gamma}$ gives
\[
\log_2(1/(5C_1\epsilon))\ge(1-\gamma/2)\log_2(1/\epsilon).
\]
At least half of Haar targets therefore require
$(1-\gamma/2)N^2\log_2(1/\epsilon)-2$ bits.
Designs with $s\log N>\gamma N^2/C$ satisfy \Cref{eq:mixedbound}
because $\mathrm{cost}\ge s$. We therefore assume $s\log N\le\gamma N^2/C$.
Subtracting the interleaver and header budgets leaves at least
$\frac\gamma4N^2\log_2(1/\epsilon)$ bits for blocks, hence
\begin{equation}
\sum_{r=0}^n(b^{(r)})^2q_r\ \ge\ \frac{c\,\gamma N^2}{\log N},
\label{eq:mixed-information}
\end{equation}
where
$q_r=\sum_{k\in G_r\setminus\mathcal O_r}n_k+\sum_{k\in\mathcal O_r}\min\{n_k,t_k\}$
is the number of stored blocks of class $r$.

\emph{Allocation.} Use the threshold allocation of \Cref{app:allocation}
in each class, with $D_r$, $\nu_r$, and
$C'_r=4\log(8D_rN^4/\nu_r)$; classes of small generators use
$\mathcal O_r=\emptyset$ and any $0<a_r\le N^3$.
By \Cref{eq:mixed-information} and $(b^{(r)})^2D_r^2=N^2$, some class has
\begin{equation}
q_r\ \ge\ I_r:=\frac{c\,\gamma D_r^2}{(n+1)\log N}.
\end{equation}
Apply the four deterministic cost cases to this class.
Cases~(a,b) are linear in $I_r$; in cases~(c,d), lookup scales
as $I_r^{3/4}\nu_r^{1/4}$ and processing as $(I_r\nu_r)^{1/2}$.
Factors of $n+1$ are polynomial in $n$, while
$\gamma^{3/4},\sqrt\gamma\ge\gamma$.
Thus this class, and hence the full design, costs at least
$(c\gamma/n^C)(\sqrt\epsilon\,ND_r^{1/4}+\epsilon\sqrt{D_r}\,(b^{(r)})^2)$.
Minimizing over $b^{(r)}$ gives \Cref{eq:mixedbound}.
\end{proof}

Treating classes separately charges the other exponentials as
interleavers. Coarsening by a factor $q$ can raise the block-processing
charge by $q^2$, while refining can increase $m_k$ and $\beta_k$
polynomially. Joint counting preserves the original interleaver budget
and each generator's block size.

\bibliography{references}

\end{document}